\documentclass[10pt,oneside,reqno]{amsart}

\usepackage{cancel,bm, amssymb}
\usepackage{color}

\usepackage{color}
\usepackage{graphicx,framed}
\usepackage[colorlinks=true, pdfstartview=FitV, linkcolor=blue, citecolor=blue, urlcolor=blue]{hyperref}
\definecolor{shadecolor}{rgb}{0.95, 0.95, 0.86}

\numberwithin{equation}{section}

\renewcommand{\d}{{\mathrm d}}

\newtheorem{theorem}{Theorem}[section]
\newtheorem{example}[theorem]{Example}
\newtheorem{exercise}[theorem]{Exercise}

\newtheorem{lemma}[theorem]{Lemma}
\newtheorem{remark}[theorem]{Remark}
\newtheorem{problem}[theorem]{Riemann-Hilbert Problem}

\newtheorem{proposition}[theorem]{Proposition} 
\newtheorem{corollary}[theorem]{Corollary} 
\newtheorem{definition}[theorem]{Definition}
\def\le{\left}
\def\ri{\right}
\def\ds{\displaystyle}

\def\bth{\begin{theorem}}
\def\et{\end{theorem}}
\def\bc{\begin{corollary}}
\def\ec{\end{corollary}}
\def\bx{\begin{example}\small}
\def\ex{\end{example}}
\def\bxr{\begin{exercise}\small}
\def\exr{\end{exercise}}
\def\bl{\begin{lemma}}
\def\el{\end{lemma}}
\def\bd{\begin{definition}}
\def\ed{\end{definition}}
\def\bp{\begin{proposition}}
\def\ep{\end{proposition}}

\def\br{\begin{remark}}
\def\er{\end{remark}}

\def\be{\begin{equation}}
\def\ee{\end{equation}}
\def\&{\hspace{-15pt}&}
\def\mod{\, \mathrm{mod}\,\,}
\def\bea{\begin{eqnarray}}
\def\eea{\end{eqnarray}}
\def\beas{\begin{eqnarray*}}
\def\eeas{\end{eqnarray*}}

\def \pa{\partial}
\def\C{{\mathbb C}}

\def\R{{\mathbb R}}
\def\N{{\mathbb N}}
\def\wh{\widehat}

\def\Z{{\mathbb Z}}
\def\u{\mathfrak u}

\def\1{{\bf 1}}

\def\z{\zeta}

\def\eqref#1{(\ref{#1})}

\begin{document}


\title[Zeros of large degree Vorob'ev-Yablonski polynomials via a Hankel determinant identity]{Zeros of large degree Vorob'ev-Yablonski polynomials via a Hankel determinant identity}

\author{Marco Bertola}
\address{Centre de recherches math\'ematiques,
Universit\'e de Montr\'eal, C.~P.~6128, succ. centre ville, Montr\'eal,
Qu\'ebec, Canada H3C 3J7 and,
Department of Mathematics and
Statistics, Concordia University, 1455 de Maisonneuve W., Montr\'eal, Qu\'ebec,
Canada H3G 1M8}
\email{bertola@mathstat.concordia.ca}

\author{Thomas Bothner}
\address{Centre de recherches math\'ematiques,
Universit\'e de Montr\'eal, C.~P.~6128, succ. centre ville, Montr\'eal,
Qu\'ebec, Canada H3C 3J7 and,
Department of Mathematics and
Statistics, Concordia University, 1455 de Maisonneuve W., Montr\'eal, Qu\'ebec,
Canada H3G 1M8}
\email{bothner@crm.umontreal.ca}

\keywords{Vorob'ev-Yablonski polynomials, Hankel determinant representation, 
asymptotic behavior of (pseudo) orthogonal polynomials, Riemann-Hilbert problem, Deift-Zhou nonlinear steepest descent method.}

\subjclass[2000]{Primary 33E17; Secondary 34E05, 34M50}

\thanks{The first author is supported
in part by the Natural Sciences and Engineering Research Council of Canada. The second author acknowledges support by Concordia University through a postdoctoral fellow
top-up award.}

\date{\today}

\begin{abstract}
In the present paper we derive a new Hankel determinant representation for the square of the Vorob'ev-Yablonski polynomial $\mathcal{Q}_n(x),x\in\mathbb{C}$. These polynomials are 
the major ingredients in the construction of rational solutions to the second Painlev\'e equation $u_{xx}=xu+2u^3+\alpha$. As an application of the new identity, we
study the zero distribution of $\mathcal{Q}_n(x)$ as $n\rightarrow\infty$ by asymptotically analyzing a certain collection of (pseudo) orthogonal polynomials connected
to the aforementioned Hankel determinant. Our approach reproduces recently obtained results in the same context by Buckingham and Miller \cite{BM}, which used the Jimbo-Miwa
Lax representation of PII equation and the asymptotical analysis thereof.
\end{abstract}
\maketitle
\section{Introduction and statement of results}\label{math-intro}

Rational solutions of the second Painlev\'e\ equation 
\be\label{PII}
  u_{xx}=xu+2u^3+\alpha,\hspace{0.5cm}\alpha\in\mathbb{C},
\ee
were introduced in \cite{V,Y}  in terms of a certain sequence of monic polynomials $\{\mathcal{Q}_n(x)\}_{n\geq 0}$, henceforth generally named Vorob'ev-Yablonski polynomials. These polynomials
are defined via the differential-difference equation
\be
  \mathcal{Q}_{n+1}(x)\mathcal{Q}_{n-1}(x) = x\mathcal{Q}_n^2(x)-4\Big[\mathcal{Q}_n''(x)\mathcal{Q}_n(x)-\big(\mathcal{Q}_n'(x)\big)^2\Big],\hspace{0.5cm}n\geq 1,
  \ x\in\mathbb{C}\nonumber
\ee
with $\mathcal{Q}_0(x)=1,\mathcal{Q}_1(x)=x$.
It was found that rational solutions of \eqref{PII} exist if and only if $\alpha=n\in\mathbb{Z}$. For each value $n\geq 1$  they are uniquely given by
\begin{equation}\label{VoYa}
  u(x)\equiv u(x;n) = \frac{\d}{\d x}\left\{\ln\left[\frac{\mathcal{Q}_{n-1}(x)}{\mathcal{Q}_n(x)}\right]\right\}\ ,\ \ \ 
  u(x;0) = 0,\qquad u(x;-n) = -u(x;n). 
\end{equation}
The Vorob'ev-Yablonski polynomial $\mathcal{Q}_n(x)$ for $n\geq 0$ is a monic polynomials of degree $\frac{n}{2}(n+1)$ with integer coefficients.
In the literature it is known \cite{KO} that $\mathcal{Q}_n(x)$ admits two determinantal representations; our first result will be a third representation.\smallskip

Of the pre-existing formul\ae\ we first state a formula of Jacobi-Trudi type: let $\{q_k(x)\}_{k\geq 0}$ be the polynomials 
defined by the generating function
\be\label{gen1}
  F_1(t;x) = \exp\left[-\frac{4t^3}{3}+tx\right] = \sum_{k=0}^{\infty}q_k(x)t^k
\ee
and set in addition $q_k(x)\equiv 0$ for $k<0$. Then
\be\label{JT}
  \mathcal{Q}_n(x) = \prod_{k=1}^n(2k+1)^{n-k}\,\det\Big[q_{n-2\ell+j}(x)\Big]_{\ell,j=0}^{n-1},\hspace{0.5cm}n\geq 1.
\ee
Secondly one can compute $\mathcal{Q}_n(x)$ from a Hankel determinant: let $\{p_k(x)\}_{k\geq 0}$ be the polynomials defined recursively via
\be
  p_0(x)=x,\hspace{0.5cm} p_1(x)=1,\hspace{0.5cm}p_{k+1}(x) = p_k'(x) +\sum_{m=0}^{k-1}p_m(x)p_{k-1-m}(x),
\ee
in particular
\begin{equation*}
  p_2(x)=x^2,\hspace{0.5cm}p_3(x)=4x,\hspace{0.5cm}p_4(x)=2x^3+5,\hspace{0.5cm}p_5(x)=16x^2,\hspace{0.5cm}p_6(x)=5x\left(x^3+10\right).
\end{equation*}
Then
\be\label{H1}
  \mathcal{Q}_n(x) = \kappa^{-\frac{n}{2}(n+1)}\det\Big[p_{\ell+j-2}\big(\kappa x\big)\Big]_{\ell,j=1}^n,\hspace{0.25cm} n\geq 1;
  \hspace{0.5cm} \kappa=-2^{-\frac{2}{3}}.
\ee
Although this identity expresses $\mathcal{Q}_n(x)$ as an exact Hankel determinant, the polynomials $\{p_k(x)\}_{k\geq 0}$ cannot be derived from an elementary generating function
as it was the case for $\{q_k(x)\}_{k\geq 0}$ in \eqref{gen1}. 

Our first major result is a seemingly new Hankel determinant representation for the squares of $\mathcal{Q}_n(x)$,
which indeed results from an elementary generating function. Let $\{\mu_k(x)\}_{k\geq0}$ be the collection of polynomials defined by the generating
function
\be\label{gen2}
  F_2(t;x) = \exp\left[-\frac{t^3}{3}+tx\right] = \sum_{k=0}^{\infty}\mu_k(x)t^k.
\ee
These polynomials satisfy the three-term recurrence 
\begin{equation}\label{rec1}
  \mu_{k+3}(x) = \frac{x\mu_{k+2}(x)}{k+3}-\frac{\mu_k(x)}{k+3},\hspace{0.5cm}k\geq 0
\end{equation}
with $\mu_0(x)=1,\mu_1(x)=x$ and $\mu_2(x)=\frac{1}{2}x^2$. Moreover
\begin{equation*}
  \mu_3(x) = \frac{x^3-2}{3!},\hspace{0.5cm}\mu_4(x) = \frac{x(x^3-8)}{4!},\hspace{0.5cm}\mu_5(x)=\frac{x^2(x^3-20)}{5!},\hspace{0.5cm}\mu_6(x)=\frac{x^6-40x^3+40}{6!},
\end{equation*}
and in general
\begin{equation*}
 \mu_k(-\kappa x) = q_k(x)(-\kappa)^k,\hspace{0.5cm}k\geq 0.
\end{equation*}
The relation to the Vorob'ev-Yablonski polynomials is as follows
\bth\label{theo1} For any $n\geq 1$, we have
\be\label{the1}
  \mathcal{Q}_{n-1}^2(x) = (-)^{\lfloor\frac{n}{2}\rfloor}\frac{1}{2^{n-1}}\prod_{k=1}^{n-1}\left[\frac{(2k)!}{k!}\right]^2\,\det\Big[\mu_{\ell+j-2}(x)\Big]_{\ell,j=1}^{n}.
\ee
where $\lfloor y\rfloor$ denotes the floor function of a real number $y$.
\et
The proof of Theorem \ref{theo1} is found in Section \ref{sec1}. Theorem \ref{theo1} can be put to practical use in the analysis of the distributions of the zeros of $\mathcal Q_n(x)$ when $n\to \infty$. This very same asymptotic problem was very recently addressed in \cite{BM} where Buckingham and Miller have analyzed the large degree asymptotics of $\mathcal{Q}_n(x)$ in different regions of the complex $x$-plane. Their approach
uses a specific Lax representation of \eqref{PII} and corresponding Riemann--Hilbert problem, which is completely different than the one we derive here (Sec. \ref{sec2}),  and then they proceed to an  asymptotical resolution of the RHP as  $n\rightarrow\infty$.\smallskip

Indeed, a direct consequence of  Theorem \ref{theo1} is that we can frame the same analysis in the relatively familiar context of large-degree asymptotics of orthogonal polynomials with respect to a varying weight in the spirit of \cite{DKMVZ}; recall that 
\be
  \mu_k(x) = \frac{1}{k!}\frac{\d^k}{\d t^k} F_2(t;x)\Big|_{t=0} = -\frac{1}{2\pi i}\oint F_2(w;x)\frac{\d w}{w^{k+1}} = -
\oint\zeta^{k}\,\d\nu(\zeta;x)
\label{moments}
\ee
where the line integrals are taken along the unit circle $S^1=\{\zeta\in\mathbb{C}: |\zeta|=1\}$ in clockwise orientation and
\begin{equation}\label{meas}
  \d\nu(\zeta;x) = \frac{1}{2\pi i}e^{-\theta(\zeta;x)}\frac{\d\zeta}{\zeta},\hspace{0.5cm} \theta(\zeta;x) = \frac{1}{3\zeta^3}-\frac{x}{\zeta}.
\end{equation}
In this setting we now introduce (pseudo) orthogonal polynomials 
\begin{definition}
 The monic orthogonal polynomials $\{\psi_n(\zeta;x)\}_{n\geq 0}$ of exact degree $n$ are defined by the requirements
\begin{eqnarray}
  \oint\psi_n(\zeta;x)\zeta^m \d\nu(\zeta;x) &=& \begin{cases}
  h_n(x),& m=n\\
  \ \ \ 0,& m\leq n-1\\
  \end{cases}
  \label{o1}\\
  \psi_n(\zeta;x) &=& \zeta^n+\mathcal O\left(\zeta^{n-1}\right),\hspace{0.5cm}\zeta\rightarrow\infty.\label{o2}
\end{eqnarray}
Also here, the line integral is taken along the unit circle $S^1$ in clockwise orientation.
\end{definition}
For any fixed $n\in \N$, the  existence of $\psi_n(\zeta;x)$ amounts to a problem of Linear Algebra and rests upon the nonvanishing of the Hankel determinant of the moments \eqref{moments} of the measure $\d\nu (\zeta;x)$ 
\begin{equation*}
  \det\Big[\mu_{\ell+j-2}(x)\Big]_{\ell,j=1}^n\neq 0.
\end{equation*}
Recall also that the normalizing constants are related to the Hankel determinants by
\begin{equation}\label{hnid}
  h_n(x) = -\frac{\det[\mu_{\ell+k-2}(x)]_{\ell,k=1}^{n+1}}{\det[\mu_{\ell+k-2}(x)]_{\ell,k=1}^n}.
\end{equation}
Now combining \eqref{hnid} with \eqref{the1} and \eqref{VoYa}, we obtain for $n\geq 1$
\begin{equation}\label{PIIconnec}
  h_n(x) = 2(-)^{n-1}\left[\frac{n!}{(2n)!}\right]^2\left(\frac{\mathcal Q_n(x)}{\mathcal Q_{n-1}(x)}\right)^2,\hspace{1cm}u(x;n) = -\frac{1}{2}\frac{h_n'(x)}{h_n(x)}.
\end{equation}
Hence zeros of the $n$-th Vorob'ev-Yablonski polynomial $\mathcal{Q}_n(x)$, respectively poles of the $n$-th rational solutions $u(x;n)$ to \eqref{PII}, are in one-to-one correspondence with the exceptional values of the parameter $x$ for which the $n$-th orthogonal polynomial $\psi_n(\zeta;x)$ \eqref{o1},\eqref{o2} ceases to exist.\smallskip


In this perspective, our second result confirms an analog one in \cite{BM}, namely it shows that
the Vorob'ev-Yablonski polynomials of large degree (after a rescaling) are zero-free outside a star shaped region $\overline{\Delta}\subset\mathbb{C}$ defined as follows
\begin{definition}\label{defstar} Let $a=a(x)$ denote the (unique) solution of the cubic equation
\begin{equation}\label{cubic}
  1+2xa^2-4a^3=0
\end{equation}
subject to boundary condition
\begin{equation*}
 a=\frac{x}{2}+\mathcal O\left(x^{-2}\right),\hspace{0.5cm}x\rightarrow\infty.
\end{equation*}
The three branch points $x_k=-\frac{3}{\sqrt[3]{2}}e^{\frac{2\pi i}{3} k},k=0,1,2$ of equation \eqref{cubic} form the vertices of the star shaped region $\overline{\Delta}=\Delta\cup\partial{\Delta}$
depicted in Figure \ref{star} below which contains the origin and whose boundary $\partial\Delta$ consists of three edges defined implicitly via the requirement
\begin{equation}\label{starb}
  \Re\left\{-2\ln\left(\frac{1+\sqrt{1+2a^3}}{ia\sqrt{2a}}\right)+\sqrt{1+2a^3}\left(\frac{4a^3-1}{3a^3}\right)\right\}=0.
\end{equation}
Here, all branches of fractional exponents and logarithms are chosen to be principal ones.
\end{definition}
\begin{figure}[tbh]
\begin{center}
\includegraphics[width=0.4\textwidth]{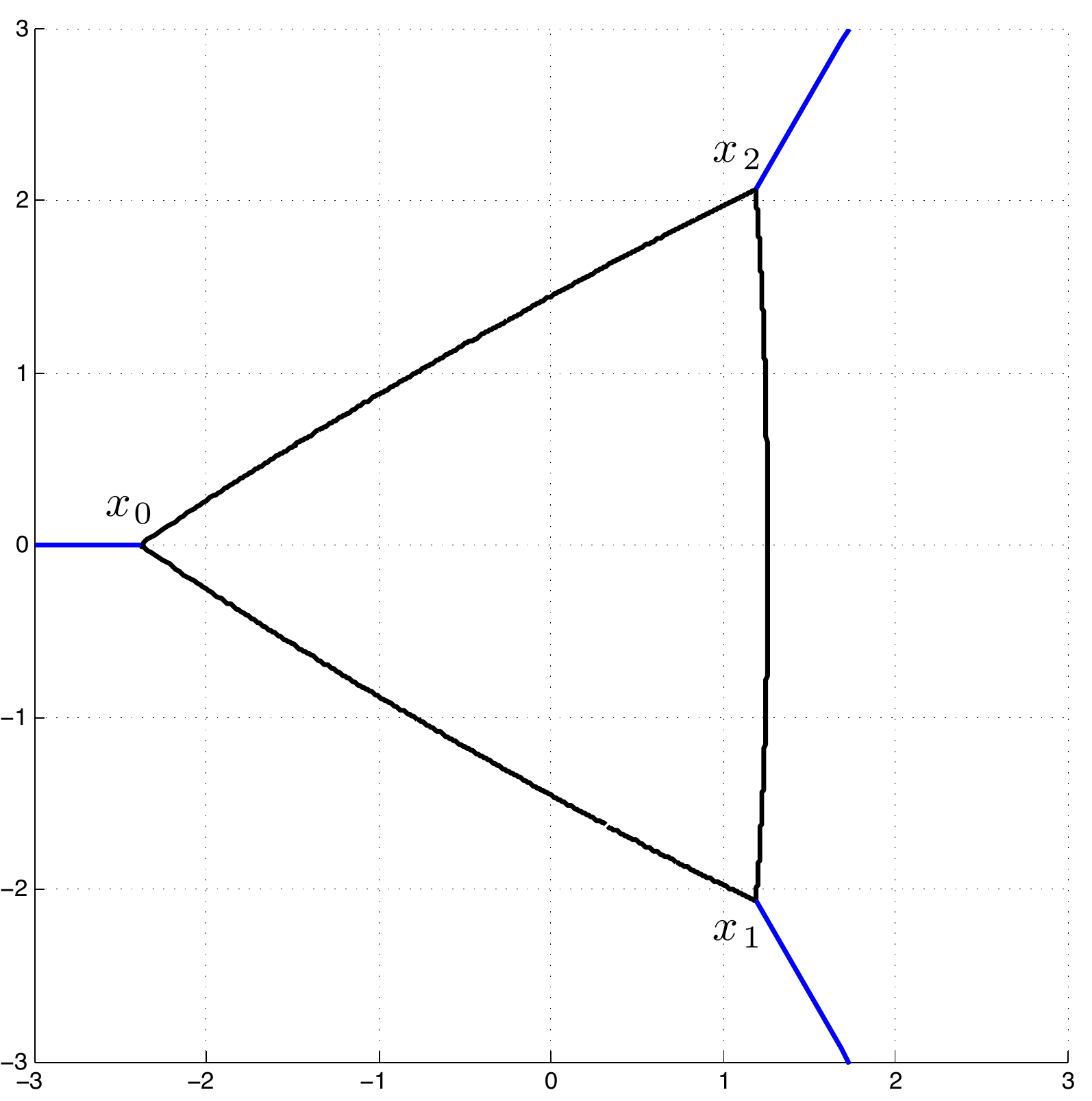}
\end{center}
\caption{The star shaped region $\overline{\Delta}=\Delta\cup\partial\Delta$. The boundary $\partial\Delta$ is given as the union of the three black solid curves.}
\label{star}
\end{figure}

In terms of the latter definition, our second main result shows that the region $\mathbb{C}\backslash\overline{\Delta}$ does not contain any zeros of $\mathcal{Q}_n(n^{\frac{2}{3}}x)$, provided $n$ is large enough. We have
\bth[see \cite{BM}, Theorem $1$]\label{theo2} Let $x\in\mathbb{C}:\textnormal{dist}(x,\overline{\Delta})\geq\delta>0$, then the orthogonal polynomials $\psi_n(\zeta;n^{\frac{2}{3}}x),\zeta\in S^1$ defined by \eqref{o1} and \eqref{o2} exist if
$n$ is sufficiently large. Equivalently, the (rescaled) Vorob'ev-Yablonski polynomials $\wh Q_n(x)= \mathcal{Q}_n(n^{\frac{2}{3}}x)$ for large $n$ have no zeros in the same region of the complex $x$-plane.
\et

We point out that while the final result overlaps (see Remark \ref{remBM1}) with the result of \cite{BM}, the method is substantially different since we start from the new determinantal expression of $\mathcal Q_n(x)$ obtained in Theorem \ref{theo1}.\smallskip

At this point we decided to perform asymptotic analysis only in the interior and exterior of the region $\Delta$; hence we shall not address issues related to the asymptotic behavior when $x$ is on the boundary (or vicinity) of $\Delta$, presumably the result would only confirm those of the forthcoming paper \cite{BM2}.
We are also focusing on the location of the zeros of $\wh Q_n(x)= \mathcal Q_n( n^\frac 23 x)$ inside of $\Delta$ (hence, location of the poles of $u(x;n)$) rather than the asymptotic behavior of the rational solution $u(x;n)$ itself, not to unnecessarily duplicate the results.\smallskip

There are interesting differences in the methods of our analysis inside $\Delta$, compared to the one in \cite{BM}\footnote{In loc. cit. the region $\Delta$ is termed the ``elliptic region''.} although the end result is the same. 
In \cite{BM} the author need to introduce an elliptic curve (of genus $1$) dependent on the value of $x$ in $\Delta$; in contrast, we need to introduce a {\em hyperelliptic curve} of genus $2$ of the form 
\be
X=\Big\{(z,w):\ w^2 =P_3\left(z^2\right)\Big\}
\label{Rg2} 
\ee
where $P_3(\z)=(\z+a^2)(\z+b^2)(\z+c^2)$ is a polynomial of degree $3$ with distinct roots given implicitly in \eqref{gen2con1} and \eqref{gen2con3}.
In \cite{BM}, the authors introduce an exceptional set of discrete points in order to complete the Riemann-Hilbert analysis inside the star shaped region $\Delta$, compare equations (4-96) and (4-97) in the aforementioned text. 
In our case the corresponding exceptional set is first defined in terms of the vanishing of a Riemann Theta function of genus $2$ (see App. \eqref{app2}); however, 
 given the high symmetry of our curve $X$, we will eventually reduce the appearance of $\Theta(z|\,{\boldsymbol \tau})$ in the definition of the corresponding exceptional set \eqref{cond0} to a condition which involves only a theta function $\vartheta(\rho)=\vartheta(\rho|\,\varkappa)$ associated to an elliptic curve. In order to explain in detail the condition, let   us set 
\be
\d \phi(z)= \frac {\sqrt{P_3(z^2)} }{z^4}\, \d z
\ee
and recall that the parameters $a,b,c$ (i.e. the branchpoints of $X$)  all depend on $x$ implicitly via \eqref{gen2con1} and \eqref{gen2con3}. Our analog to (4-96), (4-97) in \cite{BM} reads as follows.
\bth
\label{theo3} 
Let $\mathcal Z_{n}\subset \Delta$ be the discrete collection of points $\{x_{n,k}\}$ defined
via
\be\label{th3zero}
 \vartheta\left(\frac{n}{2\pi i}\left[\oint_{\mathcal{B}_1}\d\phi+\varkappa_2\oint_{\mathcal{A}_1}\d\phi\right]
  +\frac{1}{2}\left[\int_{a_1}^0\frac{\eta_2}{\mathbb{A}_{22}}+\frac{\varkappa_2}{2}\right]\right)=0
\ee
where $\vartheta(\rho) = \vartheta_3(\rho|\,\varkappa_2)= \sum_{m\in \Z}\exp[i\pi m^2 \varkappa_2 + 2\pi i m \rho]$ is the Jacobi theta function and we put 
\be
\eta_2 = \frac {z\,\d z}{w}\ ,\ \ \mathbb A_{22} = \oint_{\mathcal A_2} \eta_2\ ,\ \ \ \ \ \varkappa_2= \frac {  \oint_{\mathcal B_2} \eta_2}{ \oint_{\mathcal A_2} \eta_2}
\ee
for a specific choice of homology basis $\{\mathcal{A}_j,\mathcal{B}_j\}_{j=1}^2$ shown in Figure \ref{modRHP1}. Uniformly for $x$ belonging to any compact subset of $\Delta \setminus \mathcal Z_{n}$ the polynomial $\psi_{n}(\zeta;n^\frac 23 x),\zeta\in S^1$ exists for $n$ sufficiently large. Moreover, for $x$ in the same compact set,  $\wh Q_n(x)\neq 0$ for $n$ large enough.
\et

The condition \eqref{th3zero} is equivalently formulated as 
\be
\label{th3zeroprime}
\frac{n}{2\pi i}\left[\oint_{\mathcal{B}_1}\d\phi+\varkappa_2\oint_{\mathcal{A}_1}\d\phi\right]
  +\frac{1}{2}\left[\int_{a_1}^0\frac{\eta_2}{\mathbb{A}_{22}}+\frac{\varkappa_2}{2}\right] =
\frac {1+ \varkappa_2}2  + k + \ell \varkappa_2,\ \ \ k, \ell \in \Z
\ee
The integrals $\oint_{\mathcal{B}_1, \mathcal A_1}\d\phi$ are purely imaginary (see \eqref{gen2con3}) and since $\Im \varkappa_2>0$,  any complex number $\rho$ can be uniquely expressed as $\sigma+ \varkappa_2 \xi \ \ \sigma, \xi  \in \R$. Thus the condition \eqref{th3zeroprime} can be expressed as the pair of quantization conditions 
\be
\frac{n}{2\pi i} \oint_{\mathcal{B}_1}\d\phi  = \frac 12 + k + \sigma \ ,\qquad
\frac{n}{2\pi i} \oint_{\mathcal{A}_1}\d\phi  = \frac 1 4 + \ell  + \xi\ ,
\label{quantcond}
\ee
where $k,\ell \in \Z$ and $\sigma + \varkappa_2 \xi = \frac{1}{2}\int_{a_1}^0\frac{\eta_2}{\mathbb{A}_{22}}$.
The lines of the quantization conditions \eqref{quantcond} are shown in Figure \ref{quant} for different values of $n$. We note that the agreement is remarkably much better - even for very small values of $n$ - than what Theorem \ref{theo4} below leads to expect.

\begin{figure}
\includegraphics[width=0.23\textwidth]{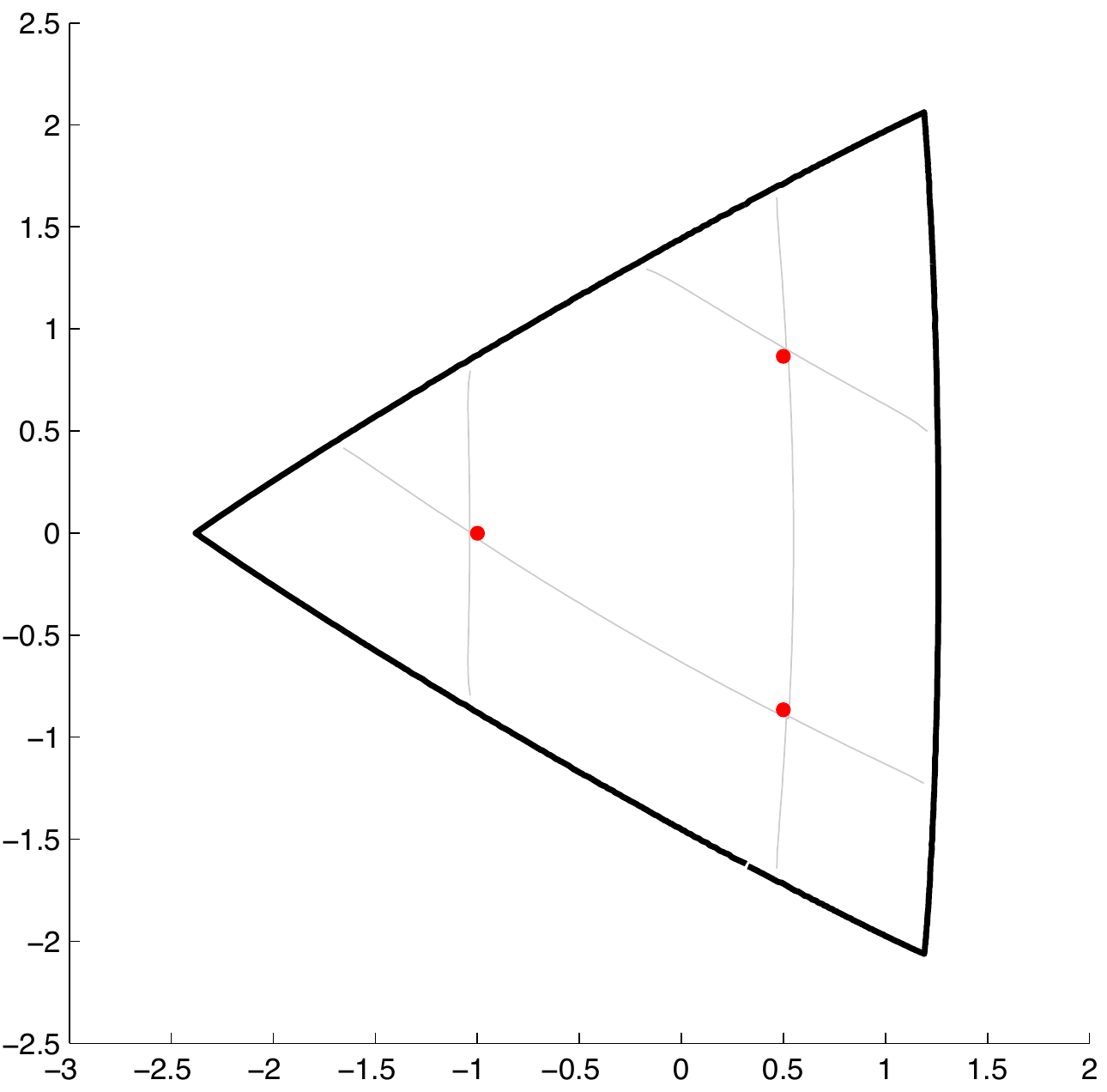}
\includegraphics[width=0.23\textwidth]{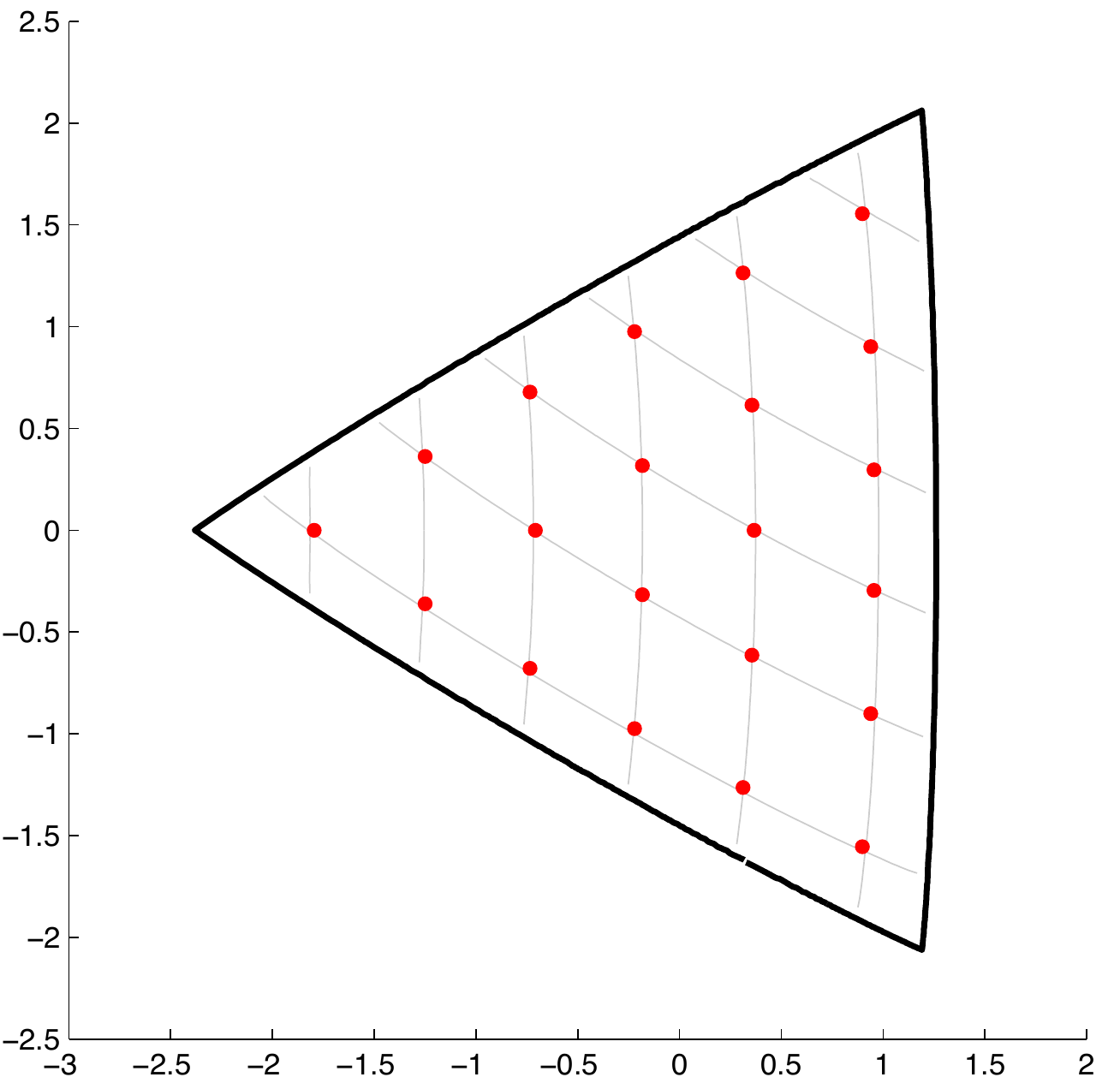}
\includegraphics[width=0.23\textwidth]{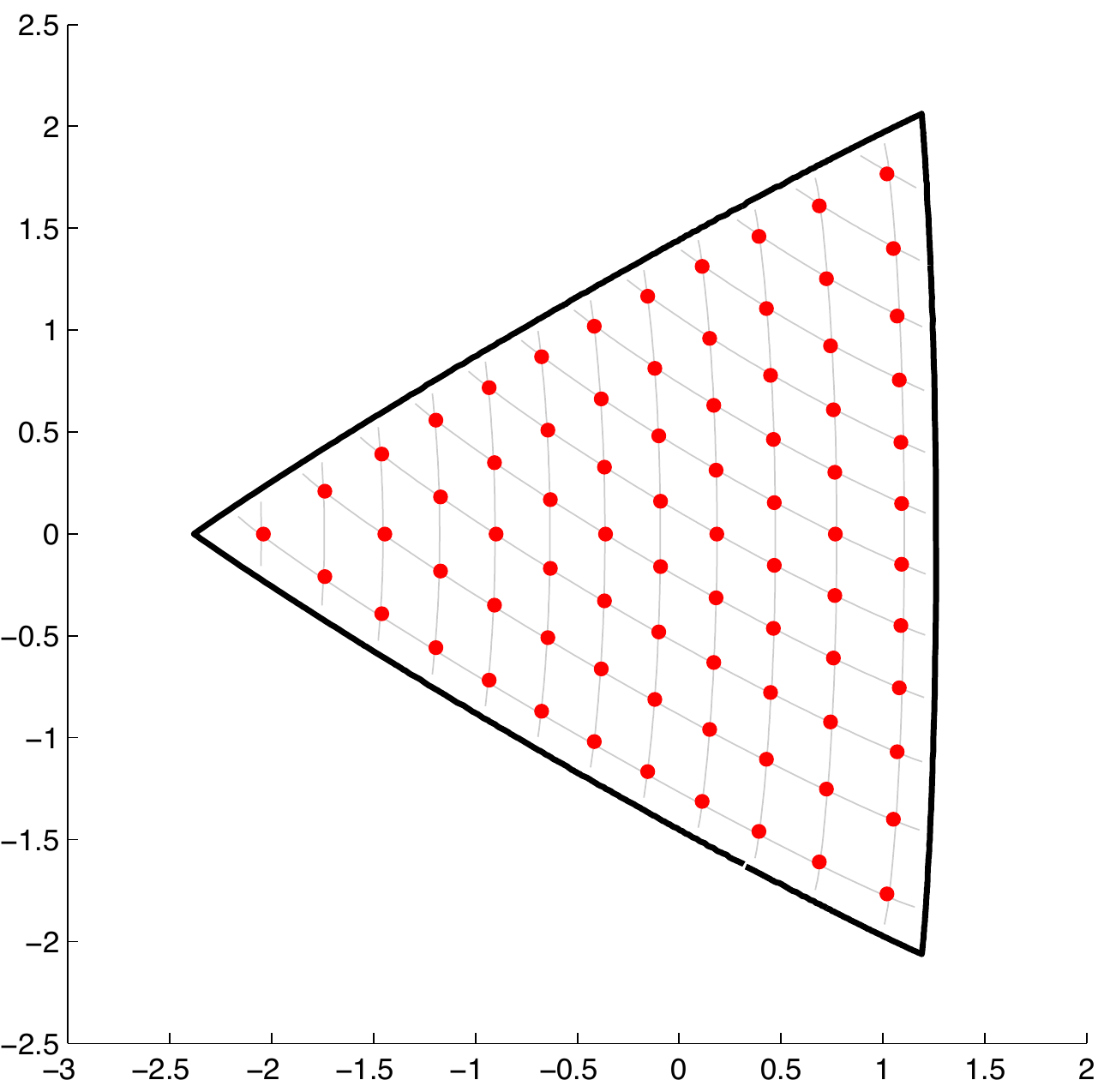}
\includegraphics[width=0.23\textwidth]{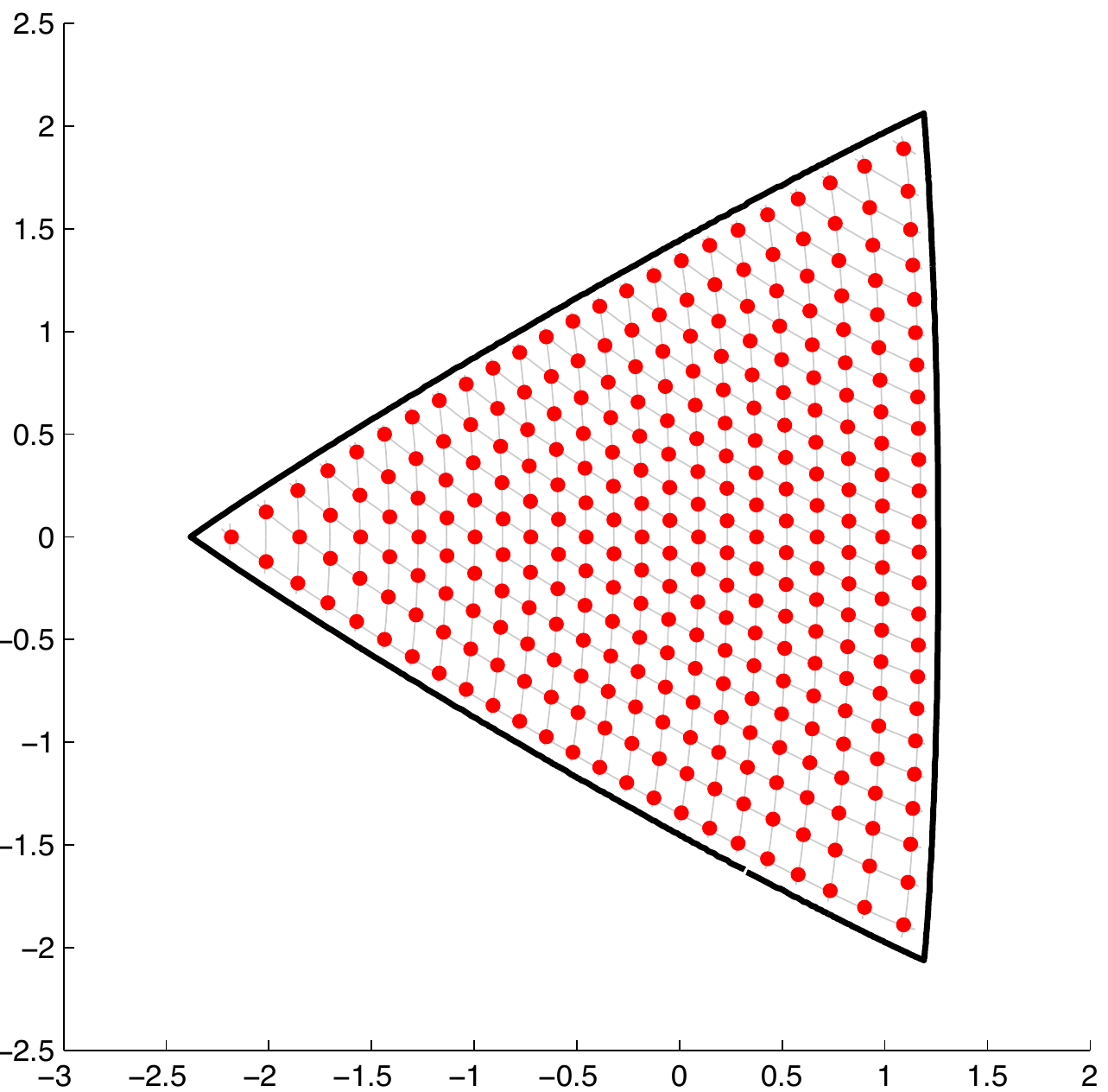}
\caption{The lines expressing the quantization conditions \eqref{quantcond} and the zeros of the polynomial $\wh Q_n(x)$ computed numerically, for $n = 2,6,12,24$ (from left to right). As can be seen, the zeros of $\wh Q_n(x)$ form a regular pattern, a feature which was first observed in \cite{CM}.}
\label{quant}
\end{figure}\smallskip

The outlined reduction of the appearing Riemann theta function to a Jacobi theta function combined with an application of the argument principle to smooth functions yields the following Theorem that localizes the zeros of $\wh Q_n(x)$ within disks of radius $\mathcal O(n^{-1})$.
\bth\label{theo4} 
For each compact subset  $K$ of the interior of $\Delta$, and for any arbitrarily small $r_0>0$ there exists $n_0= n_0(K, r_0)$  such that the zeros of  $ \wh Q_n(x) = \mathcal{Q}_n(n^{\frac{2}{3}}x),n\geq n_0$ that fall within $K$ are inside disks of radius $r_0/n$ centered around the points of the exceptional set $\mathcal Z_n$.
\et
The paper is organised as follows: we first prove Theorem \ref{theo1} in Section \ref{sec1} by applying identity \eqref{JT}. After that preliminary steps for the 
Riemann-Hilbert analysis of the (pseudo) orthogonal polynomials $\{\psi_n(\zeta;x)\}$ are taken in Section \ref{sec2}. This includes a rescaling of the weight and the construction of the relevant $g$-functions which are used outside and inside the star. The $g$-functions reduce the initial RHP to the solution of model problems and we state their explicit construction in Section \ref{sec3}. Section \ref{sec4} completes the proofs of Theorems \ref{theo2} and \ref{theo3}. In the end we compare our results obtained outside and inside the star to \cite{BM}, this is done in Section \ref{sec5} which also gives the proof of Theorem \ref{theo4}.


\section{Proof of Theorem \ref{theo1}}\label{sec1}
The identity \eqref{the1} follows from several equivalence transformations. First we go back to \eqref{JT} and notice that $q_0(x)=1$, the convention $q_k(x)\equiv 0$ for $k<0$ as well as the empty product imply
\begin{equation*}
  \mathcal{Q}_n(x) = \prod_{k=1}^n(2k+1)^{n-k}\,\det\Big[q_{n-2\ell+j}(x)\Big]_{\ell,j=0}^n,\hspace{0.5cm}n\geq 0.
\end{equation*}
Thus \eqref{the1} is in fact equivalent to the identity
\be\label{eq1}
  \left\{\det\Big[q_{n-2\ell+j}(x)\Big]_{\ell,j=1}^n\right\}^2 =(-)^{\lfloor{\frac{n}{2}\rfloor}}2^{(n-1)^2}\,\det\Big[\mu_{\ell+j-2}(x)\Big]_{\ell,j=1}^n,\hspace{0.5cm}n\geq 1.
\ee
Since
\begin{equation*}
  F_1(t;x)\big(F_1(t;x)\pm F_1(-t;x)\big)=F_2(2t;x)\pm 1
\end{equation*}
holds identically in $t$ and $x$, we have from comparison
\bea
  2^{k-1}\mu_k(x)+\frac{\delta_{k0}}{2} &=& \sum_{m=0}^kq_{k-2m}(x)q_{2m}(x),\hspace{1.2cm}k\geq 0,\label{co1}\\
  2^{k-1}\mu_k(x)-\frac{\delta_{k0}}{2} &=& \sum_{m=0}^kq_{k-2m-1}(x)q_{2m+1}(x),\hspace{0.5cm}k\geq 0.\label{co2}
\eea
Now back to the left hand side of \eqref{eq1} with $n\geq 1$. We first shift indices, then permute columns and rows in the second factor, transpose the first matrix and then evaluate
the product
\begin{align*}
  \bigg\{\det\Big[q_{n-2\ell+j}&(x)\Big]_{\ell,j=1}^n\bigg\}^2 = \left\{\det\Big[q_{n-1-2\ell+j}(x)\Big]_{\ell,j=0}^{n-1}\right\}^2=
  \det\Big[q_{n-1-2\ell+j}(x)\Big]_{\ell,j=0}^{n-1}\,\det\Big[q_{2\ell-j}(x)\Big]_{\ell,j=0}^{n-1}\\
&=\det\Big[q_{n-1-2j+\ell}(x)\Big]_{\ell,j=0}^{n-1}\,\det\Big[q_{2\ell-j}(x)\Big]_{\ell,j=0}^{n-1}=
  \det\left[\,\sum_{m=0}^{n-1}q_{n-1-2m+\ell}(x)q_{2m-j}(x)\right]_{\ell,j=0}^{n-1}.
\end{align*}
Now use \eqref{co1} and \eqref{co2} to evaluate the entries. In the first row
\begin{align*}
  &\sum_{m=0}^{n-1}q_{n-1-2m}(x)q_{2m-j}(x) \hspace{0.65cm}= 2^{n-2-j}\mu_{n-1-j}(x),\hspace{0.5cm}j=0,\ldots,n-2\\
&\sum_{m=0}^{n-1}q_{n-1-2m}(x)q_{2m-(n-1)}(x) = \begin{cases}
               \frac{\mu_0(x)}{2}-\frac{1}{2}, &n\equiv 0\mod 2\\
	\frac{\mu_0(x)}{2}+\frac{1}{2},&n\equiv 1\mod 2.
              \end{cases}
\end{align*}
For the second and subsequent rows
\begin{equation*}
  \sum_{m=0}^{n-1}q_{n-1-2m+\ell}(x)q_{2m-j}(x) = 2^{n-2-j+\ell}\mu_{n-1-j+\ell}(x),\hspace{0.5cm}j=0,\ldots,n-1,\ \ \ell=1,\ldots,n-1
\end{equation*}
which shows that
\begin{equation*}
  \left\{\det\Big[q_{n-2\ell+j}(x)\Big]_{\ell,j=1}^n\right\}^2 =\det\Big[2^{n-2-j+l}\tilde{\mu}_{n-1-j+\ell}(x)\Big]_{\ell,j=0}^{n-1} =
  (-)^{\lfloor\frac{n}{2}\rfloor}\det\Big[2^{\ell+j-1}\tilde{\mu}_{\ell+j}(x)\Big]_{\ell,j=0}^{n-1}
\end{equation*}
where we permuted only the columns ($j\mapsto n-1-j$) in the last step and introduced
\begin{equation*}
  \tilde{\mu}_k(x) = \mu_k(x),\ \ k\geq 1;\hspace{0.5cm} \tilde{\mu}_0(x) = \begin{cases}
                                                                             0, & n\equiv 0\mod 2,\\
2, & n\equiv 1\mod 2.
                                                                            \end{cases}
\end{equation*}
Notice that we have suppressed the dependency on $n$ in the notation of $\tilde{\mu}_k(x)$. Factoring out common factors we continue
\begin{equation*}
  \left\{\det\Big[q_{n-2\ell+j}(x)\Big]_{\ell,j=1}^n\right\}^2=(-)^{\lfloor\frac{n}{2}\rfloor}2^{(n-1)^2}\frac{1}{2}
  \det\Big[\tilde{\mu}_{\ell+j}(x)\Big]_{\ell,j=0}^{n-1} = (-)^{\lfloor\frac{n}{2}\rfloor}2^{(n-1)^2}\frac{1}{2}\det\Big[\tilde{\mu}_{\ell+j-2}(x)\Big]_{\ell,j=1}^n
\end{equation*}
and therefore, compare \eqref{eq1}, are left to show that
\begin{equation}\label{eq2}
  \det\Big[\tilde{\mu}_{\ell+j-2}(x)\Big]_{\ell,j=1}^{n} = 2\det\Big[\mu_{\ell+j-2}(x)\Big]_{\ell,j=1}^{n},\hspace{0.5cm}n\geq 1.
\end{equation}
This identity is definitely satisfied for $n=1$, hence let us assume that $n\geq 2$. By multilinearity 
\be\nonumber
\det\Big[\tilde{\mu}_{\ell+j-2}(x)\Big]_{\ell,j=1}^n=\det\Big[\mu_{\ell+j-2}(x)\Big]_{\ell,j=1}^n+(-)^{n-1}\det\Big[\mu_{\ell+j}(x)\Big]_{\ell,j=1}^{n-1},
\ee
thus we need to verify that
\begin{equation}\label{eq3}
  (-)^{n-1}\det\Big[\mu_{\ell+j}(x)\Big]_{\ell,j=1}^{n-1} = \det\Big[\mu_{\ell+j-2}(x)\Big]_{\ell,j=1}^{n},\hspace{0.5cm}n\geq 2.
\end{equation}
Both sides in the latter equation are polynomials in $x\in\mathbb{C}$, hence if we manage to establish equality in \eqref{eq3} outside a set $E\subset\mathbb{C}$ of 
measure zero, it follows by continuation for all $x\in\mathbb{C}$. In our case, we will verify \eqref{eq3} for $x\in\mathbb{C}\backslash E$ with
\begin{equation*}
  E=\left\{x\in\mathbb{C}:\ \det\Big[\mu_{j+k}(x)\Big]_{j,k=1}^m=0,\ \ m=1,\ldots,n-2\right\}
\end{equation*}
using the following algorithm: we start $(1)$ on the right hand side of \eqref{eq3} and add appropriate combinations of rows to subsequent rows, starting from row $n$ and continuing with row $n-1$, etc. Formally 
with $\mu_k(x)\equiv 0$ for $k<0$
\begin{eqnarray*}
  \mu_{\ell+j-2}&\mapsto& \mu^{(1)}_{\ell,j} = \mu_{\ell+j-2}-\left\{\frac{x\mu_{(\ell-1)+j-2}}{\ell-1}-\frac{\mu_{(\ell-3)+j-2}}{\ell-1}\right\},\hspace{0.5cm}\ell=4,\ldots,n\\
  \mu_{\ell+j-2}&\mapsto& \mu^{(1)}_{\ell,j} = \mu_{\ell+j-2}-\left\{\frac{x\mu_{(\ell-1)+j-2}}{\ell-1}\right\},\hspace{0.5cm}\ell=2,3
\end{eqnarray*}
for any $j\in\{1,\ldots,n\}$. Recalling \eqref{rec1} this step implies
\begin{equation*}
  \mu^{(1)}_{\ell,1}= \mu^{(1)}_{1,\ell}=0,\hspace{0.25cm} \ell=2,\ldots,n;\hspace{0.5cm}\mu^{(1)}_{\ell,\ell}= -\mu_{2\ell-2},
  \hspace{0.25cm}\ell=2,\ldots,n,\ \ell\neq 3;\hspace{0.5cm}\mu^{(1)}_{3,3}=-\mu_4-\frac{\mu_1}{2}.
\end{equation*}
In the next step $(2)$ we add an $\alpha_1$-multiple of the second row to the third row and then an $\alpha_2$-multiple of the second column to the third column, where
\begin{equation*}
  \alpha_1 = \frac{\mu^{(1)}_{23}+\mu_3}{\mu_2},\hspace{0.5cm}\alpha_2 = \frac{\mu^{(1)}_{32}+\mu_3}{\mu_2},
\end{equation*}
provided $\mu_2\neq 0$, which is satisfied for $x\in\mathbb{C}\backslash E$. Using again the recursion \eqref{rec1}, this move leads to the replacement
\begin{eqnarray*}
  \mu^{(1)}_{\ell,j}&\mapsto& \mu^{(2)}_{\ell,j} = \mu^{(1)}_{\ell,j},\hspace{0.5cm}\ell,j= 4,\ldots,n\\
  \mu^{(1)}_{2,j}&\mapsto& \mu^{(2)}_{2,j} = \mu^{(1)}_{2,j},\hspace{0.5cm}j=4,\ldots,n\\
  \mu^{(1)}_{\ell,2}&\mapsto& \mu^{(2)}_{\ell,2} = \mu^{(1)}_{\ell,2},\hspace{0.5cm}\ell=4,\ldots,n
\end{eqnarray*}
as well as
\begin{eqnarray*}
  \mu^{(1)}_{3,j}&\mapsto& \mu^{(2)}_{3,j} = \mu^{(1)}_{3,j}+\alpha_1\mu^{(1)}_{2,j},\hspace{0.5cm}j=4,\ldots,n\\
  \mu^{(1)}_{\ell,3}&\mapsto& \mu^{(2)}_{\ell,3} = \mu^{(1)}_{\ell,3}+\alpha_2\mu^{(1)}_{\ell,2},\hspace{0.5cm}\ell=4,\ldots,n
\end{eqnarray*}
and most importantly
\begin{equation}
  \mu^{(1)}_{2,2}\mapsto \mu^{(2)}_{2,2} = -\mu_2,\hspace{0.5cm} \mu^{(1)}_{2,3}\mapsto \mu^{(2)}_{2,3}=-\mu_3,\hspace{0.5cm}
  \mu^{(1)}_{3,2}\mapsto\mu^{(2)}_{3,2}=-\mu_3,\hspace{0.5cm}\mu^{(1)}_{3,3}\mapsto\mu^{(2)}_{3,3}=-\mu_4.
\end{equation}
Hence step $(2)$ shows that
\begin{equation*}
  \det\Big[\mu_{\ell+j-2}(x)\Big]_{\ell,j=1}^{n} = \det\begin{bmatrix}
                                                        \mu_0 & 0 & 0 & 0 &\cdots & 0\\
0 & -\mu_2 & -\mu_3 & \mu_{24}^{(2)} & \cdots & \mu_{2n}^{(2)}\\
0 & -\mu_3 & -\mu_4 & \mu_{34}^{(2)} & \cdots & \mu_{3n}^{(2)}\\
0 & \mu_{42}^{(2)} & \mu_{43}^{(2)} & -\mu_6 & & \vdots\\
\vdots& \vdots & \vdots & &\ddots & \\
0 & \mu_{n2}^{(2)} & \mu_{n3}^{(2)} & \cdots & & -\mu_{2n-2}\\
                                                       \end{bmatrix},\hspace{0.5cm}\mu_2\neq 0.
\end{equation*}
In step $(3)$ we add a $\beta_{11}$-multiple of the second column and a $\beta_{21}$-multiple of the third column to the fourth column, followed by
then adding a $\beta_{12}$-multiple of the second row and a $\beta_{22}$-multiple of the third row to the fourth row. Here $\{\beta_{jk}\}$ are 
determined from the linear system
\begin{equation*}
  \begin{bmatrix}
   \mu_2 & \mu_3\\
  \mu_3 & \mu_4 \\
  \end{bmatrix}\begin{bmatrix}
  \beta_{11}& \beta_{12}\\
  \beta_{21}& \beta_{22}\\
\end{bmatrix} = \begin{bmatrix}
  \mu_{24}^{(2)} +\mu_4 &\mu_{42}^{(2)}+\mu_4\\
  \mu_{34}^{(2)} +\mu_5 & \mu_{43}^{(2)}+\mu_5\\
\end{bmatrix}.
\end{equation*}
provided the determinant of its coefficients matrix, i.e. $\det\big[\mu_{j+k}\big]_{j,k=1}^2$ does not vanish, which again is guaranteed for 
$x\in\mathbb{C}\backslash E$. In terms of the recursion \eqref{rec1}, this leads us to
\begin{equation*}
  \det\Big[\mu_{\ell+j-2}(x)\Big]_{\ell,j=1}^{n} =\det\begin{bmatrix}
                                                       \mu_0 & 0 & 0 & 0 & 0 &\cdots & 0\\
0 & -\mu_2 & -\mu_3 & -\mu_4 & \mu_{25}^{(3)} & \cdots & \mu_{2n}^{(3)}\\
0 & -\mu_3 & -\mu_4 & -\mu_5 & \mu_{35}^{(3)} &\cdots & \mu_{3n}^{(3)}\\
0 & -\mu_4 & -\mu_5 & -\mu_6 & \mu_{45}^{(3)} & \cdots & \mu_{4n}^{(3)}\\
0 & \mu_{52}^{(3)} & \mu_{53}^{(3)} & \mu_{54}^{(3)} & -\mu_8 & & \vdots\\
\vdots & \vdots & \vdots & \vdots & &\ddots &\\
0 & \mu_{n2}^{(3)} & \mu_{n3}^{(3)} & \mu_{n4}^{(3)} & \cdots & & -\mu_{2n-2}\\
                                                      \end{bmatrix},\hspace{0.5cm}\det\big[\mu_{j+k}\big]_{j,k=1}^2\neq 0.
\end{equation*}
Step $(3)$ is then followed by step $(4)$ in which we add appropriate combinations of the second, third and fourth column/row to the fifth column/row, and so forth.
After $(n-1)$ steps in this algorithm, we end up with the identity
\begin{equation*}
  \det\Big[\mu_{\ell+j-2}(x)\Big]_{\ell,j=1}^{n} =\det\begin{bmatrix}
                                                       \mu_0 & 0 & 0 & \cdots &0 & 0\\
0 & -\mu_2 & -\mu_3 & \cdots & -\mu_{n-1} &  \mu_{2n}^{(n-1)}\\
0 & -\mu_3 & -\mu_4 & \cdots & -\mu_n &  \mu_{3n}^{(n-1)}\\
\vdots & \vdots & \vdots & \ddots &   & \vdots\\
0 & -\mu_{n-1} & -\mu_n &  & -\mu_{2n-4} & \mu_{n-1,n}^{(n-1)}\\
& & & & & \\
0 & \mu_{n2}^{(n-1)} & \mu_{n3}^{(n-1)} & \cdots & \mu_{n,n-1}^{(n-1)}  & -\mu_{2n-2}\\
                                                      \end{bmatrix},\hspace{0.5cm}\det\big[\mu_{j+k}\big]_{j,k=1}^{n-3}\neq 0.
\end{equation*}
In the final step $(n)$ we add combinations of the second, third, fourth,$\ldots$,$(n-1)^{\textnormal{st}}$ row/column to the $n^{\textnormal{th}}$ row/column according to the
system
\begin{equation*}
  \begin{bmatrix}
   \mu_2 & \mu_3 & \cdots & \mu_{n-1}\\
  \vdots & & & \vdots\\
  \mu_{n-1} & \mu_{n} &\cdots & \mu_{2n-4}\\
  \end{bmatrix}\begin{bmatrix}
  \gamma_{11} & \gamma_{12}\\
  \vdots & \vdots\\
  \gamma_{n-2,1} & \gamma_{n-2,2} \\
  \end{bmatrix} = \begin{bmatrix}
  \mu_{2n}^{(n-1)}+\mu_n & \mu_{n2}^{(n-1)}+\mu_n\\
  \vdots & \vdots\\
  \mu_{n-1,n}^{(n-1)} +\mu_{2n-3} & \mu_{n,n-1}^{(n-1)}+\mu_{2n-3}\\
  \end{bmatrix}
\end{equation*}
and establish \eqref{eq3} from the recursion \eqref{rec1} after extracting $(n-1)$ signs, provided that $\det\big[\mu_{j+k}\big]_{j,k=1}^{n-2}\neq 0$, which holds
for $x\in\mathbb{C}\backslash E$. This verifies \eqref{eq3} by analytic continuation and tracing back
all equivalence transformations completes therefore the proof of Theorem \ref{theo1}.

\section{Riemann-Hilbert analysis - preliminary steps}\label{sec2}
It is well known that orthogonal polynomials can be characterized in terms of the solution of a Riemann-Hilbert problem (RHP), first introduced by Fokas, Its and Kitaev \cite{FIK}. 
In present context of \eqref{meas}, the relevant RHP is defined as follows:
\begin{definition} 
Let $\gamma$ be a simple, smooth Jordan curve encircling the origin in clockwise orientation. 
Determine the $2\times 2$ matrix-valued piecewise analytic function $\Gamma(z)\equiv\Gamma(z;x,n)$ such that
 \begin{itemize}
  \item $\Gamma(z)$ is analytic for $z\in\mathbb{C}\backslash\gamma$
  \item The boundary values on $\gamma$ are related via
  \be\label{j1}
    \Gamma_+(z) = \Gamma_-(z)\begin{bmatrix}
                                          1 & w(z;x)\\
0 & 1\\
                                         \end{bmatrix},\hspace{0.5cm}z\in\gamma;\hspace{1cm} w(z;x) =\frac{1}{2\pi i}e^{-\theta(z;x)}\frac{1}{z}
  \ee
  with $\theta$ as in \eqref{meas}.
  \item As $z\rightarrow\infty$, we have
  \be\label{asy1}
    \Gamma(z) = \left(I+\mathcal O\left(z^{-1}\right)\right)z^{n\sigma_3}
  \ee
 \end{itemize}
\end{definition}
The solvability of the $\Gamma$-RHP is equivalent to the existence of the orthogonal polynomial $\psi_{n}(\zeta;x)$, in fact \cite{D}
\begin{equation}\label{Drel}
 \psi_n(\zeta;x) = \Gamma_{11}(\zeta;x,n),
\end{equation}
and in addition
\begin{equation*}
 h_n(x) = -2\pi i\lim_{z\rightarrow\infty}z\Big(\Gamma(z;x,n)z^{-n\sigma_3}-I\Big)_{12},\hspace{0.5cm}
 \big(h_{n-1}(x)\big)^{-1}=\frac{i}{2\pi}\lim_{z\rightarrow\infty}z\Big(\Gamma(z;x,n)z^{-n\sigma_3}-I\Big)_{21}.
\end{equation*}
We will solve the latter RHP as $n\rightarrow\infty$ for rescaled $x\in\mathbb{C}$  outside and inside (there subject to an additional constraint) the star shaped region described in Definition \eqref{defstar}. Our approach uses standard methods from
the Deift-Zhou nonlinear steepest descent framework (cf. \cite{DZ},\cite{DKM},\cite{DKMVZ}) and consists of a series of explicit and invertible transformations.
\subsection{Rescaling and the abstract g-function} 
In order to study the polynomials $\wh Q_n(x) = \mathcal Q_n(n^\frac 23 x)$ we consider the following change of variables 
\begin{equation}\label{orel}
 \psi_n^o(z;x) = N^{\frac{n}{3}}\psi_n\left(N^{-\frac{1}{3}}z;N^{\frac{2}{3}}x\right),\hspace{1cm} h_n^o(x) = N^{\frac{2n}{3}}h_n\left(N^{\frac{2}{3}}x\right).
\end{equation}
Consequently, the measure of orthogonality of these new orthogonal polynomials is 
\begin{equation}\label{resc}
  \d\nu(z;x)\mapsto \d\nu^o(z;x) = \frac{1}{2\pi i} e^{-N\theta(z;x)}\frac{\d z}{z},\hspace{0.5cm}N\in\mathbb{N}
\end{equation}
Under the scaling \eqref{resc}, the initial $\Gamma$-RHP is replaced by a RHP for the function $\Gamma^o(z)\equiv\Gamma^o(z;x,n,N)$ with jump
\begin{equation*}
    \Gamma^o_+(z)=\Gamma^o_-(z)\begin{bmatrix}
					     1 & w^o(z;x)\\
					    0 & 1\\
                                           \end{bmatrix},\hspace{0.5cm}z\in \gamma;\hspace{1cm}w^o(z;x) = \frac{1}{2\pi i}e^{-N\theta(z;x)}\frac{1}{z}
\end{equation*}
and asymptotical behavior \eqref{asy1}. As we are interested in the large $n$ asymptotics of the normalizing coefficients $h_n(x)$, we will solve the $\Gamma^o$-RHP for 
$\Gamma^o(z)=\Gamma^o(z;x,n,n)$.\smallskip

\paragraph{\bf Construction of the g-function.} The purpose of the so--called $g$-function is to normalize the RHP at infinity. This function is analytic off  $\mathcal{B}\subset\mathbb{C}$ which consists of a finite union of oriented smooth arcs, whose endpoints and shape  depend on $x\in\mathbb{C}$. 
We shall present a {\em heuristic} derivation of the $g$ function; this will be used as an Ansatz whose validity is confirmed a posteriori. Much of the underlying logic is well known and has been used repeatedly in the literature.

Suppose that there is a {\em positive} density $\rho(z) \d z$ on $\mathcal B$ such that  (the parametric dependence on $x$ is understood):
\begin{equation}\label{con1}
  g(z) = \int_{\mathcal{B}}\ln(z-w)\rho(w)|\d w|,\hspace{0.25cm}z\in\mathbb{C}\backslash\mathcal{B}\hspace{1cm} g_+(z)+g_-(z) =\theta(z;x)+\ell+i\alpha_j,\hspace{0.5cm}z\in\mathcal{B}_j
\end{equation}
where $\mathcal B_j$ denote the connected components of $\mathcal B$ and  $\ell\in\mathbb{C},\alpha_j\in\mathbb{R}$ can only depend on $x$ and furthermore
\begin{equation}\label{gfuncasy}
  g(z)=\ln z+\mathcal O\left(z^{-1}\right),\hspace{0.5cm}z\rightarrow\infty.
\end{equation}
(The conditions implicitly require that $g(z)$ has a jump $g_+(z) - g_-(z) =2\pi i$ on a contour that extends to infinity and that $\int_{\mathcal B} \rho(z)\d z =1$.)
Assuming temporarily the existence of $g(z)$, differentiating in \eqref{con1} with respect to $z$ and applying the Plemelj formula, we have
\begin{equation*}
  \big(g'(z)\big)^2_+=\big(g'(z)\big)_-+2\pi i\theta_{z}(z;x)\rho(z),\hspace{0.5cm}z\in\mathcal{B},\hspace{1cm} (')=\frac{\partial}{\partial z}
\end{equation*}
which is solved as
\begin{equation}\label{gcurve}
  \big(g'(z)\big)^2 = \int_{\mathcal{B}}\frac{\theta_{w}(w;x)\rho(w)}{w-z}\d w = \theta_{z}(z;x)g'(z)+
  \int_{\mathcal{B}}\frac{\theta_{w}(w;x)-\theta_{z}(z;x)}{w-z}\rho(w)\d w.
\end{equation}
The last integral defines a meromorphic function in $z\in\mathbb{C}$ with its only singularity being a fourth order pole at the origin, thus
\eqref{gcurve} implies for $y(z) = g'(z)-\frac{1}{2}\theta_{z}(z;x)$ that
\begin{equation}	
\label{yfun}
  y^2 = \left(\frac{\theta_{z}}{2}\right)^2 + \int_{\mathcal{B}}\frac{\theta_{w}(w;x)-\theta_{z}(z;x)}{w-z}\rho(w)\d w = \frac{P_6(z;x)}{z^8}
\end{equation}
with a polynomial $P_6(z;x) = z^6+\mathcal O\left(z^5\right),z\rightarrow\infty$. All together
\begin{equation}\label{gfunction}
  g(z) = \frac{1}{2}\theta(z;x)+\int_{z_0}^zy(\lambda)\d\lambda+\frac{\ell}{2},\hspace{0.5cm}z\in\mathbb{C}\backslash\mathcal{B}
\end{equation}
where the choice of the initial point $\lambda=z_0$ in the line integral is related to the topology of the branchcut $\mathcal{B}$. 
Additional properties of the real part of $g(z)$ follow from the requirement that $\rho(z)$ is a positive density, but these will be verified en route.

At this stage of the construction of
$g(z)$, the choice of $x\in\mathbb{C}$ is important. To this end let us from now on treat the $g$-function \eqref{gfunction} as defined on a Riemann surface $X$ of genus
$g\geq 0$. The distinction according to the genus places constraints on the topology of $\mathcal{B}$ or equivalently on the form of $P_6(z;x)$.
\subsection{The concrete g-function for genus zero}\label{ggen0} 
We now assume that $\mathcal B$ consists of a single connected component and thus the genus of the Riemann surface where $y(z)$ \eqref{yfun} is defined is zero.
 Since $g=0$, we must have $P_6(z;x) = \left(P_2(z;x)\right)^2(z-c_1)(z-c_2),c_1\neq c_2$ with 
$P_2(z;x)$ a monic polynomial of $\deg P_2 =2$. We are thus left with four unknowns which are determined by the requirements
\begin{equation}\label{cond}
  y(z) = -\frac{1}{2}\theta_{z}(z;x)+\mathcal O(1),\hspace{0.5cm}z\rightarrow 0;\hspace{1cm}y(z) = \frac{1}{z}+\mathcal O\left(z^{-2}\right),
  \hspace{0.5cm}z\rightarrow\infty.
\end{equation}
One solution to the resulting system is given by
\begin{equation}\label{y1}
  y(z) = \frac{1}{z^4}\left(z^2-\frac{1}{2a}\right)\left(z^2+a^2\right)^{\frac{1}{2}},\hspace{0.5cm}z\in\mathbb{C}\backslash\mathcal{B},
\hspace{0.5cm}\mathcal{B}=[-ia,ia]
\end{equation}
where $a=a(x)$ solves the cubic equation (see \eqref{cubic})
\begin{equation}\label{c1}
  1+2xa^2-4a^3=0,
\end{equation}
subject to the condition 
\begin{equation}\label{c2}
 a= \frac{x}{2}+\mathcal{O}\left(x^{-2}\right)\ \ \textnormal{as}\  x\rightarrow\infty.
\end{equation}
Here, $y(z)$ is analytic on $\mathbb{C}$ with a branchcut $\mathcal{B}$ that extends between
the branchpoints $\pm ia$.  As we need to require that $a=a(x)$ is analytic for sufficiently large $|x|$, identity \eqref{y1} with the latter choice \eqref{c2} of $a$ is in fact the
unique solution to the system with the aforementioned characteristica. All other candidates for $a$ and $y$, in particular the non-symmetric ones corresponding to $c_1+c_2\neq 0$, 
have to be excluded - otherwise the replacement $x\mapsto x e^{2\pi i},|x|>R$, which does not affect the polynomials $\{\psi_n^o(z)\}_{n\geq 0}$, would change the $g$-function below and therefore the large degree asymptotics.\smallskip

Substituting \eqref{y1} into \eqref{gfunction} with $z_0=ia$ yields for $z\in\mathbb{C}\backslash[-ia,ia]$
\begin{equation}\label{gfunczero}
  g(z)=\frac{1}{2}\theta(z;x) +\ln\left(z+\sqrt{z^2+a^2}\right)+\sqrt{z^2+a^2}\left(\frac{z^2(1-6a^3)+a^2}{6a^3\lambda^3}\right)
-\ln(ia)+\frac{\ell}{2}
\end{equation}
where we choose the principal branch for the logarithm and all fractional power exponents and the value of the Lagrange multiplier
\begin{equation}\label{Lagra}
  \ell = 2-\frac{1}{3a^3}+\ln\left(\frac{-a^2}{4}\right)
\end{equation}
follows from comparison of \eqref{gfunczero} with \eqref{gfuncasy} as $z\rightarrow\infty$. We now have to discuss the dependency of \eqref{y1} on the choice of $x\in\mathbb{C}$.
Notice that the only branch points of the cubic equation \eqref{c1} are given by the three points 
\begin{equation}\label{xcrit}
  x_k = -\left(\frac{3}{\sqrt[3]{2}}\right) e^{\frac{2\pi i}{3}k},\hspace{0.5cm}k=0,1,2
\end{equation}
or equivalently, these points (in the complex $x$-plane) correspond (via \eqref{c1}) to the critical situation (in the complex $z$-plane) when the branchpoints $z=\pm ia$ collide with one of the
saddle points $z=\pm \frac{1}{\sqrt{2a}}$. Define the complex effective potential
\begin{eqnarray*}
  \varphi(z) &=& \theta(z;x)-2g(z)+\ell = -2\int_{ia}^{z}y(\lambda)\d\lambda,\hspace{0.5cm}z\in\mathbb{C}\backslash[-ia,ia]\\
&=&-2\ln\left(\frac{z+\sqrt{z^2+a^2}}{ia}\right)+\sqrt{z^2+a^2}\left(\frac{z^2(6a^3-1)-a^2}{3a^3z^3}\right).
\end{eqnarray*}
In terms of this potential, the connecting edges $\partial\Delta$ of the star shaped region $\overline{\Delta}$ as introduced in Definition \ref{defstar} are determined
by the condition that the real part of $\varphi(z)$ at one of the saddle-points vanishes, i.e. 
\begin{equation}\label{b1}
  \Re\,\varphi(z)\Big|_{z=\pm\frac{1}{\sqrt{2a}}} = \pm \Re\left\{-2\ln\left(\frac{1+\sqrt{1+2a^3}}{ia\sqrt{2a}}\right)+
\sqrt{1+2a^3}\left(\frac{4a^3-1}{3a^3}\right)\right\}=0.
\end{equation}
Our subsequent analysis will (a posteriori) show that these three curves determine precisely the transition between the genus zero and genus two situation in the $z$-plane,
respectively the transition from the zero-free region $\mathbb{C}\backslash\overline{\Delta}$ to the zero containing region $\overline{\Delta}$ in the (rescaled) $x$-plane.
\br\label{rem1} In \cite{BM}, the conditions for the boundary edges in the complex $\xi$-plane are stated as
\begin{eqnarray*}
\Re \left\{ 
-\ln  \left( -S+\sqrt {{\frac {3{S}^{3}-4}{3S}}} \right) 
-\frac 1{4}\,{S}^{2}\sqrt {{\frac {3\,{S}^{3}-4}{3S}}}-
\frac 2 {3S} \sqrt {{\frac {3\,{S}^{3}-4}{3S}}}+\ln  \left( {\frac {2}{\sqrt {3S}}} \right)\right\}&=&0 \\
3S^3 + 4 \xi S + 8&=&0
\end{eqnarray*}
and the latter system, under the identifications
\be\nonumber
\xi = (12)^{\frac 13} x\ ,\ \  S =  -\le(\frac 23 \ri)^\frac 1 3 \frac 1 a,
\ee
is identical to \eqref{b1},\eqref{c1}. Also in the notation of \cite{BM} with the latter identifications
\begin{equation*}
  \mathfrak c \le(-\le(\frac 23 \ri)^\frac 1 3 \frac 1 a\ri) =2 \varphi\le( \frac{1}{\sqrt{2a}}\ri),
\end{equation*}
hence the star shaped region of Figure \ref{star} is, up to a rescaling, identical to the one shown in Figure 16 in \cite{BM}.
\er
We finish our discussion of the genus zero case by depicting the branchcut $\mathcal{B}$ and various sign properties of $\varphi(z)$: To this end assume that 
$x\in\mathbb{C}:\textnormal{dist}(x,\overline{\Delta}=\Delta\cup\partial\Delta)\geq\delta>0$, i.e. $x$ is chosen from the unbounded domain and we stay away from
the edges and vertices. For such $x$
the support $\mathcal{B}$ is determined implicitly via \eqref{con1} and \eqref{gfunczero}. In Figure \ref{gen0carousel} the branch cut $\mathcal{B}$ is indicated in red for several choices
$x\in\mathbb{C}:\textnormal{dist}(x,\overline{\Delta})\geq\delta>0$ in the complex $z$-plane. The orientation is such that the $(-)$ side extends to the unbounded component:
\begin{figure}
\begin{center}
\resizebox{0.6\textwidth}{!}{\includegraphics{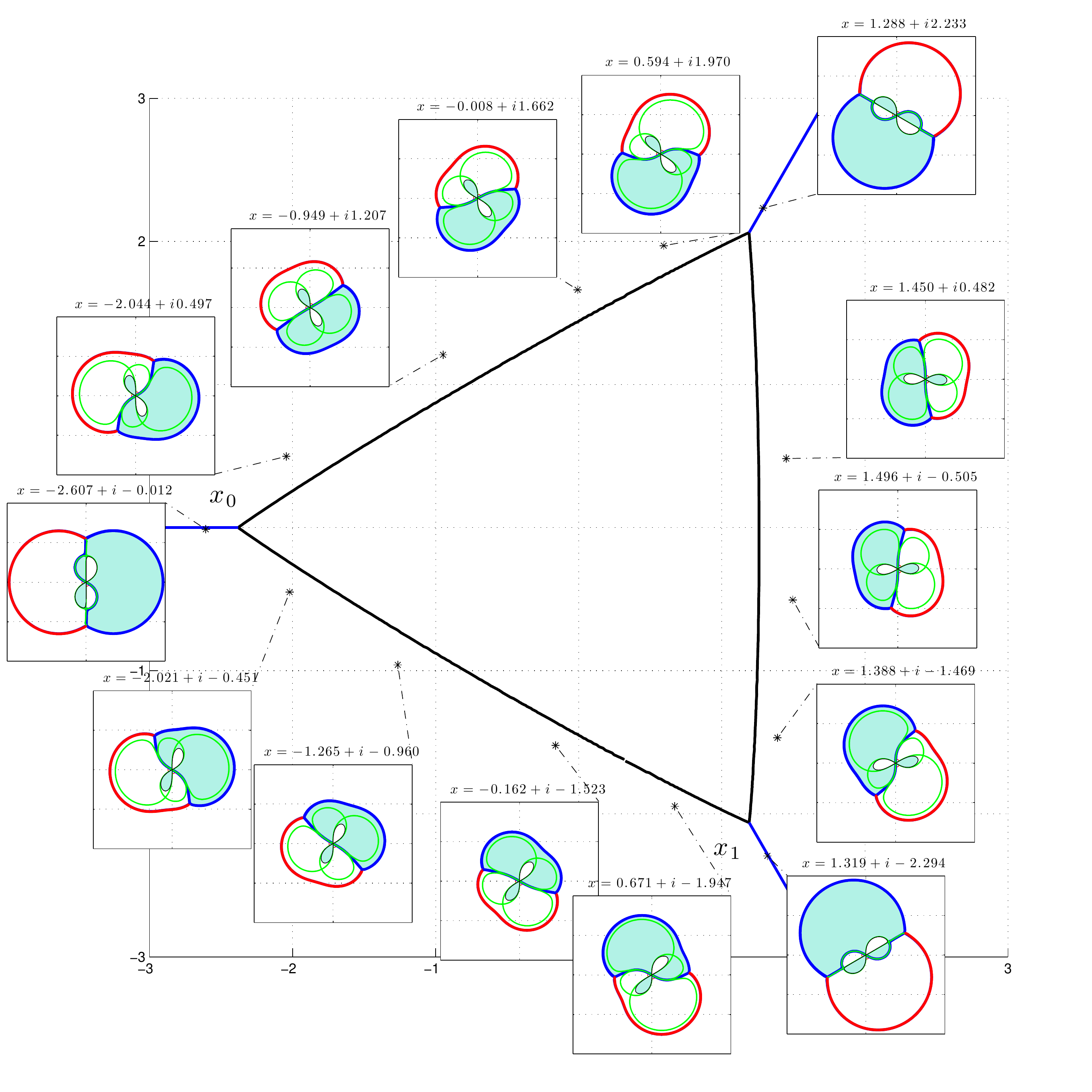}}
\caption{We plot the branch cut $\mathcal{B}$ in red for several choices $x\in\mathbb{C}:\textnormal{dist}(x,\overline{\Delta})\geq\delta>0$. The level sets $\Re\varphi(z)=0$
are shown as solid blue lines and the  shaded regions resemble the components were $\Re\varphi(z)>0$. In the white shaded regions we have $\Re\varphi(z)<0$ and along the green
lines $\Re\varphi(z)\equiv\Re\varphi(\pm(2a)^{-\frac{1}{2}})$.}
\label{gen0carousel}
\end{center}
\end{figure}
\subsection{The concrete g-function for genus two}\label{ggen2} If $g=2$, we have $P_6(z;x)\equiv R(z) = \prod_{k=1}^6(z-a_k)$ where $a_j\neq a_k$ for $j\neq k$. This means we are
working with the hyperelliptic curve
\begin{equation}\label{hyperell}
 X=\big\{(z,w):\,w^2=R(z)\big\};\hspace{0.5cm}\mathcal{B}=\bigcup_{k=1}^3[a_{2k-1},a_{2k}]
\end{equation}
for which we use the representation as two-sheeted covering of the Riemann sphere $\mathbb{C}\mathbb{P}^1$, obtained by glueing together two copies of $\mathbb{C}\backslash\mathcal{B}$ along
$\mathcal{B}$ in the standard way. For future purposes, we let $\sqrt{R(z)}\sim z^3$ as $z\rightarrow\infty^+$ on the first sheet, and $\sqrt{R(z)}\sim -z^3$ as $z\rightarrow
\infty^-$ on the second sheet. As our subsequent analysis shows, we can consider the symmetric choice
\begin{equation}\label{ygen2}
 y(z)=\frac{\sqrt{R(z)}}{z^4},\hspace{0.5cm}z\in\mathbb{C}\backslash\mathcal{B}
\end{equation}
with
\begin{equation*}
  a_1=ia,\ \ \ a_2=ib,\ \ \ a_3=ic;\hspace{0.5cm}a_{k+3}=-a_k,\ \ k=1,2,3.
\end{equation*}
Here, the points $a=a(x),b=b(x),c=c(x)\in\mathbb{C}$ are determined implicitly from \eqref{cond}, i.e. they satisfy
\begin{equation}\label{gen2con1}
  abc=-\frac{1}{2},\hspace{0.5cm}a^2b^2+a^2c^2+b^2c^2=-\frac{x}{2},
\end{equation}
and in addition from the requirement that the level curves
\begin{equation}\label{gen2con2}
 \Re\left(\int_{a_1}^zy(\lambda)\d\lambda\right)\equiv 0
\end{equation}
are connecting the branchpoints where integration is always carried out on the first sheet of $X$ without crossing the branchcut $\mathcal{B}$. We notice that \eqref{gen2con1}
yields two complex equations for the three (complex) unknowns $a,b$ and $c$. However \eqref{gen2con2}, the Boutroux condition, gives another set of two real conditions: by symmetry and since
the residue of the meromorphic differential $\d\phi=y(z)\d z$ at the origin as well as at the two copies of infinity is already real-valued, we can state \eqref{gen2con2} equivalently
as
\begin{equation}\label{gen2con3}
  \Re\left(\oint_{\mathcal{A}_1}\d\phi\right)=0,\hspace{1cm}\Re\left(\oint_{\mathcal{B}_1}\d\phi\right)=0.
\end{equation}
Here we use the cycles $\{\mathcal{A}_j,\mathcal{B}_j\}_{j=1}^2$ which form a basis of the homology group of $X$, see Figure \ref{modRHP1}. Imposing \eqref{gen2con1}, we obtain
\begin{equation}
 R(z)=z^6+z^4\big(a^2+b^2+c^2\big)-z^2\left(\frac{x}{2}\right)+\frac{1}{4}
 \label{R(z)}
\end{equation}
in which the coefficient of $\mathcal{O}(z^4)$ is still undetermined. We write this coefficient as
\begin{equation*}
 s+it = a^2+b^2+c^2=-\sum_{j=1}^3a_j^2,\hspace{0.5cm}s,t\in\mathbb{R}
\end{equation*}
and study the mapping
\begin{equation*}
 (s,t)\stackrel{\phi}{\mapsto}\big(I_1(s,t),I_2(s,t)\big);\hspace{0.5cm}I_1(s,t)=\Re\left(\oint_{\mathcal{A}_1}\d\phi\right),\
 \ I_2(s,t)=\Re\left(\oint_{\mathcal{B}_1}\d\phi\right).
\end{equation*}
Since the functions $I_1,I_2$ only depend on the homology classes, we can compute the Jacobian of the mapping $\phi:\mathbb{R}^2\rightarrow\mathbb{R}^2$ as (cf. \cite{FK})
\begin{equation}\label{Jac}
  \det\begin{bmatrix}
       \partial_s I_1 &\partial_t I_1\\
\partial_s I_2 &\partial_t I_2\\
      \end{bmatrix}=-\frac{1}{4}\det\begin{bmatrix}
\Re\left(\oint_{\mathcal{A}_1}\eta_1\right) & \Im\left(\oint_{\mathcal{A}_1}\eta_1\right)\\
\Re\left(\oint_{\mathcal{B}_1}\eta_1\right) & \Im\left(\oint_{\mathcal{B}_1}\eta_1\right)\\
\end{bmatrix}=-\frac{1}{4}\Im\left(\,\overline{\oint_{\mathcal{A}_1}\eta_1}\cdot\oint_{\mathcal{B}_1}\eta_1\,\right)\neq 0
\end{equation}
which is valid for all $x\in\mathbb{C}:\textnormal{dist}(x,\mathbb{C}\backslash\overline{\Delta})\geq\delta>0$ and were we used the holomorphic differential 
$\eta_1$ written in \eqref{forms}. But \eqref{gen2con1},\eqref{gen2con3} can be solved for $x=0$ as
\begin{equation*}
  a_{1,0}=ia_0=\frac{1}{\sqrt[3]{2}}e^{-i\frac{5\pi}{6}},\hspace{0.5cm} a_{2,0}=ib_0=\frac{1}{\sqrt[3]{2}}e^{i\frac{5\pi}{6}},\hspace{0.5cm} a_{3,0}=ic_0=\frac{1}{\sqrt[3]{2}}e^{i\frac{\pi}{2}},
\end{equation*}
and we have
\begin{equation*}
  s+it\Big|_{a_j=a_{j,0}}=-\sum_{j=1}^3a_{j,0}^2 =0.
\end{equation*}
Hence the non-vanishing of the Jacobian \eqref{Jac} implies (by implicit function theorem) that this solution can be extended uniquely to nonzero $x$ inside the star. Now given
$(s,t)$ corresponding to $x\neq 0$ inside the star, we determine the branchpoints $\pm ia,\pm ib,\pm ic$ from the system
\begin{equation*}
  E_1\equiv abc=-\frac{1}{2},\hspace{0.75cm}E_2\equiv a^2b^2+a^2c^2+b^2c^2=-\frac{x}{2},\hspace{0.75cm}E_3\equiv a^2+b^2+c^2=s+it.
\end{equation*}
These are three equations for the three unknowns, with underlying Jacobian
\begin{equation*}
 \det\frac{\partial(E_1,E_2,E_3)}{\partial(a,b,c)} = 4\big(b^2-c^2\big)\big(a^2-c^2\big)\big(a^2-b^2\big)=4\!\!\!\!\!\prod_{1\leq j<k\leq 3}\!\!\!\big(a_k^2-a_j^2\big)
\end{equation*}
which does not vanish in the genus two case. Thus \eqref{gen2con1},\eqref{gen2con3} determine the branchpoints uniquely as long as we impose the genus two validity. The branchpoints at hand, the $g$-function is now given as in \eqref{gfunction}, i.e.
\begin{equation}\label{gfunctwo}
 g(z) = \frac{1}{2}\theta(z;x)+\int_{a_1}^zy(\lambda)\d\lambda+\frac{\ell}{2},\hspace{0.5cm}z\in\mathbb{C}\backslash\mathcal{B}
\end{equation}
with the Lagrange multiplier equal to
\begin{equation}\label{Lagra2}
  \ell = 2\ln a_1-2\int_{a_1}^{\infty^+}\left(y(\lambda)-\frac{1}{\lambda}\right)\d\lambda.
\end{equation}
The nonlinear steepest descent analysis carried out in Sections \ref{secgen2} and \ref{secgen22} below shows that \eqref{gfunctwo} is precisely the correct $g$-function for the analysis inside the
star shaped region, i.e. for $x\in\mathbb{C}:\textnormal{dist}(x,\mathbb{C}\backslash\overline\Delta)\geq\delta>0$. For such $x$ several level curves of the effective potential
\begin{equation*}
  \varphi(z) = \theta(z;x)-2g(z)-\ell = -2\int_{a_1}^z\d\phi,\hspace{0.5cm}z\in\mathbb{C}\backslash\mathcal{B}
\end{equation*}
are shown in Figure \ref{gen2carousel}.
\begin{figure}[tbh]
\begin{center}
\resizebox{0.5\textwidth}{!}{\includegraphics{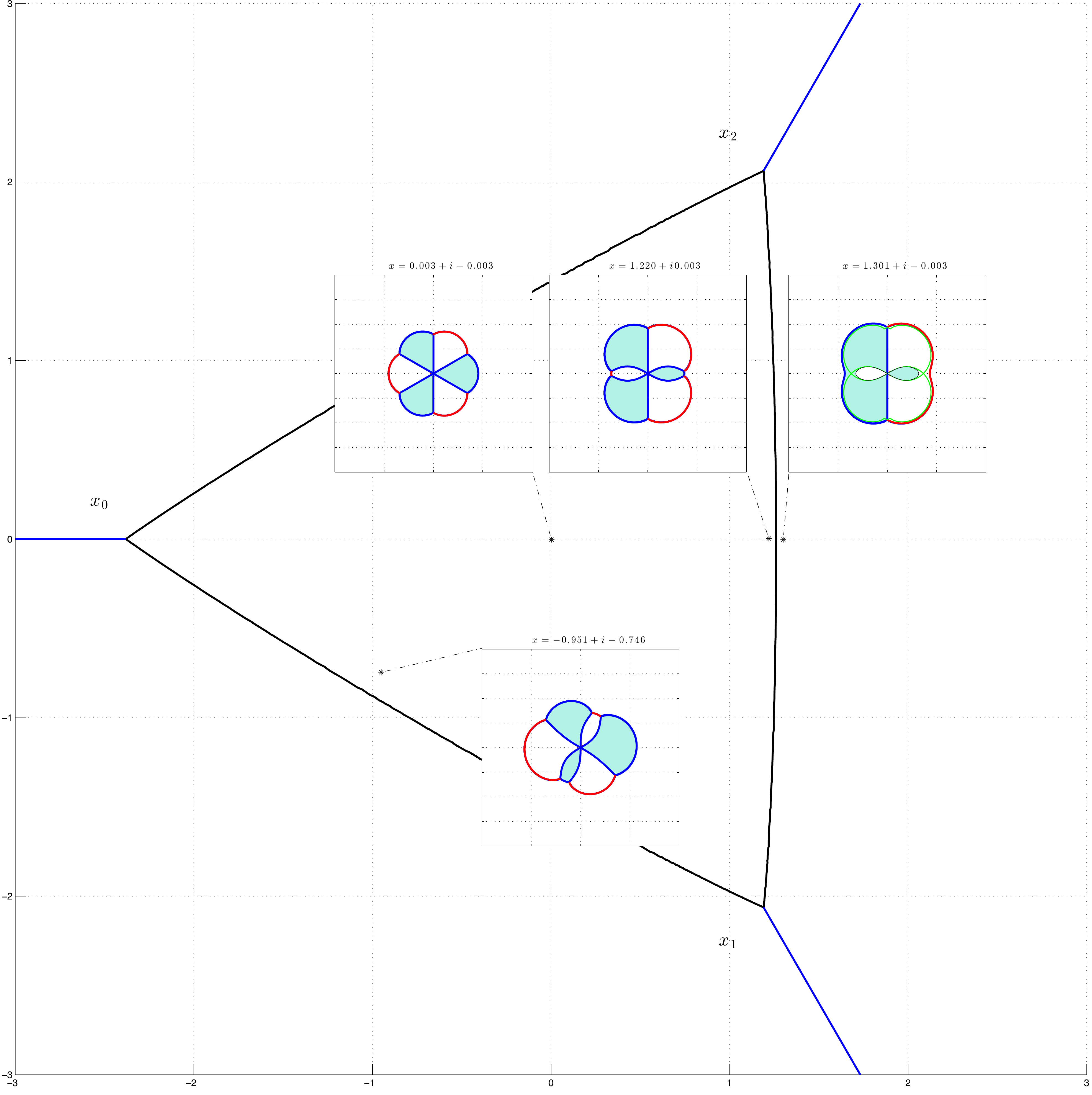}}
\caption{We plot the branch cut $\mathcal{B}$ in red for several choices $x\in\mathbb{C}:\textnormal{dist}(x,\mathbb{C}\backslash\overline{\Delta})\geq\delta>0$. The level sets $\Re\varphi(z)=0$
are shown as solid blue lines and the  shaded regions resemble the components were $\Re\varphi(z)>0$. In the white shaded regions we have $\Re\varphi(z)<0$.}
\label{gen2carousel}
\end{center}
\end{figure}

 At this point we have enough information to move on to the next transformation in the nonlinear steepest descent analysis.
\section{Riemann-Hilbert analysis - construction of parametrices}\label{sec3}
The $g$-functions derived in Subsections \ref{ggen0} and \ref{ggen2} are used to normalize the RHP for $\Gamma^o(z;x,n,n)$ in the spectral variable $z$ at infinity, depending
on whether $x$ lies outside the star shaped region or inside. This eventually reduces the global solution of the RHPs to the construction of local model functions (parametrices)
which are standard near the branchpoints. We emphasize the existence or non-existence of the outer parametrix.
\subsection{Genus zero parametrices}\label{RHP0} Let $x\in\mathbb{C}:\textnormal{dist}(x,\overline{\Delta})\geq\delta>0$, i.e. away from the edges and vertices of the star shaped region. 
Before we employ the $g$-function transformation, we first deform the original jump contour $\gamma$ to a contour which passes through
the branchpoints $\pm ia$, which on one side follows $\mathcal{B}$ and on the other side lies inside the shaded region and again connects the two branch points. We denote the latter part of the
jump contour with $\mathcal{L}$, see Figure \ref{deform1} below for one possible choice. Such a contour deformation is always possible since $w^o(z;x)$ is analytic away from
the origin.
\begin{figure}[tbh]
\begin{center}
\resizebox{0.3\textwidth}{!}{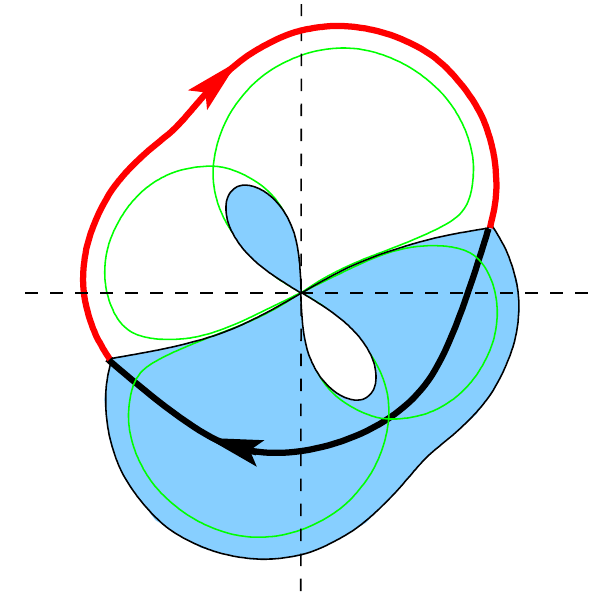}\hspace{2cm} \resizebox{0.3\textwidth}{!}{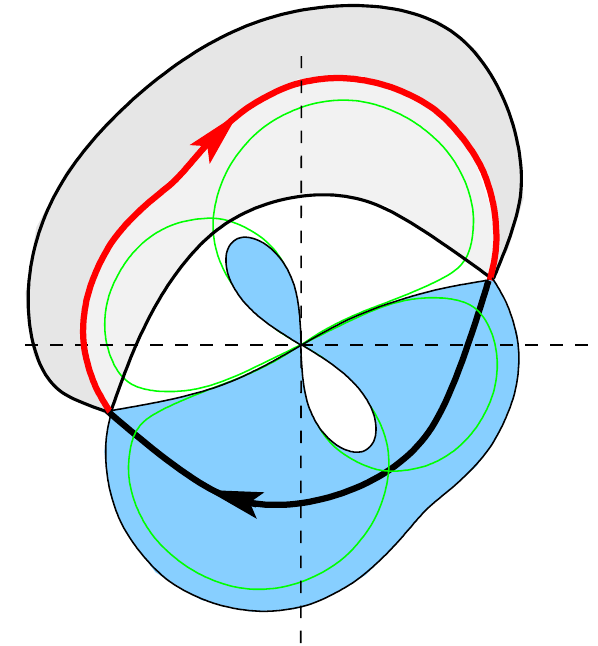}
\begin{minipage}{0.5\textwidth}
\caption{Deformation of the jump contour $\gamma$ to the union of $\mathcal{B}\cup\mathcal{L}$. 
The branchcut $\mathcal{B}$ is indicated in red and $\mathcal{L}$ in black. The picture corresponds to one possible choice of $x\in\mathbb{C}:\textnormal{dist}(x,\overline{\Delta})\geq\delta>0$
with $\Re x<0$ and $\Im x>0$.}
\label{deform1}
\end{minipage}
\begin{minipage}{0.49\textwidth}
\caption{Opening of lenses in genus zero. The contours $\mathcal{B}^{\pm}$ are given the same orientation as $\mathcal{B}$.}
\label{open1}
\end{minipage}
\end{center}
\end{figure}

Now introduce
\begin{equation*}
  Y(z) = \exp\left[-\frac{n\ell}{2}\sigma_3\right]\Gamma^o(z)\exp\left[-n\left(g(z)-\frac{\ell}{2}\right)\sigma_3\right],\hspace{0.5cm}z\in\mathbb{C}\backslash\mathcal{B}
\end{equation*}
where $g(z)$ is given in \eqref{gfunczero} and the Lagrange multiplier in \eqref{Lagra}. Recalling \eqref{con1} (here in genus zero case with $\alpha=0$) we are lead to the
following RHP
\begin{itemize}
 \item $Y(z)$ is analytic for $z\in\mathbb{C}\backslash\left(\mathcal{B}\cup\mathcal{L}\right)$
 \item On the clockwise oriented contour $\mathcal{L}\cup\mathcal{B}$ as shown in Figure \ref{deform1}
  \begin{eqnarray*}
    Y_+(z) &=& Y_-(z)\begin{bmatrix}
				 e^{-n(g_+(z)-g_-(z))} & (2\pi i z)^{-1}\\
  0 & e^{n(g_+(z)-g_-(z))}
                               \end{bmatrix},\hspace{0.5cm}z\in\mathcal{B}\\
  Y_+(z) &=& Y_-(z)\begin{bmatrix}
                                1 & (2\pi iz)^{-1}e^{-n\varphi(z)}\\
  0 & 1\\
                               \end{bmatrix},\hspace{0.5cm}z\in\mathcal{L}
  \end{eqnarray*}
  \item As $z\rightarrow\infty$, we see from \eqref{gfuncasy} that
  \begin{equation*}
    Y(z) = I+\mathcal O\left(z^{-1}\right)
  \end{equation*}
\end{itemize}
As we have $\Re\varphi(z)>0$ in the  shaded regions, one concludes
\begin{equation}\label{esti1}
 \begin{bmatrix}
                                1 & (2\pi iz)^{-1}e^{-n\varphi(z)}\\
  0 & 1\\
                               \end{bmatrix}\longrightarrow I,\hspace{0.5cm}n\rightarrow\infty
\end{equation}
where the convergence is exponentially fast for $z\in\mathcal{L}$ away from the branchpoints $z=\pm ia$. On the other hand
\begin{equation}\label{capitalG}
 G(z)=g_+(z)-g_-(z),\hspace{0.5cm}z\in\mathcal{B}
\end{equation}
admits local analytical continuation into the bounded and unbounded white shaded regions (compare Figure \ref{deform1}). In fact with \eqref{con1} on the $(-)$ side
\begin{equation*}
  G(z) = -2g_-(z)+\theta(z;x)+\ell=\varphi_-(z),\hspace{0.5cm}z\in\mathcal{B}
\end{equation*}
and on the $(+)$ side
\begin{equation*}
 G(z) = 2g_+(z)-\theta(z;x)-\ell=-\varphi_+(z),\hspace{0.5cm}z\in\mathcal{B}.
\end{equation*}
These continuations allow us to factorize the jump on $\mathcal{B}$
\begin{eqnarray*}
  \begin{bmatrix}
    e^{-nG(z)} & (2\pi iz)^{-1}\\
  0 & e^{nG(z)}\\
  \end{bmatrix}&=&\begin{bmatrix}
1 & 0\\
2\pi iz e^{n\varphi_-(z)} & 1\\
\end{bmatrix}\begin{bmatrix}
  0 & (2\pi iz)^{-1}\\
-2\pi iz & 0\\
\end{bmatrix}\begin{bmatrix}
1 & 0\\
2\pi iz e^{n\varphi_+(z)}&1\\
\end{bmatrix}\\
&=&S_{L_1}(z)S_P(z)S_{L_2}(z)
\end{eqnarray*}
and open lenses: We depicted the contours $\mathcal{B}^{\pm}$ in Figure \ref{open1} and introduce
\begin{equation}\label{openup}
  S(z) = \begin{cases}
                Y(z) S_{L_1}(z), &\ z\in\mathcal{L}_1\\
		Y(z) S_{L_2}^{-1}(z),&\ z\in\mathcal{L}_2\\
		Y(z),&\ \textnormal{else}.
               \end{cases}
\end{equation}

This opening leads to jumps on the lense boundaries $\mathcal{B}^{\pm}$
\begin{equation*}
 S_+(z)=S_-(z)\begin{bmatrix}
               1 & 0\\
2\pi ize^{n\varphi(z)} & 1\\
              \end{bmatrix},\hspace{0.5cm}z\in\mathcal{B}^{\pm}
\end{equation*}
as well as on the contours $\mathcal{B}\cup\mathcal{L}$
\begin{equation*}
  S_+(z) = S_-(z)\begin{bmatrix}
                              0 & (2\pi iz)^{-1}\\
-2\pi iz & 0\\
                             \end{bmatrix},\ \ z\in\mathcal{B};\hspace{0.5cm} S_+(z) = S_-(z)\begin{bmatrix}
1 & (2\pi iz)^{-1}e^{-n\varphi(z)}\\
0 & 1\\
\end{bmatrix},\ \ z\in\mathcal{L}.
\end{equation*}
However $\Re\varphi(z)<0$ in the white shaded regions, thus
\begin{equation}\label{esti2}
 \begin{bmatrix}
  1 & 0\\
2\pi i z e^{n\varphi(z)}& 1\\
 \end{bmatrix}\longrightarrow I,\hspace{0.5cm}n\rightarrow\infty
\end{equation}
again exponentially fast for $z\in\mathcal{B}^{\pm}$ away from the branchpoints $z=\pm ia$. The latter \eqref{esti2} combined with \eqref{esti1}, we therefore have to focus
on the local contributions arising from the contour $\mathcal{B}$ and the neighborhood of the branchpoints $z=\pm ia$:\bigskip

Define the outer parametrix $M=M(z;x)$ as
\begin{equation}\label{outer}
  M(z) = (2\pi i)^{-\frac{1}{2}\sigma_3}\left(\frac{a}{2}\,\right)^{-\frac{1}{2}\sigma_3}\big(\delta(z)\big)^{-\sigma_2}\mathcal{D}(z)^{\sigma_3}(2\pi i)^{\frac{1}{2}\sigma_3},\hspace{0.5cm}z\in\mathbb{C}\backslash\mathcal{B}
\end{equation}
where the scalar Szeg\"o function is given by
\begin{equation*}
  \mathcal{D}(z) = \exp\left[\frac{\sqrt{z^2+a^2}}{2\pi i}\int_{ia}^{-ia}\frac{\ln(w)}{\sqrt{w^2+a^2}_+}\frac{\d w}{w-z}\right]=\sqrt{a}
\left(\frac{\sqrt{z^2+a^2}-a}{\sqrt{z^2+a^2}+z}\right)^{\frac{1}{2}}
\end{equation*}
with principal branches for all fractional power functions and
\begin{equation*}
  \delta(z) = \left(\frac{z-ia}{z+ia}\right)^{\frac{1}{4}}\longrightarrow 1,\hspace{0.5cm}z\rightarrow\infty
\end{equation*}
is analytic on $\mathbb{C}\backslash\mathcal{B}$. One checks readily that \eqref{outer} is analytic on $\mathbb{C}\backslash\mathcal{B}$, square integrable up to the boundary and
\begin{equation*}
  M_+(z)=M_-(z)\begin{bmatrix}
                            0 & (2\pi iz)^{-1}\\
-2\pi iz& 0
                           \end{bmatrix},\ \ z\in\mathcal{B};\hspace{1cm}M(z)\longrightarrow I,\hspace{0.5cm}z\rightarrow\infty.
\end{equation*}
Hence the outer parametrix $M=M(z;x),z\in\mathbb{C}\backslash\mathcal{B}$ exists for all $x\in\mathbb{C}:\textnormal{dist}(x,\overline{\Delta})\geq\delta>0$.\bigskip

The inner parametrices near the branchpoints are standard objects in the Deift-Zhou framework since they are constructed out of Airy-functions, see e.g. \cite{DKMVZ}.
We briefly state the final formulae in this subsection and summarize other necessary details in Appendix \ref{app1}. All constructions are motived from the local expansions 
\begin{eqnarray}
  \varphi(z) &=& c_0(z-ia)^{\frac{3}{2}}\big(1+\mathcal O(z-ia)\big),\hspace{0.5cm}z\rightarrow ia,\ \ z\in\mathcal{B}^+\cup\mathcal{B}^-\label{an}\\
  \varphi(z) &=&-2\pi i+\hat{c}_0(z+ia)^{\frac{3}{2}}\big(1+\mathcal O(z+ia)\big),\hspace{0.5cm}z\rightarrow -ia,\ \ z\in\mathcal{B}^+\cup\mathcal{B}^-\label{man}
\end{eqnarray}
where the function $(z+ia)^{\frac{3}{2}}$ is defined for $z\in\mathbb{C}\backslash(-\infty,-ia]$, i.e. with a branchcut to the left of $-ia$ and
$(z-ia)^{\frac{3}{2}}$ for $z\in\mathbb{C}\backslash[ia,\infty)$, i.e. with a branchcut to the right of $ia$. Specifically the parametrix $U(z)$ near $z=-ia$ is given as
\begin{equation}\label{mapara}
  U(z) = B_U(z)\big(-i\sqrt{\pi}\big)A^{RH}\big(\zeta(z)\big)e^{\frac{2}{3}\zeta^{3/2}(z)\sigma_3}(2\pi iz)^{\frac{1}{2}\sigma_3},
  \hspace{0.5cm}|z+ia|<r
\end{equation}
where $A^{RH}(\zeta)$ is defined in \eqref{pl:3}, we use the locally analytic (compare \eqref{man}) change of variables
\begin{equation*}
  \zeta(z) = \left(\frac{3N}{4}\right)^{\frac{2}{3}}\Big(-2g(z)+\theta(z;x)+\ell+2\pi i\Big)^{\frac{2}{3}},\hspace{0.5cm}|z+ia|<r
\end{equation*}
and the multiplier $B_U(z)$ equals
\begin{equation*}
  B_U(z) = M(z)(2\pi iz)^{-\frac{1}{2}\sigma_3}\begin{bmatrix}
                                                          -i & i\\
1 & 1\\
                                                         \end{bmatrix}\zeta^{-\frac{1}{4}\sigma_3}(z).
\end{equation*}
By construction, $B_U(z)$ can have at worst a singularity of square root type at $z=-ia$, however for $z\in\mathcal{B}$ close to $z=-ia$,
\begin{equation*}
  \big(B_U(z)\big)_+ = M_-(z)\begin{bmatrix}
                                          0 & (2\pi iz)^{-1}\\
-2\pi iz & 0\\
                                         \end{bmatrix}(2\pi iz)^{-\frac{1}{2}\sigma_3}\begin{bmatrix}
-i & i\\
1 & 1\\
\end{bmatrix}\zeta_-^{-\frac{1}{4}\sigma_3}(z)e^{-i\frac{\pi}{2}\sigma_3}=\big(B_r(z)\big)_-.
\end{equation*}
Thus the singularity has to be removable and $B_U(z)$ is in fact analytic in a neighborhood of $z=-ia$. We now easily check that the behavior of $A^{RH}(\zeta)$, see Figure
\ref{Airy1}, implies jumps as depicted in Figure \ref{Upara} for $U(z)$. Here the jump contours can always be locally deformed to match the local contours in the $S$-RHP
near the branchpoints.
\begin{figure}[tbh]
\begin{flushleft}
\hspace{2.5cm}\resizebox{0.4\textwidth}{!}{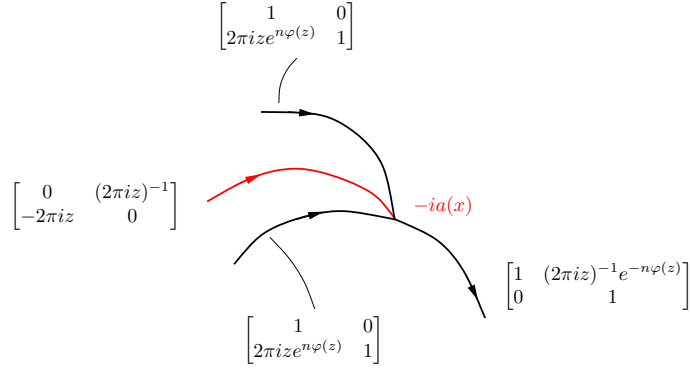}
\caption{Jump behavior of $U(z)$ near $z=-ia$}
\label{Upara}
\end{flushleft}
\end{figure}

Also, as $n\rightarrow\infty$ (hence $|\zeta|\rightarrow\infty$), the two model functions $M(z)$ and $U(z)$ satisfy the desired matching condition, i.e.
\begin{eqnarray}
  U(z) &=& M(z)(2\pi iz)^{-\frac{1}{2}\sigma_3}\Bigg\{I+\frac{1}{48\zeta^{3/2}}\begin{bmatrix}
	1 & 6i\\
	6i & -1\\
	\end{bmatrix}+\mathcal O\left(\zeta^{-6/2}\right)\Bigg\}(2\pi iz)^{\frac{1}{2}\sigma_3}\nonumber\\
& =& \left(I+\mathcal O\left(n^{-1}\right)\right)M(z),\hspace{0.5cm}n\rightarrow\infty\label{match1}
\end{eqnarray}
valid for $x\in\mathbb{C}:\textnormal{dist}(x,\overline{\Delta})\geq\delta>0$ and for all $z\in\mathbb{C}$ such that $0<r_1\leq|z+ia|\leq r_2<\frac{\delta}{2}$.\bigskip

The remaining parametrix near $z=ia$ is introduced along the same lines. We take
\begin{equation}\label{apara}
  V(z) = B_V(z)i\sqrt{\pi}\tilde{A}^{RH}\big(\zeta(z)\big)e^{\frac{2}{3}i\zeta^{3/2}(z)\sigma_3}(2\pi iz)^{\frac{1}{2}\sigma_3},
  \hspace{0.5cm}|z-ia|<r
\end{equation}
with the multiplier
\begin{equation*}
  B_V(z)=M(z)(2\pi iz)^{-\frac{1}{2}\sigma_3}\begin{bmatrix}
                                                        -i & -i\\
1 & -1\\
                                                       \end{bmatrix}\left(e^{-i\pi}\zeta(z)\right)^{\frac{1}{4}\sigma_3},
\end{equation*}
the change of variables
\begin{equation*}
  \zeta(z) = e^{i\pi}\left(\frac{3N}{4}\right)^{\frac{2}{3}}\Big(-2g(z)+\theta(z;x)+\ell\Big)^{\frac{2}{3}},\hspace{0.5cm}|z-ia|<r
\end{equation*}
and the function $\tilde{A}^{RH}$ is given in \eqref{tildeARH}. Also here $B_V(z)$ is analytic near $z=ia$ since
\begin{equation*}
  \big(B_V(z)\big)_+ = M_-(z)\begin{bmatrix}
                                          0 & (2\pi iz)^{-1}\\
-2\pi iz& 0\\
                                         \end{bmatrix}(2\pi iz)^{-\frac{1}{2}\sigma_3}\begin{bmatrix}
-i & -i\\
1 & -1\\
\end{bmatrix}\big(e^{-i\pi}\zeta(z)\big)^{\frac{1}{4}\sigma_3}_-e^{-i\frac{\pi}{2}\sigma_3} = \big(B_V(z)\big)_-
\end{equation*}
but the singularity can be at worst of square root type. Thus $V(z)$ has jumps as in Figure \ref{Vpara} and we have the matching relation
\begin{eqnarray}
  V(z)&=&M(z)(2\pi iz)^{-\frac{1}{2}\sigma_3}\Bigg\{I+\frac{i}{48\zeta^{3/2}}\begin{bmatrix}
	-1 & 6i\\
	6i & 1\\
	\end{bmatrix}+\mathcal O\left(\zeta^{-6/2}\right)\Bigg\}(2\pi iz)^{\frac{1}{2}\sigma_3}\nonumber\\
&=&\left(I+\mathcal O\left(n^{-1}\right)\right)M(z),\hspace{0.5cm}n\rightarrow\infty\label{match2}
\end{eqnarray}
valid for $x\in\mathbb{C}:\textnormal{dist}(x,\overline{\Delta})\geq\delta>0$ and for all $z$ such that $0<r_1\leq|z+ia|\leq r_2<\frac{\delta}{2}$. This completes the 
construction of all relevant parametrices in the genus zero case.
\begin{figure}[tbh]
\begin{flushleft}
\hspace{3cm}\resizebox{0.4\textwidth}{!}{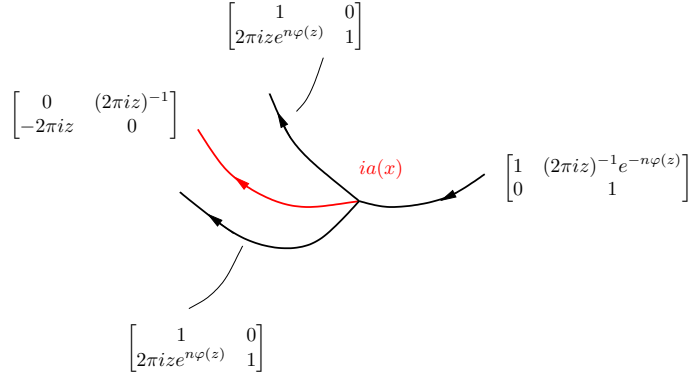}
\caption{Jump behavior of $V(z)$ near $z=ia$}
\label{Vpara}
\end{flushleft}
\end{figure}
\subsection{Genus two parametrices}\label{secgen2} 
Let $x\in\mathbb{C}:\textnormal{dist}(x,\mathbb{C}\backslash\overline{\Delta})\geq\delta>0$ throughout, i.e. we are inside the star shaped
region but stay away from the edges and vertices. Again, we first deform the original jump contour $\gamma$ to a contour which passes through all branchpoints $z=a_j,j=1,\ldots,6$, which on one side follows along the branchcut $\mathcal{B}$ and on the other side lies inside the  shaded region, see Figure \ref{deform2} for a possible choice
\begin{figure}[tbh]
\begin{center}
\resizebox{0.3\textwidth}{!}{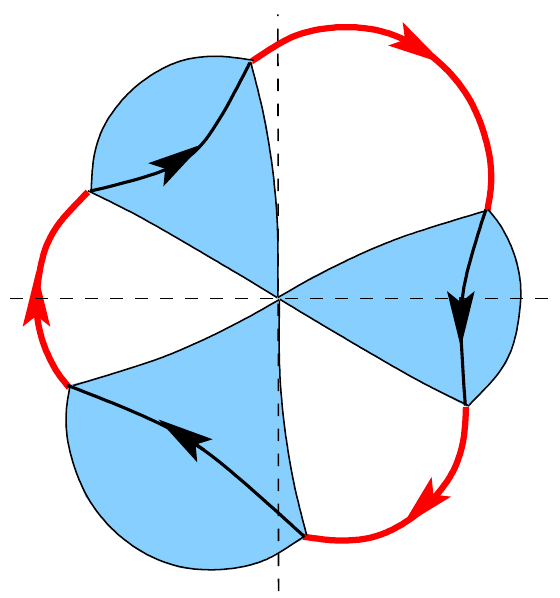}\hspace{2cm} \resizebox{0.3\textwidth}{!}{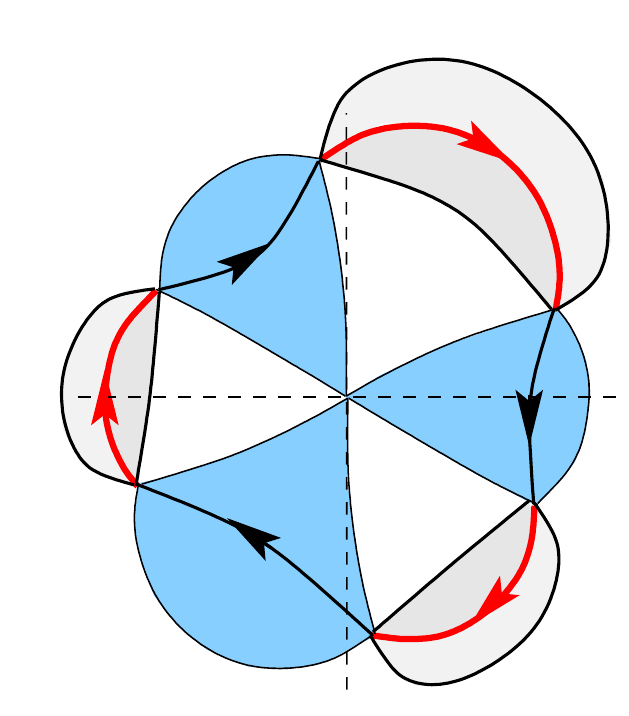}
\begin{minipage}{0.5\textwidth}
\caption{Deformation of the jump contour $\gamma$ to the union of $\mathcal{B}\cup\mathcal{L}$. 
The branchcuts $\mathcal{B}$ are indicated in red and $\mathcal{L}$ in black. The picture corresponds to one possible choice of $x\in\mathbb{C}:
\textnormal{dist}(x,\mathbb{C}\backslash\overline{\Delta})\geq\delta>0$ with $\Re x>0,\Im x>0$.}
\label{deform2}
\end{minipage}
\begin{minipage}{0.45\textwidth}

\caption{Opening of lenses in genus two. We give $\mathcal{B}_j^{\pm}$ the same orientation as $\gamma_j$.}
\label{open2}

\end{minipage}
\end{center}
\end{figure}

We will denote the segments of the deformed contour as follows
\begin{enumerate}
 \item 	The branchcuts $(a_{2j-1},a_{2j}),j=1,2,3$ whose union equals $\mathcal{B}$ are denoted by $\gamma_j$
  \item The gaps $(a_{2j},a_{2j+1}),j=1,2$ are denoted by $\epsilon_j$
  \item The gap $(a_6,a_1)$ is denoted by $\epsilon_0$
\end{enumerate}
With these, the $g$-function transformation
\begin{equation*}
  Y(z) = \exp\left[-\frac{n\ell}{2}\sigma_3\right]\Gamma^o(z)\exp\left[-n\left(g(z)-\frac{\ell}{2}\right)\sigma_3\right],\hspace{0.5cm}z\in\mathbb{C}\backslash\mathcal{B}
\end{equation*}
with \eqref{gfunctwo} and \eqref{Lagra2} transforms the initial RHP to the following one 
\begin{problem}
\label{RHPY} 
Find a $2\times 2$ matrix valued function $Y(z;x)$ such that 
\begin{itemize}
 \item $Y(z)$ is analytic for $z\in\mathbb{C}\backslash(\mathcal{B}\cup \mathcal{L})$
  \item We have jumps
 \begin{eqnarray*}
    Y_+(z) &=& Y_-(z) \begin{bmatrix}
e^{-nG(z)} & (2\pi iz)^{-1}e^{in\alpha_{j-1}} \\
0 & e^{nG(z)} \\
\end{bmatrix},\hspace{0.5cm}z\in\gamma_j,\ j=1,2,3\\
Y_+(z)&=&Y_-(z)\begin{bmatrix}
                            e^{-nG(z)} & (2\pi iz)^{-1}e^{-n\varphi(z)}\\
0 & e^{nG(z)}
                           \end{bmatrix},\hspace{0.5cm}z\in \epsilon_j,\ j=0,1,2
 \end{eqnarray*}
where we use once more
\begin{equation*}
 G(z) = g_+(z)-g_-(z),\hspace{0.5cm}z\in\mathcal{B}\cup \epsilon_0\cup\epsilon_1\cup \epsilon_2;\hspace{1cm}G(z)=0,\ \ z\in\epsilon_0
\end{equation*}
and $\alpha_0=0,\alpha_1,\alpha_2\in\mathbb{R}$
  \item As $z\rightarrow\infty$,
\begin{equation*}
  Y(z) = I+\mathcal O\left(z^{-1}\right)
\end{equation*}
\end{itemize}
\end{problem}
Since in all  shaded regions $\Re\varphi(z)>0$, we obtain for the jump matrix $G_Y(z)$ in the latter problem 
\begin{equation}\label{esti3}
  G_Y(z)e^{nG(z)\sigma_3}\longrightarrow I,\hspace{0.5cm}z\in \epsilon_j,\ j=1,2
\end{equation}
as $n\rightarrow\infty$ and the convergence is exponentially fast away from the branchpoints $z=a_j,j=1,\ldots,6$. In the white shaded regions one uses again
the analytical continuation of $G(z)$ combined with matrix factorizations. These techniques allow us to split the original contours $\gamma_1,\gamma_2,\gamma_3$ as shown
in Figure \ref{open2}. Without listing all formal steps, compare \eqref{openup} in genus zero case, we are lead to a RHP for a function $S(z)$ with jumps
\begin{equation*}
  S_+(z)=S_-(z)\begin{bmatrix}
                            0 & (2\pi iz)^{-1}e^{in\alpha_{j-1}}\\
-2\pi iz e^{-in\alpha_{j-1}}& 0\\
                           \end{bmatrix},\hspace{0.5cm}z\in \gamma_j,\ \ j=1,2,3
\end{equation*}
on the branchcuts. The jumps on the corresponding lense boundaries are again exponentially close to the unit matrix in the limit $n\rightarrow\infty$, hence we need to focus
on the construction of the parametrices.\bigskip

In order to formulate the model RHP we neglect the entries in the jumps of $S(z)$ that are exponentially suppressed and 
use that $G(z)=g_+(z) - g_-(z)$ for $z\in \epsilon_j$ is piecewise constant 
\begin{equation*}
  G(z) = -i\pi\Omega_j,\hspace{0.5cm} j=1,2.
\end{equation*}
We then are lead to the following model RHP
\begin{problem}
\label{RHPM}
Find a $2\times 2$ matrix valued piecewise analytic function $M(z)=M(z;x)$ such that
\begin{itemize}
 \item $M(z)$ is analytic for $z\in\mathbb{C}\backslash(\mathcal{B}\cup \epsilon_1\cup \epsilon_2)$
  \item The boundary values are connected via the jump relations
\begin{eqnarray*}
  M_+(z)&=&M_-(z)\begin{bmatrix}
                              0 & (2\pi iz)^{-1}e^{in\alpha_{j-1}}\\
-2\pi iz e^{-in\alpha_{j-1}}& 0\\
                             \end{bmatrix},\hspace{0.5cm}z\in\gamma_j,\ \ j=1,2,3\\
M_+(\lambda)&=&M_-(\lambda)e^{i n \pi \Omega_j\sigma_3},\hspace{0.5cm}z\in \epsilon_j,\ \ j=1,2
\end{eqnarray*}
  \item $M(z)$ is square integrable at the branchpoints, more precisely for $j\in\{1,\ldots,6\}$
\begin{equation*}
  M(z) = \mathcal O\left(|z-a_j|^{-1/4}\right),\hspace{0.5cm}z\rightarrow a_j,\ z\not\in\mathcal{B}\cup \epsilon_1\cup \epsilon_2
\end{equation*}
  \item We have the normalization
\begin{equation*}
  M(z) = I+\mathcal O\left(z^{-1}\right),\hspace{0.5cm}z\rightarrow\infty
\end{equation*}
\end{itemize}
\end{problem}
In Figure \ref{modRHP1} we depict schematically the jump matrices of the RHP \ref{RHPM}.
\begin{figure}[tbh]
\begin{center}
\resizebox{1\textwidth}{!}{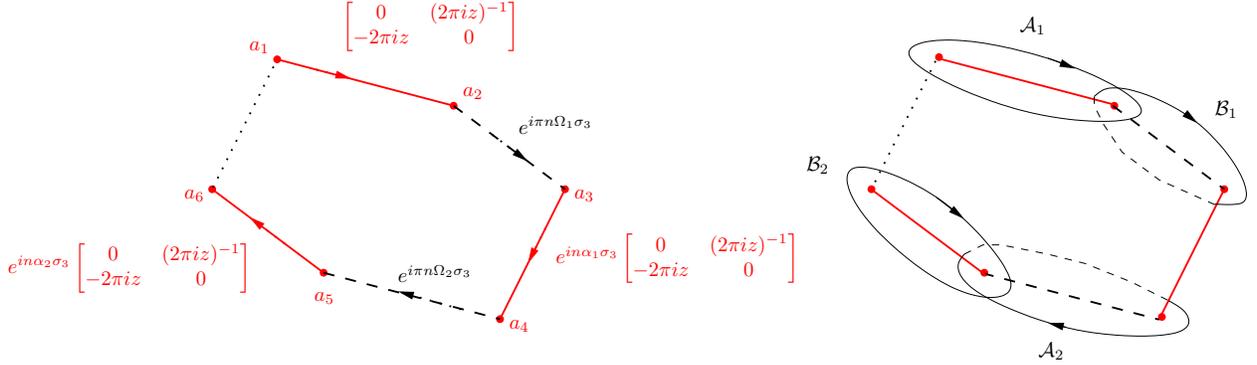}
\caption{The jump contour for $M(z)$ on the left and on the right the homology basis for $X$}
\label{modRHP1}
\end{center}
\end{figure}
Next we introduce the cycles $\{\mathcal{A}_j,\mathcal{B}_j\}_{j=1}^2$ as indicated in the same Figure \ref{modRHP1} on the right: these cycles form a homology 
basis for $X$ (cf. \cite{FK}). The values of $\Omega_j = \frac{1}{i\pi}(g_+(z) - g_-(z)),\ \ z\in \epsilon_j$ and $\alpha_{j-1} = \frac{1}{i}(g_+(z) + g_-(z) - \theta(z)- \ell) ,\ \ z\in \mathcal \gamma_j$ (cf. \eqref{con1}, \eqref{gfunction}) can then be expressed in terms of the periods of the meromorphic differential $\d\phi=y(z)\d z$ as follows
\begin{equation}\label{gj}
  \alpha_1 = \frac{1}{i}\oint_{\mathcal{B}_1}\d\phi,\hspace{0.25cm}\alpha_2=\frac{1}{i}\left(\oint_{\mathcal{B}_1}\d\phi+\oint_{\mathcal{A}_2}\d\phi\right);\hspace{0.75cm}
  \Omega_1 = \frac{1}{i\pi}\oint_{\mathcal{A}_1}\d\phi,\hspace{0.25cm}\Omega_2=-\frac{1}{i\pi}\oint_{\mathcal{B}_2}\d\phi.
\end{equation}

\subsection{Period matrices and normalized differentials}
\label{permat}

We are now going to construct an explicit solution to the RHP \ref{RHPM} in terms of theta functions, however this requires some preparation. Recall the 
homology basis $\{\mathcal{A}_j,\mathcal{B}_j\}_{j=1}^2$ as shown in Figure \ref{modRHP1} on the right. Introduce two holomorphic one forms on $X$ and respective periods
\begin{equation}\label{forms}
  \eta_1 = \frac{\d z}{\sqrt{R(z)}},\hspace{0.5cm}\eta_2 = \frac{z\,\d z}{\sqrt{R(z)}}\ ;\qquad
 \mathbb{A}_{jk} = \oint_{\mathcal A_k}\eta_j,\hspace{1cm}\mathbb{B}_{jk}=\oint_{\mathcal{B}_k}\eta_j.
\end{equation}
Recalling the symmetry of the branchpoints  $a_{k+3} =-a_{k}, \ k=1,2,3$ the reader verifies immediately that 
\begin{equation}\label{formrel}
  \oint_{\mathcal{A}_1}\eta_1 = \oint_{\mathcal{A}_2}\eta_1,\hspace{0.5cm}\oint_{\mathcal{B}_1}\eta_1 = \oint_{\mathcal{B}_2}\eta_1;\hspace{1cm}
\oint_{\mathcal{A}_1}\eta_2 =-\oint_{\mathcal{A}_2}\eta_2,\hspace{0.5cm}\oint_{\mathcal{B}_1}\eta_2 = -\oint_{\mathcal{B}_2}\eta_2.
\end{equation}
It is well-known (cf. \cite{FK}) that the $A$-period matrices  $\mathbb{A}=[\mathbb{A}_{jk}]_{j,k=1}^2$, resp. $B$-period matrix $\mathbb{B}=[\mathbb{B}_{jk}]_{j,k=1}^2$ are non-singular, in particular from \eqref{formrel}
\begin{equation*}
  \mathbb{A} = \begin{bmatrix}
                \mathbb{A}_{11} & \mathbb{A}_{11}\\
-\mathbb{A}_{22} & \mathbb{A}_{22}\\
               \end{bmatrix},\hspace{1cm}\mathbb{A}_{jj} = \oint_{\mathcal{A}_j}\eta_j,\ \ j=1,2.
\end{equation*}
This allows us to introduce the {\em normalized (first kind)} differentials $\{\omega_j\}_{j=1}^2$
\begin{equation}\label{cano}
  \omega_1 = \frac{1}{2}\left(\frac{\eta_1}{\mathbb{A}_{11}}-\frac{\eta_2}{\mathbb{A}_{22}}\right),\hspace{0.5cm} 
  \omega_2 = \frac{1}{2}\left(\frac{\eta_1}{\mathbb{A}_{11}}+\frac{\eta_2}{\mathbb{A}_{22}}\right)
\end{equation}
which satisfy the standard normalization
\begin{equation*}
  \oint_{\mathcal A_k}\omega_j = \delta_{jk},\hspace{0.5cm}j,k=1,2.
\end{equation*}
The corresponding matrix of $B$-periods, ${\boldsymbol \tau}=[\tau_{jk}]_{j,k=1}^2$ with $\tau_{jk} = \oint_{\mathcal B_j}\omega_k$, is computed as
\begin{equation}\label{Bmat}
  {\boldsymbol \tau}=\frac{1}{2}\begin{bmatrix}
        \varkappa_1+\varkappa_2 & \varkappa_1-\varkappa_2\\
\varkappa_1-\varkappa_2 & \varkappa_1+\varkappa_2\\
       \end{bmatrix},\hspace{1cm} \varkappa_j = \frac{1}{\mathbb{A}_{jj}}\oint_{\mathcal{B}_j}\eta_j=\frac{\mathbb{B}_{jj}}{\mathbb{A}_{jj}},\ \ j=1,2.
\end{equation}

Finally we define the Abel map\footnote{To be precise, we are defining the Abel map only of one sheet of the Riemann surface. In the present setting, the Abel map of the other sheet is obtained by simply changing the overall sign $\u(z) \mapsto -\u(z)$.} by 
\begin{equation*}
 	\u: \C\mathbb P^1 \setminus (\mathcal B\cup \epsilon_1\cup \epsilon_2)\rightarrow \C^2,\hspace{0.5cm}z  \mapsto  \ds \u(z) = \int_{a_1}^z \vec \omega
\end{equation*}
where the integration contour is the same for both components and it is chosen in the simply connected domain $\C\mathbb P^1 \setminus (\mathcal B\cup \epsilon_1\cup \epsilon_2)$. We summarize the following properties
\bp\label{Abprop} The Abelian integral $\mathfrak{u}(z)$ is single-valued and analytic in $\mathbb{CP}^1\backslash(\mathcal{B}\cup\epsilon_1\cup\epsilon_2)$. Moreover
\begin{equation*}
  \mathfrak{u}_+(z)+\mathfrak{u}_-(z) = \begin{cases}
                                                     0,&z\in \gamma_1\\
{\boldsymbol \tau}_1,&z\in \gamma_2\\
{\bf e}_2+{\boldsymbol \tau}_1,&z\in \gamma_3\\
                                                    \end{cases},\hspace{1cm}\mathfrak{u}_+(z)-\mathfrak{u}_-(z) = \begin{cases}
0,&z\in \epsilon_0\\
{\bf e}_1,&z\in \epsilon_1\\
-{\boldsymbol \tau}_2,&z\in \epsilon_2
\end{cases}
\end{equation*}
where ${\bf e}_j$ denotes again the standard basis vector in $\mathbb{C}^2$ and ${\boldsymbol \tau}_j={\boldsymbol \tau}{\bf e}_j$. Also $\mathfrak{u}(a_1)=0$ and
\begin{equation*}
 \mathfrak{u}(a_2) = \frac{1}{2}{\bf e}_1,\hspace{0.25cm}\mathfrak{u}(a_3)=\frac{1}{2}({\bf e}_1+{\boldsymbol \tau}_1),
  \hspace{0.25cm}\mathfrak{u}(a_4) = \frac{1}{2}({\boldsymbol \tau}_1-{\boldsymbol \tau}_2),\hspace{0.25cm}\mathfrak{u}(a_5) =\frac{1}{2}({\boldsymbol \tau}_1-{\boldsymbol \tau}_2+{\bf e}_2),
  \hspace{0.25cm}\mathfrak{u}(a_6) = \frac{1}{2}({\boldsymbol \tau}_1+{\bf e}_2)
\end{equation*}
where all values are taken from the $(+)$ side.
\ep

\subsection{Szeg\"o function}
\label{Szegofun}

 Next we define a scalar Szeg\"o function $\mathcal{D}(z)$ for $z\in\mathbb{C}\mathbb P^1\backslash (\mathcal B\cup \epsilon_1\cup \epsilon_2)$
\begin{equation}\label{Szdef}
  \mathcal{D}(z) = \exp\left[\frac{\sqrt{R(z)}}{2\pi i}\left\{\sum_{j=1}^3\int_{a_{2j-1}}^{a_{2j}}
  \frac{\ln w}{\sqrt{R(w)}_+}\frac{\d w}{w-z}-\sum_{j=1}^2\int_{a_{2j}}^{a_{2j+1}}\frac{i\pi\delta_j}{\sqrt{R(w)}_+}
  \frac{\d w}{w-z}\right\}\right]
\end{equation}
where
\begin{equation}\label{dchoice}
  \vec{\delta}=(\delta_1,\delta_2)^t = 2\big[{\boldsymbol \tau}_1,{\bf e}_2\big]^{-1}\big(\mathfrak{u}(\infty)-\mathfrak{u}(0)\big).
\end{equation}
 One checks directly that $\mathcal{D}(z)$ has the following analytical properties
\begin{itemize}
 \item $\mathcal{D}(z)$ is analytic for $z\in\mathbb{C}\backslash[a_1,a_6]$
  \item The following jumps hold, with orientation as indicated in Figure \ref{modRHP1}
\begin{eqnarray*}
  \mathcal{D}_+(z)\mathcal{D}_-(z) &=& z,\hspace{0.5cm}z\in\gamma_j,\ \ j=1,2,3\\
  \mathcal{D}_+(z)&=&\mathcal{D}_-(z)e^{-i\pi\delta_j},\hspace{0.5cm}\lambda\in \epsilon_j,\ \ j=1,2
\end{eqnarray*}
  \item The function is bounded at infinity 
  thanks to the following identities 
  \begin{equation*}
  \sum_{j=1}^2\int_{a_{2j}}^{a_{2j+1}}\frac{w^{k-1}i\delta_j}{\sqrt{R(w)}_+}\d w = \sum_{j=1}^3\int_{a_{2j-1}}^{a_{2j}}\frac{w^{k-1}}{\sqrt{R(w)}_+}
  \ln(w)\d w = i\pi \int_0^{\infty}\frac{w^{k-1}}{\sqrt{R(w)}}\,\d w\ ,\ \ k=1,2,
\end{equation*} 
which we can rewrite as a system
\begin{equation*}
  \left(\delta_1\int_{a_2}^{a_3}+\delta_2\int_{a_4}^{a_5}\right)\begin{bmatrix}
                                                               1\\
w\\
                                                              \end{bmatrix}\frac{\d w}{\sqrt{R(w)}_+} = 
\frac{1}{2}\left(\delta_1\oint_{\mathcal{B}_1}+\delta_2\oint_{\mathcal{A}_2}\right)\begin{bmatrix}
                                                                                    1\\
w\\
                                                                                   \end{bmatrix}\frac{\d w}{\sqrt{R(w)}}
=\int_0^{\infty}\begin{bmatrix}
                 1\\
w\\
                \end{bmatrix}\frac{\d w}{\sqrt{R(w)}}.
\end{equation*}
Indeed, multiplying the above by $\mathbb A^{-1}$  we obtain
\begin{equation*}
 \delta_1\oint_{\mathcal{B}_1}\vec{\omega}+\delta_2\oint_{\mathcal{A}_2}\vec{\omega}=2\int_0^{\infty}\vec{\omega} = 2\big(\mathfrak{u}(\infty)-\mathfrak{u}(0)\big)
\end{equation*}
and therefore 
\begin{equation}
  \delta_1{\boldsymbol \tau}_1+\delta_2{\bf e}_2 = \big[{\boldsymbol \tau}_1,{\bf e}_2\big]\vec{\delta} = 2\big(\mathfrak{u}(\infty)-\mathfrak{u}(0)\big)
  \label{deltabel}
\end{equation}
where ${\bf e}_j$ denotes the standard basis vector in $\mathbb{C}^2$ and $\boldsymbol \tau_j=\boldsymbol \tau{\bf e}_j$. Hence \eqref{dchoice} ensures the 
required normalization $\mathcal{D}(\infty)<\infty$.
\end{itemize}
\subsection{Intermediate Step}
Keeping the properties of $\mathcal{D}(z)$ in mind, introduce
\begin{equation*}
  \Psi(z) = e^{i\frac{\pi}{4}\sigma_3}(2\pi i)^{\frac{1}{2}\sigma_3}\big(\mathcal{D}(\infty)\big)^{\sigma_3}M(z)
  \big(\mathcal{D}(z)\big)^{-\sigma_3}(2\pi i)^{-\frac{1}{2}\sigma_3}e^{-i\frac{\pi}{4}\sigma_3},\hspace{0.5cm}z\in\mathbb{C}\backslash(\mathcal{B}\cup \epsilon_1\cup \epsilon_2)
\end{equation*}
and obtain the following RHP with the jumps schematically depicted in Figure \ref{szego}.
\begin{figure}[tbh]
\begin{center}
\resizebox{0.3\textwidth}{!}{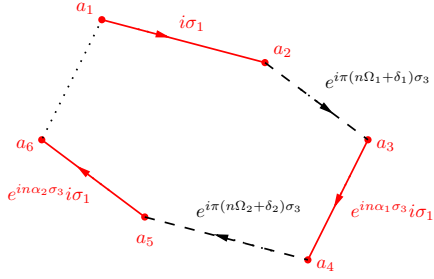}
\caption{The jump contour for $\Psi(z)$.}
\label{szego}
\end{center}
\end{figure}
\begin{problem}
\label{RHPPsi}
Find the $2\times 2$ matrix valued function $\Psi(z)$ such that 
\begin{itemize}
 \item 	$\Psi(z)$ is analytic for $z\in\mathbb{CP}^1\backslash(\mathcal{B}\cup \epsilon_1\cup \epsilon_2)$
  \item The jumps are as follows
\begin{eqnarray*}
  \Psi_+(z) &=& \Psi_-(z)e^{i\pi d_j\sigma_3}i\sigma_1,\hspace{0.5cm}z\in\gamma_j,\ \ \ j=0,1,2\\
  \Psi_+(z) &=& \Psi_-(z)e^{i\pi c_j\sigma_3},\hspace{0.5cm}z\in \epsilon_j,\ \ \ j=1,2
\end{eqnarray*}
where we introduced the abbreviations
\begin{equation}\label{abbrevi}
 c_j=n\Omega_j+\delta_j,\ \ j=1,2;\hspace{1cm}d_j = \frac{n}{\pi}\alpha_j,\ \ j=1,2;\hspace{0.5cm}d_0=0
\end{equation}
  \item As $z\rightarrow\infty$,
\begin{equation*}
  \Psi(z)= I+\mathcal O\left(z^{-1}\right).
\end{equation*}
\end{itemize}
\end{problem}

The construction of $\Psi$ is the last step in the construction of $M(z)$. To this end we introduce the function 
\begin{equation*}
  h(z) = \sqrt[4]{\frac{z-a_6}{\prod_1^5(z-a_j)}},\hspace{0.5cm}z\in\mathbb{C}\big\backslash(\mathcal{B}\cup\epsilon_1\cup\epsilon_2)
\end{equation*}
with the branch fixed by the requirement $h(z)\sim \frac{1}{z}$ as $z\rightarrow\infty$. The boundary values of $h(z)$ satisfy
\begin{equation}\label{hj1}
  h_+(z)=h_-(z),\ \ z\in \epsilon_0;\hspace{0.5cm}h_+(z)=-h_-(z),\ \ z\in\epsilon_1;\hspace{0.5cm}h_+(z)=h_-(z),\ \ z\in\epsilon_2
\end{equation}
\begin{equation}\label{hj2}
 h_+(z)=ih_-(z),\ \ z\in\gamma_1;\hspace{0.5cm}h_+(z)=-ih_-(z),\ \ z\in\gamma_2;\hspace{0.5cm}h_+(z)=ih_-(z),\ \ z\in\gamma_3.
\end{equation}

We now construct the solution to the model problem in terms of the Riemann theta function
\begin{equation*}
  \Theta(\vec{z}\,)\equiv \Theta(\vec{z}\,|{\boldsymbol \tau}) = \sum_{\vec{k}\in\mathbb{Z}^2}\exp\Big[\pi\langle\vec{k}{\boldsymbol \tau},\vec{k}\rangle+2\pi i\langle \vec{k},\vec{z}\,\rangle\Big],
\ \ \vec{z}\in\mathbb{C}^2;\hspace{0.5cm} \langle\vec{a},\vec{c}\,\rangle = \sum_{j=1}^2a_jc_j.
\end{equation*}
It is convenient also to introduce the {\em theta function with characteristics} $\vec{\boldsymbol \alpha},\vec{\boldsymbol \beta}\in\mathbb{C}^2$
\begin{equation*}
  \Theta\begin{bmatrix}
         \vec{\boldsymbol \alpha}\,\\
\vec{\boldsymbol \beta}\,
        \end{bmatrix}(\vec{z}\,|{\boldsymbol \tau}) = \exp\left[2\pi i\left(\frac{1}{8}\langle\vec{\boldsymbol \alpha}{\boldsymbol \tau},\vec{\boldsymbol \alpha}\rangle+\frac{1}{2}\langle\vec{\boldsymbol \alpha},\vec{z}\,\rangle
+\frac{1}{4}\langle \vec{\boldsymbol \alpha},\vec{\boldsymbol \beta}\,\rangle\right)\right]\Theta\left(\vec{z}+\frac{1}{2}\vec{\boldsymbol \beta}+\frac{1}{2}{\boldsymbol \tau}\vec{\boldsymbol \alpha}\,\bigg|{\boldsymbol \tau}\right).
\end{equation*}
The reader will find in Appendix \ref{app2} all the main properties that are used below. Since we are dealing with a hyperelliptic Riemann surface $X$, the vector of Riemann constants $\mathcal{K}$ (cf. \cite{FK}) is given by
\begin{equation}
  \mathcal{K} = \sum_{j=1}^2\mathfrak{u}(a_{2j+1}) \equiv \frac{1}{2}({\bf e}_1+{\bf e}_2-{\boldsymbol \tau}_2) \mod \Lambda
  \label{RiemannK}
\end{equation}
where $\Lambda = \mathbb{Z}^2+{\boldsymbol \tau}\mathbb{Z}^2$ is the period lattice. Recall also (cf. \cite{FK}) that 
\begin{equation}\label{ass1}
  {\bf f}^{(\pm)}(z)=\Theta\left(\mathfrak{u}(z)\mp\mathfrak{u}(\infty)-\mathfrak{u}(a_6)-\mathcal{K}\right)
\end{equation}
does not vanish identically, since the divisor of the points $\infty^{\pm},a_6$ is nonspecial on the hyperelliptic Riemann surface $X$ (compare again Appendix \ref{app2} for a
short summary of the relevant theory). This observation allows us to introduce the functions $P^{(\pm)}(z) = P^{(\pm)}(z;\vec{\boldsymbol \alpha},\vec{\boldsymbol \beta})$ with
\begin{equation*}
  P^{(\pm)}(z) = \left(\frac{\Theta\begin{bmatrix}
         \vec{\boldsymbol \alpha}\,\\
\vec{\boldsymbol \beta}\,
        \end{bmatrix}\big(\mathfrak{u}(z)\mp\mathfrak{u}(\infty)-\mathcal{K}\big)}{\Theta\big(\mathfrak{u}(z)\mp\mathfrak{u}(\infty)-\mathfrak{u}(a_6)-\mathcal{K}\big)},
\frac{\Theta\begin{bmatrix}
         \vec{\boldsymbol \alpha}\,\\
\vec{\boldsymbol \beta}\,
        \end{bmatrix}\big(-\mathfrak{u}(z)\mp\mathfrak{u}(\infty)-\mathcal{K}\big)}{\Theta\big(-\mathfrak{u}(z)\mp\mathfrak{u}(\infty)-\mathfrak{u}(a_6)-\mathcal{K}\big)}
\right)h(z)e^{i\pi  u_1(z)\sigma_3}.
\end{equation*}
where we use $\mathfrak{u}(z)=\big(u_1(z),u_2(z)\big)^t$. The following Proposition is crucial in the construction of the outer parametrix.
\bp Both functions, $P^{(+)}(z)$ and $P^{(-)}(z)$, are single-valued and analytic in $\mathbb{C}\backslash(\mathcal{B}\cup\epsilon_1\cup\epsilon_2)$ with
\begin{eqnarray*}
  P^{(\pm)}_+(z) &=&P^{(\pm)}_-(z)(i\sigma_1),\ \ z\in \gamma_1\\
  P^{(\pm)}_+(z) &=&P^{(\pm)}_-(z)\exp\Big[i\pi\langle\vec{\boldsymbol \alpha},{\bf e}_1\rangle\sigma_3\Big],\ \ z\in \epsilon_1\\
  P^{(\pm)}_+(z) &=&P^{(\pm)}_-(z)\exp\left[i\pi\langle{\bf e}_1,\vec{\boldsymbol \beta}\,\rangle\sigma_3\right](-i\sigma_1),\ \ z\in \gamma_2\\
  P^{(\pm)}_+(z) &=&P^{(\pm)}_-(z)\exp\left[i\pi\big(1+\langle{\bf e}_2,\vec{\boldsymbol \beta}\,\rangle\big)\sigma_3\right],\ \ z\in \epsilon_2\\
  P^{(\pm)}_+(z) &=&P^{(\pm)}_-(z)\exp\left[-i\pi\big(\langle\vec{\boldsymbol \alpha},{\bf e}_2\rangle-\langle{\bf e}_1,\vec{\boldsymbol \beta}\,\rangle\big)\sigma_3\right](i\sigma_1),\ \ z\in\gamma_3.
\end{eqnarray*}
\ep
\begin{proof}
As the Abelian integral $\mathfrak{u}(z)$ is single-valued and analytic on $\mathbb{C}\backslash(\mathcal{B}\cup\epsilon_1\cup\epsilon_2)$ and ${\bf f}^{(\pm)}(z)$ does not vanish identically, we first obtain
(cf. \cite{FK}) that $P^{(\pm)}(z)$ is single-valued and meromorphic on $\mathbb{C}\backslash(\mathcal{B}\cup\epsilon_1\cup\epsilon_2)$. Moreover, general theory (see Theorem \ref{generalTheta}) asserts, that ${\bf f}^{(+)}(z)$
has precisely two zeros on $X$, both on the first sheet at $z=\infty^+$ and at $z=a_6$. However $h(z)$ has zeros at the same points and its local behavior
matches the vanishing behavior of ${\bf f}^{(+)}(z)$, hence we obtain analyticity of the first column in $P^{(\pm)}(z)$ for $z\in\mathbb{C}\backslash(\mathcal{B}\cup\epsilon_1\cup\epsilon_2)$.
The second column can be treated similarly using the parity of the theta-function. The stated jumps follow now directly from Proposition \ref{Abprop} and \eqref{hj1},\eqref{hj2} using that
\begin{equation*}
  F(\vec{z}\,) = \frac{\Theta\begin{bmatrix}
         \vec{\boldsymbol \alpha}\,\\
\vec{\boldsymbol \beta}\,
        \end{bmatrix}\big(\vec{z}\mp\mathfrak{u}(\infty)-\mathcal{K}\,|{\boldsymbol \tau}\big)}{\Theta(\vec{z}\mp\mathfrak{u}(\infty)-\mathfrak{u}(a_6)-\mathcal{K}\,|{\boldsymbol \tau})}
e^{i\pi \langle\vec{z},{\bf e}_1\rangle}
\end{equation*}
formally satisfies
\begin{equation*}
  F(\vec{z}+\vec{\mu}+{\boldsymbol \tau}\vec{\lambda}\,|{\boldsymbol \tau}) = \exp\left[i\pi \left(\langle\vec{\mu},{\bf e}_1\rangle-\langle\vec{\lambda},{\bf e}_2\rangle+\langle\vec{\boldsymbol \alpha},\vec{\mu}\,\rangle-\langle\vec{\lambda},\vec{\boldsymbol \beta}\,\rangle\right)\right]F(\vec{z}\,),
\hspace{0.5cm}\vec{\mu},\vec{\lambda}\in\mathbb{Z}^2.
\end{equation*}
\end{proof}
We now compare the jumps of $P^{(\pm)}(z)$ to the ones stated in Figure \ref{szego} for $\Psi(z)$. This in turn leads to the following system in
 $\mathbb{Z}/2\mathbb{Z}$ for the yet unknowns $\vec{\boldsymbol \alpha},\vec{\boldsymbol \beta}$
\begin{equation*}
 \langle\vec{\boldsymbol \alpha},{\bf e}_1\rangle \equiv c_1,\hspace{0.5cm}\langle{\bf e}_1,\vec{\boldsymbol \beta}\,\rangle+1\equiv d_1,\hspace{0.5cm}\langle{\bf e}_2,\vec{\boldsymbol \beta}\,\rangle+1\equiv c_2,
\hspace{0.5cm}\langle{\bf e}_1,\vec{\boldsymbol \beta}\,\rangle-\langle{\bf e}_2,\vec{\boldsymbol \alpha}\,\rangle \equiv d_2
\end{equation*}
and we take as solution in $\mathbb{C}^2$
\begin{equation}\label{albe}
  \vec{\boldsymbol \alpha}=\begin{bmatrix}
                c_1\\
d_1-d_2-1\\
               \end{bmatrix},\hspace{0.5cm}\vec{\boldsymbol \beta}=\begin{bmatrix}
d_1+1\\
c_2+1\\
\end{bmatrix}.
\end{equation}
With the latter choice \eqref{albe} and $P^{(\pm)}(z)=\big(P_1^{(\pm)}(z),P_2^{(\pm)}(z)\big)$
\bp 
\label{propQ}
{\bf [1]} The function
\begin{equation}\label{Qdef}
  Q(z) = Q(z;\vec{\boldsymbol \alpha},\vec{\boldsymbol \beta}) =\begin{bmatrix}
                                                           P^{(+)}_1(z) & P^{(+)}_2(z)\\
P^{(-)}_1(z) & P^{(-)}_2(z)
                                                         \end{bmatrix},\hspace{0.5cm}z\in\mathbb{C}\backslash(\mathcal{B}\cup\epsilon_1\cup\epsilon_2)
\end{equation}
with $\vec{\boldsymbol \alpha},\vec{\boldsymbol \beta}$ as in \eqref{albe} is single-valued and analytic in $\mathbb{C}\backslash(\mathcal{B}\cup\epsilon_1\cup\epsilon_2)$. Its jump behavior is depicted in Figure \ref{szego}.
Moreover, as $z\rightarrow\infty$,
\begin{equation}\label{psiexp}
  Q(z) = C_0\sigma_3e^{i\pi u_1(\infty)\sigma_3}\Theta\begin{bmatrix}
         \vec{\boldsymbol \alpha}\,\\
\vec{\boldsymbol \beta}\,
        \end{bmatrix}\big(-\mathcal{K}\big)\bigg\{I+\frac{Q_1}{z}+\mathcal O\left(z^{-2}\right)\bigg\},\ \ \ Q_1=\big(Q_1^{jk}\big)_{j,k=1}^2
\end{equation}
where
\begin{equation*}
  C_0^{-1} = \langle\nabla\Theta(-\mathfrak{u}(a_6)-\mathcal{K})\mathbb{A}^{-1},{\bf e}_2\rangle\neq 0
\end{equation*}
and
\begin{equation}\label{coeff}
  Q_1^{21}=-\frac{C_0^{-1}e^{2\pi iu_1(\infty)}}{\Theta\big(2\mathfrak{u}(\infty)-\mathfrak{u}(a_6)-\mathcal{K}\big)}\frac{\Theta\begin{bmatrix}
                                                                                                                                  \vec{\boldsymbol \alpha}\,\\
\vec{\boldsymbol \beta}\,\\
                                                                                                                                 \end{bmatrix}
\big(2\mathfrak{u}(\infty)-\mathcal{K}\big)}{\Theta\begin{bmatrix}
                                                    \vec{\boldsymbol \alpha}\,\\
\vec{\boldsymbol \beta}\,\\
                                                   \end{bmatrix}\big(-\mathcal{K}\big)}
\end{equation}
{\bf [2]} As a function of the characteristics $\vec{\boldsymbol \alpha}, \vec{\boldsymbol \beta}$, the matrix $Q(z)$ is periodic
\be
Q(z; \vec{\boldsymbol \alpha}, \vec{\boldsymbol \beta}) =Q(z; \vec{\boldsymbol \alpha} + 2\vec \nu, \vec{\boldsymbol \beta} + 2 \vec \nu')\ ,\ \ \ \forall \vec \nu ,\vec \nu' \in \Z^2. 
\ee
\ep
The property {\bf [2]} in Prop. \ref{propQ} follows from Prop. \ref{propB2}.
Note that the dependency on $n$ is only in the linear dependency of the  characteristics $\vec{\boldsymbol \alpha},\vec{\boldsymbol \beta}$ \eqref{abbrevi}.
Collecting the results we have completed the construction of $\Psi$ which we summarize hereafter for reference.
\bc
\label{solPsi}
{\bf [1]} The solution of the RHP \ref{RHPPsi} is given by 
\be
\Psi(z) := Q^{-1}(\infty) Q(z)
\ee
with $Q(z)$ as in Prop. \ref{propQ} and 
\be
Q(\infty) =              C_0\sigma_3e^{i\pi u_1(\infty)\sigma_3}\Theta\begin{bmatrix}
         \vec{\boldsymbol \alpha}\,\\
\vec{\boldsymbol \beta}\,
        \end{bmatrix}\big(-\mathcal{K}\big)
\ee
and the solution exists if and only if $\Theta\begin{bmatrix}
         \vec{\boldsymbol \alpha}\,\\
\vec{\boldsymbol \beta}\,
        \end{bmatrix}\big(-\mathcal{K}\big)\neq 0$.\\
{\bf [2]} For each compact subset of its domain of analyticity in $z$, the entries of $\Psi(z)$ are uniformly bounded with respect to the characteristics  in any closed subset of the domain 
 \be
 \vec{\boldsymbol \alpha}, \vec {\boldsymbol \beta} \in \R^2:  \ \ \le| \Theta\begin{bmatrix}
         \vec{\boldsymbol \alpha}\,\\
\vec{\boldsymbol \beta}\,
        \end{bmatrix}\big(-\mathcal{K}\big) \ri| > 0
        \label{cond0}
 \ee 
\ec
Note that the condition \eqref{cond0} is well defined because the absolute value of the Theta function involved is a periodic function of the characteristics (compare with the second property in Prop. \ref{propB2}).
The condition \eqref{cond0} can be made more transparent in terms of the data of our problem 
(we use  \eqref{RiemannK})
\begin{equation*}
  \Theta\begin{bmatrix}
         \vec{\boldsymbol \alpha}\,\\
\vec{\boldsymbol \beta}\,
        \end{bmatrix}\big(-\mathcal{K}\big)\propto \Theta\left(\frac{1}{2}\vec{\boldsymbol \beta}+\frac{1}{2}{\boldsymbol \tau}\vec{\boldsymbol \alpha}-\mathcal{K}\right) 
\end{equation*}
where the proportionality is by a never-vanishing term.
Replacing the expressions \eqref{abbrevi}, \eqref{gj}, \eqref{Bmat}, \eqref{deltabel} in the above formula yields
\begin{equation*}
  \frac{1}{2}\vec{\boldsymbol \beta}+\frac{1}{2}{\boldsymbol \tau}\vec{\boldsymbol \alpha} -\mathcal{K}=\frac{1}{2}\begin{bmatrix}
                                                                              d_1\\
c_2\\
                                                                             \end{bmatrix}+\frac{{\boldsymbol \tau}}{2}\begin{bmatrix}
c_1\\
d_1-d_2\\
\end{bmatrix} = \frac{n}{2\pi i}\begin{bmatrix}
1\\
-1\\
\end{bmatrix}\left(\oint_{\mathcal{B}_1}\d\phi+\varkappa_2\oint_{\mathcal{A}_1}\d\phi\right)+\mathfrak{u}(\infty)-\mathfrak{u}(0).
\end{equation*}

We can further simplify the expression (all values are taken from the (+) side of the branchcuts):
\begin{eqnarray*}
  \mathfrak{u}(0) &\ds\mathop{=}^{\eqref{cano}}& \frac{1}{2}\begin{bmatrix}
                                1\\
1\\
                               \end{bmatrix}\int_{a_1}^0\frac{\eta_1}{\mathbb{A}_{11}}-\frac{1}{2}\begin{bmatrix}
1\\
-1\\
\end{bmatrix}\int_{a_1}^0\frac{\eta_2}{\mathbb{A}_{22}}=\frac{1}{2}\begin{bmatrix}
1\\
1\\
\end{bmatrix}\frac{\varkappa_1}{2}-\frac{1}{2}\begin{bmatrix}
1\\
-1\\
\end{bmatrix}\int_{a_1}^0\frac{\eta_2}{\mathbb{A}_{22}}\\
&=&-\frac{1}{2}\left(\int_{a_1}^0\frac{\eta_2}{\mathbb{A}_{22}}-\frac{\varkappa_2}{2}\right)\begin{bmatrix}
                                                                                          1\\
-1\\
                                                                                         \end{bmatrix}-\mathcal{K}+\frac{1}{2}\begin{bmatrix}
1\\
1\\
\end{bmatrix}.
\end{eqnarray*}
Thus we get
\begin{eqnarray*}
  \mathfrak{u}(\infty)-\mathfrak{u}(0)
&=&\frac{1}{2}\left(\int_{a_1}^0\frac{\eta_2}{\mathbb{A}_{22}}+\frac{\varkappa_2}{2}\right)\begin{bmatrix}
                                                                                            1\\
-1\\
                                                                                           \end{bmatrix}+\mathfrak{u}(\infty)+\mathcal{K}-\frac{1}{2}(\varkappa_2+1)\begin{bmatrix}
1\\
-1\\
\end{bmatrix}-{\bf e}_2
\end{eqnarray*}
which implies all together
\begin{eqnarray}
 \frac{1}{2}\vec{\boldsymbol \beta}+\frac{1}{2}{\boldsymbol \tau}\vec{\boldsymbol \alpha} -\mathcal{K}&=&\rho_n\begin{bmatrix}
                                                                   1\\
-1\\
                                                                  \end{bmatrix}+\mathfrak{u}(\infty)+\mathcal{K}-\frac{1}{2}(\varkappa_2+1)\begin{bmatrix}
1\\
-1\\
\end{bmatrix}-{\bf e}_2;\label{simple}\\
\rho_n&=&\frac{n}{2\pi i}\left(\oint_{\mathcal{B}_1}\d\phi+\varkappa_2\oint_{\mathcal{A}_1}\d\phi\right)
  +\frac{1}{2}\left(\int_{a_1}^0\frac{\eta_2}{\mathbb{A}_{22}}+\frac{\varkappa_2}{2}\right)\nonumber.
\end{eqnarray}

Thus the non-solvability condition of the RHP \ref{RHPPsi} can be written in any of the following equivalent forms
\begin{equation}\label{zeroset}
 \Theta\left(
 \frac {n}{2\pi i} \le(\oint_{\mathcal B_1} + \varkappa_2 \oint_{\mathcal A_1}\ri)\d\phi \le[{1 \atop -1}\ri] + \u(\infty)-\u(0)
 \right)\mathop{=}^{\eqref{simple}}\Theta\left(
 \le(\rho_n - \frac {\varkappa_2+1}2\ri) \le[{1\atop -1}\ri]+\mathfrak{u}(\infty)+\mathcal{K}
 \right)=0.
\end{equation}
This in turn defines implicitly a discrete set $\mathcal Z_{n} = \{x_{n,k}\}$ of points inside the star shaped region $\overline{\Delta}$ which eventually shall be identified with  the zero set of the Vorob'ev-Yablonski polynomial
$\mathcal{Q}_n(x)$ for sufficiently large $n$ (compare Corollary $1$ on page $65$ in \cite{BM} in the setting of the poles of rational PII solutions). From now on we stipulate to stay away from the points of $\mathcal Z_n$ \eqref{zeroset}. Once
this additional constraint on $x\in\mathbb{C}:\textnormal{dist}(x,\mathbb{C}\backslash\overline{\Delta})\geq\delta>0$ is in place we complete the construction of the outer
parametrix.
\bp 
Let $x\not \in \mathcal Z_n$ and $x\in\mathbb{C}\backslash\overline{\Delta}$; then the model problem for the outer parametrix $M(z)=M(z;x)$ depicted in \ref{modRHP1} is solvable. The solution is given explicitly by
\begin{equation*}
 M(z) = e^{-i\frac{\pi}{4}\sigma_3}(2\pi i)^{-\frac{1}{2}\sigma_3}\big(\mathcal{D}(\infty)\big)^{-\sigma_3}
  \Psi(z)
  \big(\mathcal{D}(z)\big)^{\sigma_3}(2\pi i)^{\frac{1}{2}\sigma_3}e^{i\frac{\pi}{4}\sigma_3},\hspace{0.5cm}z\in\mathbb{C}\backslash(\mathcal{D}\cup \epsilon_1\cup \epsilon_2).
\end{equation*}
with $\mathcal{D}(z)$ as in \eqref{Szdef} and $\Psi(z)$ as in Cor. \ref{solPsi}.
For any closed subset of the domain of analyticity in $z$ the entries of $M(z)$ are uniformly bounded in any compact subsets of $(n,x)\in \R\times  \Delta$ where \eqref{cond0} holds. 
\ep

The remaining six local parametrices near the branchpoints are defined in the disks 
\begin{equation*}
  D(a_j,r) = \{z\in\mathbb{C}\,|\, |z-a_j|<r\},\hspace{0.5cm} j=1,\ldots,6
\end{equation*}
with $r>0$ sufficiently small. The construction follows the standard lines using again Airy functions and we will not give details here. We only list the relevant matching relations
between parametrices $P_j(z)$ and the outer model function $M(z)$, in fact
\begin{equation}\label{matchgen2}
  P_j(z) = \left(I+\mathcal O\left(n^{-1}\right)\right)M(z),\hspace{0.5cm}n\rightarrow\infty
\end{equation}
which holds for $x\in\mathbb{C}:\textnormal{dist}(x,\mathbb{C}\backslash\overline{\Delta})\geq\delta>0$ away from the zero set $\mathcal{Z}_n$ and uniformly 
for $z\in\bigcup_{j=1}^6\partial D(a_j,r)$. This completes the construction of the parametrices in the genus two situation.
\subsection{Reduction to Jacobi theta function of genus $1$}
We now show that the expression on the left hand side of \eqref{zeroset} is expressible as a square of the ordinary Jacobi theta function
\begin{equation}\label{Jacell}
  \vartheta(z)\equiv\vartheta(z|\varkappa_2)=\sum_{k\in\mathbb{Z}}\exp\left[i\pi k^2\varkappa_2+2\pi ikz\right],\hspace{0.5cm}\varkappa_2 = 
  \frac{\mathbb{B}_{22}}{\mathbb{A}_{22}}=\frac{\oint_{\mathcal{B}_2}\eta_2}{\oint_{\mathcal{A}_2}\eta_2}.
\end{equation}
\bl
\label{cute} 
{\bf [1]} Let  $\mathcal{K}=\frac{1}{2}({\bf e_1}+{\bf e_2}-{\boldsymbol \tau}_2)$ be the vector of Riemann constants \eqref{RiemannK} and $ {\boldsymbol \tau}$ as in  \eqref{Bmat}. Then  we have
\begin{equation}\label{neat}
  \Theta\left(\lambda\begin{bmatrix}
                                                       1\\
-1\\
                                                      \end{bmatrix}+\mathfrak{u}(z)+\mathcal{K}-\frac{1}{2}(\varkappa_2+1)\begin{bmatrix}
                                                                                                                           1\\
-1\\
                                                                                                                          \end{bmatrix}\right)=C(z)
\vartheta(\lambda)\vartheta\left(\lambda-\frac{1}{\mathbb{A}_{22}}\int_{a_1}^{z}\eta_2-\frac{\varkappa_2}{2}\right)
\end{equation}
identically for $\lambda \in \C$ and $z\in X$, where $C(z)$ is independent of $\lambda$ and is a nowhere zero function of $z$ on the universal cover of $X$. 
{\bf [2]} In particular if $z=\infty^{\pm}$, we obtain
\begin{equation*}
  \Theta\left(\lambda\begin{bmatrix}
                                                       1\\
-1\\
                                                      \end{bmatrix}+\mathfrak{u}(\infty^{\pm})+\mathcal{K}-\frac{1}{2}(\varkappa_2+1)\begin{bmatrix}
                                                                                                                           1\\
-1\\
                                                                                                                          \end{bmatrix}\right)=C(\infty^{\pm})e^{2\pi i\lambda}
\vartheta^2(\lambda).\ \ 
\end{equation*}
\el
\begin{proof}
{\bf [1]} 
Define for $\lambda\in\mathbb{C}$ and $z\in X$ 
\begin{equation*}
  f(\lambda) = f(\lambda;z|\varkappa_2) = \frac{\Theta\left(\lambda\begin{bmatrix}
                                                       1\\
-1\\
                                                      \end{bmatrix}+\mathfrak{u}(z)+\mathcal{K}-\frac{1}{2}(\varkappa_2+1)\begin{bmatrix}
                                                                                                                           1\\
-1\\
                                                                                                                          \end{bmatrix}\right)}
{\vartheta(\lambda)\vartheta(\lambda+c(z))}.
\end{equation*}
Using the  periodicity properties of the Theta functions involved the reader may verify that 
\begin{equation*}
 f(\lambda+1+\varkappa_2)=f(\lambda)\exp\left[2\pi i\Big(c(z)+\frac{1}{\mathbb{A}_{22}}\int_{a_1}^z\eta_2+\frac{\varkappa_2}{2}\Big)\right]
\end{equation*}
and therefore with $c(z) = -\frac{1}{\mathbb{A}_{22}}\int_{a_1}^z\eta_2-\frac{\varkappa_2}{2}$ the function $f=f(\lambda)$ is elliptic. The Jacobi elliptic function 
$\vartheta(\lambda)$ as in \eqref{Jacell} has a simple zero at $\lambda=\frac{1}{2}(1+\varkappa_2)$, hence $f(\lambda)$ can have at most two simple poles in the fundamental
region $\mathcal R$ of the quotient $\C/(\Z + \varkappa_2 \Z)$
\begin{figure}[tbh]
\begin{center}
\resizebox{0.3\textwidth}{!}{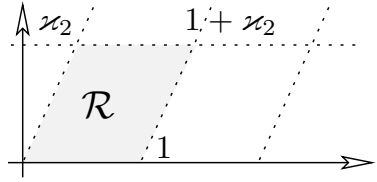}
\caption{The fundamental region $\mathcal{R}$}
\end{center}
\end{figure}

If we substitute $\lambda=\frac{1}{2}(1+\varkappa_2)$ then the numerator of $f(\lambda)$ 
becomes $ \Theta\big(\mathfrak{u}(z)+\mathcal{K}\big)$,
which vanishes since the argument  is the image of a divisor of degree $g-1=1$ (see Corollary \ref{generalThetadiv}). Hence $f(\lambda)$ can have at most one simple pole in $\mathcal{R}$, i.e.
$f$ is an elliptic function of order one, and therefore a constant. We have thus established \eqref{neat} with a $\lambda$ independent term $C=C(z)$. We now have to show that
$C(z)$ does not vanish. To this end consider the behavior of
\begin{equation*}
 C(z)=\frac{\Theta\left(\lambda\begin{bmatrix}
                                                       1\\
-1\\
                                                      \end{bmatrix}+\mathfrak{u}(z)+\mathcal{K}-\frac{1}{2}(\varkappa_2+1)\begin{bmatrix}
                                                                                                                           1\\
-1\\
                                                                                                                          \end{bmatrix}\right)}
{\vartheta(\lambda)\vartheta\left(\lambda-\frac{1}{\mathbb{A}_{22}}\int_{a_1}^{z}\eta_2-\frac{\varkappa_2}{2}\right)}
\end{equation*}
as $z$ varies over $X$. Once more, the periodicity properties of the Theta functions involved give the following behavior under analytic continuation along a closed contour $\gamma$:
\begin{align*}
 z\mapsto z_{\gamma}\ \ \textnormal{along}\ \mathcal{A}_1:\ \ & C(z_{\gamma})=C(z);&& z\mapsto z_{\gamma}\ \ \textnormal{along}\ \mathcal{A}_2: && C(z_{\gamma}) = C(z)\\
 z\mapsto z_{\gamma}\ \ \textnormal{along}\ \mathcal{B}_1:\ \ & C(z_{\gamma})= C(z)e^{-2\pi iu_2(z)};&&
  z\mapsto z_{\gamma}\ \ \textnormal{along}\ \mathcal{B}_2: && C(z_{\gamma})=C(z)e^{-2\pi iu_1(z)+i\pi\varkappa_2}.
\end{align*}
If we assume that $C(z)$ is not identically zero (we shall prove this later), we can count the zeros of $C=C(P)$ on $X$ by integrating $\d\ln C(P)$ along the boundary of the canonical dissection $\hat{X}$,
i.e. we compute
\begin{eqnarray*}
  \frac{1}{2\pi i}\oint_{\partial\hat{X}}\d\ln C(P) &=&\frac{1}{2\pi i}\sum_{j=1}^2\left(\int_{P_0}^{P_0+\mathcal{A}_j}+\int_{P_0+\mathcal{A}_j}^{P_0+\mathcal{A}_j+\mathcal{B}_j}
+\int_{P_0+\mathcal{A}_j+\mathcal{B}_j}^{P_0+\mathcal{B}_j}+\int_{P_0+\mathcal{B}_j}^{P_0}\right)\d\ln C(P)\\
 &=&\frac{1}{2\pi i}\sum_{j=1}^2\left(\int_{P_0}^{P_0+\mathcal{A}_j}-\int_{P_0+\mathcal{B}_j}^{P_0+\mathcal{A}_j+\mathcal{B}_j}+\int_{P_0+\mathcal{B}_j}^{P_0}-
\int_{P_0+\mathcal{A}_j+\mathcal{B}_j}^{P_0+\mathcal{A}_j}\right)\d\ln C(P)\\
&=&\int_{P_0}^{P_0+\mathcal{A}_1}\omega_2+\int_{P_0}^{P_0+\mathcal{A}_2}\omega_1 =0
\end{eqnarray*}
and the last equality follows from the normalization $\oint_{\mathcal{A}_k}\omega_j=\delta_{jk}$ of the canonical differentials. Hence $C(z)$ is either identically zero or 
it has no zeros at all. We now show that it cannot be identically zero; if this were the case,  we would have
\begin{equation}\label{as1}
 \Theta\left(a\begin{bmatrix}
                                                       1\\
-1\\
                                                      \end{bmatrix}+\mathfrak{u}(z)+\mathcal{K}\right)\equiv 0,\hspace{0.5cm}\forall a\in\mathbb{C},\ \forall z\in X.
\end{equation}
The Riemann surface under consideration has an involution $\mathfrak{j}:X\rightarrow X$ with
\begin{equation*}
  \mathfrak{j}(z,w) = \big(-z,\mathfrak{s}(z)w\big)
\end{equation*}
where $\mathfrak{s}(z):= \frac {\sqrt{R(z)}}{\sqrt{R(-z)}} \in\{\pm 1\}$; specifically  $\mathfrak s(z) =-1$ on the outside of the region bounded by $\bigcup_{j=1}^3 \gamma_j \cup \bigcup_{j=1}^3 (-1)\gamma_j$, and $\mathfrak s(z) =1$ inside (see Fig. \ref{fraks}). The function $\mathfrak{s}(z)$ accounts for the fact that $\sqrt{R(z)}:\mathbb{C}\backslash\mathcal{B}\rightarrow\mathbb{C}$, with $R(z)$ as in \eqref{R(z)} and the cuts of the square root as stipulated, is 
neither an even nor odd function. For any point $P\in X$
\begin{equation*}
  \mathfrak{u}(P)=\int_{a_1}^P\vec{\omega} = \int_{\mathfrak{j}a_1}^{\mathfrak{j}P}\mathfrak{j}\vec{\omega}=-\sigma_1\int_{a_4}^{\mathfrak{j}P}\vec{\omega}=
  -\sigma_1\mathfrak{u}(\mathfrak{j}P)+\sigma_1\mathfrak{u}(a_4),
\end{equation*}
and therefore
\begin{equation}\label{in1}
  \mathfrak{u}(P)+\mathfrak{u}(\mathfrak{j}P)=(I-\sigma_1)\mathfrak{u}(P)+\mathfrak{u}(a_4)\equiv b\begin{bmatrix}
                                                                                              1\\
-1\\
                                                                                             \end{bmatrix},\hspace{0.5cm}b\in\mathbb{C}.
\end{equation}
\begin{figure}
\resizebox{0.3\textwidth}{!}{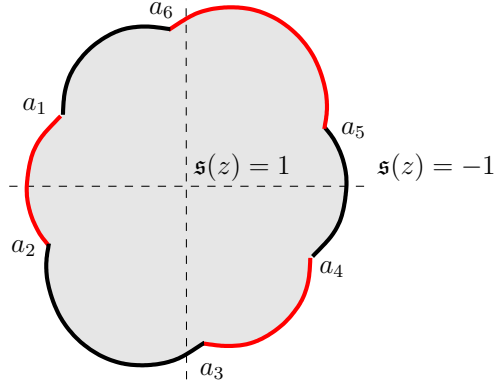}
\caption{The signature function $\mathfrak s(z) = \frac {\sqrt{R(-z)}}{\sqrt{R(z)}}$. }
\label{fraks}
\end{figure}
Now back to \eqref{as1} choose $z=a_1$ so that $\mathfrak{u}(z)=0$. Equation \eqref{in1} shows that vectors of the form $[a,-a]^t$ are images of symmetric 
divisors of degree $2$. However (compare Definition \ref{defspecial}) the special divisors of degree $2$ on $X$ are those that are invariant under the hyperelliptic involution and the only one that is also
invariant under $\mathfrak{j}$ is the divisor of the two points above $z=0$. Thus generically vectors of the form $[a,-a]^t$ are images of nonspecial divisors and the
theta function therefore not identically zero. Combined with the previous argument principle computation this shows that $C(z)\neq 0$ for any $z\in X$. The second statement {\bf [2]} follows from
$
  -\frac{1}{\mathbb{A}_{22}}\int_{a_1}^{\infty}\eta_2 = -\frac{\varkappa_2}{2}
$
and the periodicity of the Jacobi theta function.
\end{proof}

Now we go back to \eqref{zeroset} and obtain with Lemma \ref{cute} 
\begin{equation}
  \Theta\begin{bmatrix}
         \vec{\boldsymbol \alpha}\,\\
\vec{\boldsymbol \beta}\,\\
        \end{bmatrix}\left(-\mathcal{K}\right)=C(\infty^+)e^{2\pi i(\rho_n+\frac{1}{8}\langle\vec{\boldsymbol \alpha}{\boldsymbol \tau},\vec{\boldsymbol \alpha}\rangle-\frac{1}{2}
\langle\vec{\boldsymbol \alpha},\mathcal{K}\,\rangle+\frac{1}{4}\langle\vec{\boldsymbol \alpha},\vec{\boldsymbol \beta}\rangle)}\vartheta^2(\rho_n)
\label{g22g1}
\end{equation}
in other words the zeroset $\mathcal Z_n$ \eqref{zeroset} is equivalently determined by the requirement
\begin{equation}\label{zerosetneat}
 \vartheta\left(\frac{n}{2\pi i}\left[\oint_{\mathcal{B}_1}\d\phi+\varkappa_2\oint_{\mathcal{A}_1}\d\phi\right]
  +\frac{1}{2}\left[\int_{a_1}^0\frac{\eta_2}{\mathbb{A}_{22}}+\frac{\varkappa_2}{2}\right]\right)=0
\end{equation}
which only involves a Jacobi theta function corresponding to a Riemann surface of genus one.

\section{Completion of Riemann-Hilbert analysis - proof of theorems \ref{theo2} and \ref{theo3}}\label{sec4}
We combine the local parametrices and move on to the ratio problems. These are solved by standard small norm arguments and Neumann series expansions.
\subsection{Proof of Theorem \ref{theo2}} Recall the explicit construction of $M(z),U(z)$ and $V(z)$ in genus zero and define
\begin{equation}\label{ratio}
  {\mathcal E}(z) = S(z)\begin{cases}
                          \big(U(z)\big)^{-1},&|z+ia|<r\\
\big(V(z)\big)^{-1},&|z-ia|<r\\
\big(M(z)\big)^{-1},&|z\pm ia|>r.
                         \end{cases}
\end{equation}
This function has jumps on the contour $\Sigma_{\mathcal E}$ shown in Figure \ref{ratio1}, jumps which are given as ratios of parametrices 
\begin{equation*}
  {\mathcal E}_+(z)={\mathcal E}_-(z)U(z)\big(M(z)\big)^{-1},\ \ z\in C_1;\hspace{1cm}{\mathcal E}_+(z)={\mathcal E}_-(z)V(z)\big(M(z)\big)^{-1},\ \ z\in C_2
\end{equation*}
as well as conjugations with the outer model function
\begin{eqnarray*}
  {\mathcal E}_+(z)&=&{\mathcal E}_-(z)M(z)\begin{bmatrix}
                                      1 & 0\\
2\pi iz e^{n\varphi(z)}& 1\\
                                     \end{bmatrix}\big(M(z)\big)^{-1},\ \ z\in\hat{\mathcal{B}}^{\pm};\\
 {\mathcal E}_+(z)&=&{\mathcal E}_-(z)M(z)\begin{bmatrix}
1 & (2\pi iz)^{-1}e^{-n\varphi(z)}\\
0 & 1\\
\end{bmatrix}\big(M(z)\big)^{-1},\ \ z\in\hat{\mathcal{L}}.
\end{eqnarray*}

\begin{figure}[tbh]
\begin{center}
\resizebox{0.4\textwidth}{!}{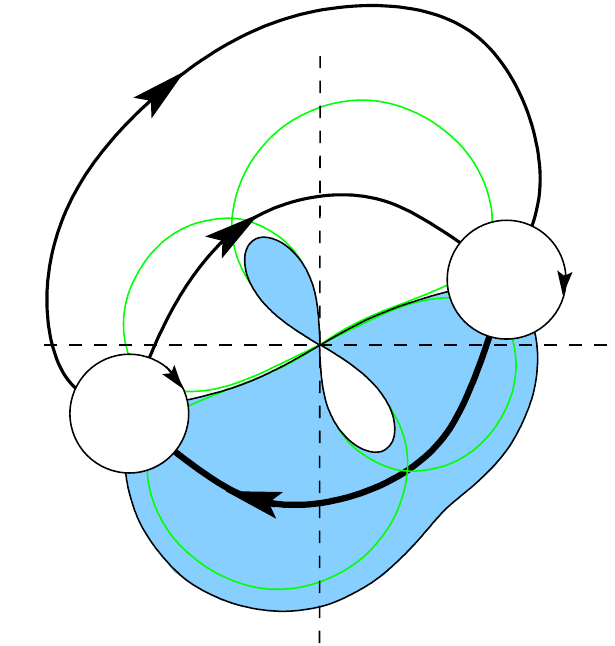}\hspace{2cm}\resizebox{0.4\textwidth}{!}{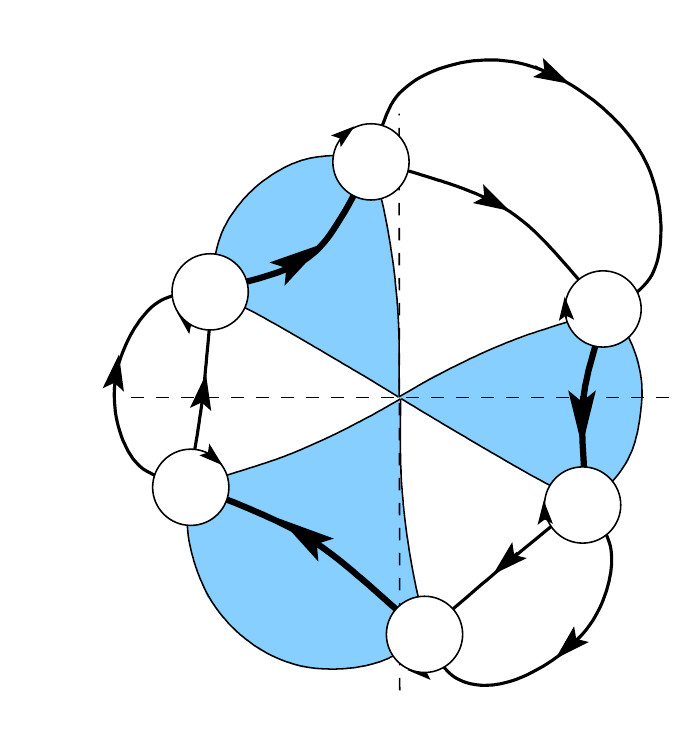}
\begin{minipage}{0.49\textwidth}
\caption{Jump contours in the ratio problem for ${\mathcal E}(z)$ as solid black lines - genus zero situation}
\label{ratio1}
\end{minipage}
\begin{minipage}{0.49\textwidth}
\caption{Jump contours in the ratio problem for ${\mathcal E}(z)$ as solid black lines - genus two situation}
\label{ratio2}
\end{minipage}
\end{center}
\end{figure}

Also the function ${\mathcal E}(z)$ is normalized as
\begin{equation*}
  {\mathcal E}(z)=I+\mathcal O\left(z^{-1}\right),\hspace{0.5cm}z\rightarrow\infty.
\end{equation*}
In terms of the previously derived estimates \eqref{esti1},\eqref{esti2},\eqref{match1} and \eqref{match2}, we conclude for the jump matrix $G_{\mathcal E}(z)$ in the latter ratio problem,
\begin{equation}\label{DZ1}
  \|G_{\mathcal E}-I\|_{L^2\cap L^{\infty}(\Sigma_{\mathcal E})}\leq \frac{c}{n},\hspace{0.5cm}n\rightarrow\infty,\ c>0
\end{equation}
which is uniform with respect to $x\in\mathbb{C}:\textnormal{dist}(x,\overline{\Delta})\geq\delta>0$. Hence (cf. \cite{DZ}) we can iteratively solve the singular integral equation 
\begin{equation*}
  {\mathcal E}_-(z) = I+\frac{1}{2\pi i}\int_{\Sigma_{\mathcal E}}{\mathcal E}_-(w)\big(G_{\mathcal E}(w)-I\big)\frac{\d w}{w-z_-},\hspace{0.5cm}z\in\Sigma_{\mathcal E}
\end{equation*}
in $L^2(\Sigma_{\mathcal E})$ which is in fact equivalent to the ${\mathcal E}$-RHP. Moreover its unique solution satisfies 
\begin{equation}\label{DZ2}
  \|{\mathcal E}_--I\|_{L^2(\Sigma_{\mathcal E})}\leq \frac{c}{n},\hspace{0.5cm}n\rightarrow\infty,\ c>0.
\end{equation}
As we have employed a series of invertible transformations
\begin{equation}\label{trace}
  \Gamma(z)\mapsto\Gamma^o(z)\mapsto Y(z)\mapsto S(z)\mapsto {\mathcal E}(z),
\end{equation}
the unique solvability of the ${\mathcal E}$-RHP as $n\rightarrow\infty$ for $x\in\mathbb{C}:\textnormal{dist}(x,\overline{\Delta})\geq\delta>0$ implies Theorem \ref{theo2} through \eqref{Drel}.
\subsection{Proof of Theorem \ref{theo3}}\label{secgen22} Fix $x\in\mathbb{C}:\textnormal{dist}(x,\mathbb{C}\backslash\overline{\Delta})\geq\delta>0$
away from the zeroset $\{x_k\}$ defined in \eqref{zeroset}. Now combine the outer parametrix $M(z)$ and the local ones $P_j(z)$ into the ratio function
\begin{equation*}
{\mathcal E}(z)= S(z)\begin{cases}
                         \big(P_j(z)\big)^{-1},& z\in D(a_j,r),\ j=1,\ldots,6\\
\big(M(z)\big)^{-1},&|z-a_j|>r.
                        \end{cases}
\end{equation*}
with $r>0$ sufficiently small. The ratio solves a Riemann-Hilbert problem with jumps on a contour as shown in Figure \ref{ratio2} below and is normalized as
\begin{equation*}
  {\mathcal E}(z)=I+\mathcal O\left(z^{-1}\right),\hspace{0.5cm}z\rightarrow\infty.
\end{equation*}
Since $M(z;x)$ is bounded on $\partial D(a_j,r)$ we use \eqref{esti3} and \eqref{matchgen2} to conclude
\begin{equation*}
  \|G_{\mathcal E}-I\|_{L^2\cap L^{\infty}(\Sigma_{\mathcal E})}\leq\frac{c}{n},\hspace{0.5cm}n\rightarrow\infty,\ c>0
\end{equation*}
which once more leads to the unique solvability of the ratio problem in the given situation. Tracing back the invertible transformations we get Theorem \ref{theo3}.

\section{Asymptotics for normalizing coefficients: proof of Theorem \ref{theo4}}\label{sec5}
In this section we extract expansions for $h_n(x)$ as $n\rightarrow\infty$ and compare the results to \cite{BM}.
\subsection{Expansions outside the star} We go back to \eqref{orel} and trace back the transformations
\begin{equation*}
  h_n^o(x)=-2\pi i\lim_{z\rightarrow\infty}z\Big(\Gamma^o(z)z^{-n\sigma_3}-I\Big)_{12},\hspace{0.5cm}\Gamma^o(z)z^{-n\sigma_3}
 = e^{\frac{n\ell}{2}\sigma_3}{\mathcal E}(z)M(z)e^{n(g(z)-\frac{\ell}{2}-\ln z)\sigma_3},\hspace{0.5cm}z\rightarrow\infty.
\end{equation*}
For $x\in\mathbb{C}:\textnormal{dist}(x,\overline{\Delta})\geq\delta>0$ we have
\begin{equation*}
 g(z) =\ln z-\frac{x}{2z}+\mathcal O\left(z^{-2}\right),\hspace{1cm}
  M(z)=I+\frac{a}{2z}\begin{bmatrix}
                                   -1 & (a\pi i)^{-1}\\
-a\pi i& 1\\
                                  \end{bmatrix}+\mathcal O\left(z^{-2}\right)
\end{equation*}
as $z\rightarrow\infty$ and this combined with
\begin{equation*}
  {\mathcal E}(z)=I+\frac{i}{2\pi z}\int_{\Sigma_{\mathcal E}}{\mathcal E}_-(w)\big(G_{\mathcal E}(w)-I\big)\,\d w+\mathcal O\left(z^{-2}\right)
\end{equation*}
leads us to
\begin{eqnarray*}
  \Gamma^o(z)z^{-n\sigma_3}-I &=& \frac{e^{\frac{n\ell}{2}\sigma_3}}{z}\Bigg\{\frac{a}{2}\begin{bmatrix}
                                   -1 & (a\pi i)^{-1}\\
-a\pi i& 1\\
                                  \end{bmatrix}-\frac{x\sigma_3}{2}\\
&&+\frac{i}{2\pi}\int_{\Sigma_{\mathcal E}}{\mathcal E}_-(w)\big(G_{\mathcal E}(w)-I\big)\,\d w+\mathcal O\left(z^{-1}\right)\Bigg\}
e^{-\frac{n\ell}{2}\sigma_3}.
\end{eqnarray*}
Thus
\begin{equation}\label{hnexp1}
  h_n^o(x) = -e^{n\ell}\Big\{1+\mathcal O\left(n^{-1}\right)\Big\},\hspace{0.5cm}n\rightarrow\infty
\end{equation}
where we used \eqref{DZ1} and \eqref{DZ2} and which is uniform with respect to $x\in\mathbb{C}:\textnormal{dist}(x,\overline{\Delta})\geq\delta>0$. Now combine \eqref{hnexp1} with
\eqref{Lagra},\eqref{orel} and \eqref{PIIconnec} and obtain
\begin{corollary}
Let $x\in\mathbb{C}:\textnormal{dist}(x,\overline{\Delta})\geq\delta>0$. Then for sufficiently large $n$ the rational solutions $u(x;n)$ to PII equation
\eqref{PII} satisfy
\begin{equation}\label{Pexp}
  u(x;n) = -n^{\frac{1}{3}}\frac{a'\big(n^{-\frac{2}{3}}x\big)}{a\big(n^{-\frac{2}{3}}x\big)}\left(1+\frac{1}{2a^3\big(n^{-\frac{2}{3}}x\big)}\right)+\mathcal O\left(n^{-1}\right),\hspace{0.5cm}n\rightarrow\infty
\end{equation}
where $a=a(x)$ is the unique solution to the cubic equation \eqref{c1} subject to the condition \eqref{c2}.
\end{corollary}
\br\label{remBM1}
If we substitute the large argument expansion of $a(x)$ into \eqref{Pexp}, we easily obtain for $x\in\mathbb{C}:\textnormal{dist}(x,\overline{\Delta})\geq\delta>0$
\begin{equation}\label{exp1}
 u(x;n) = -\left(\frac{n}{2}\right)^{\frac{1}{3}}\left(1+\mathcal O\left(n^{-\frac{2}{3}}\right)\right),\hspace{0.5cm}n\rightarrow\infty.
\end{equation}
On the other hand the rational solutions $\mathcal{P}_m(\xi)$ (to a rescaled PII equation) in \cite{BM} are shown to satisfy the following large $m$-behavior
\begin{equation}\label{exp2}
  \mathcal P_m(\xi) = m^{\frac{1}{3}}\left\{\dot{\mathcal{P}}\left(\left(m-\frac{1}{2}\right)^{-\frac{2}{3}}\xi\right)+\mathcal O\left(m^{-1}\right)\right\},\hspace{0.5cm}m\rightarrow\infty
\end{equation}
outside the corresponding star shaped region in the complex $\xi$-plane, compare Remark \ref{rem1}. Here $\dot{\mathcal{P}}(\xi) = -\frac{1}{2}S(\xi)$, where $S=S(\xi)$ solves
the cubic equation
\begin{equation*}
  3S^3+4\xi S+8=0,\hspace{0.5cm} S(\xi) = -\frac{2}{\xi}+\mathcal O\left(\xi^{-4}\right),\ \ \xi\rightarrow\infty.
\end{equation*}
The relation between $\mathcal{P}_m(\xi)$ and $u(x;n)$ is as follows
\begin{equation*}
 u(x;n) = -\left(\frac{3}{2}\right)^{\frac{1}{3}}\mathcal P_n\left(\left(\frac{3}{2}\right)^{\frac{1}{3}}\xi\right),\hspace{0.5cm}\xi=(12)^{\frac{1}{3}}x
\end{equation*}
and we recall that
\begin{equation*}
 S(\xi) = -\left(\frac{2}{3}\right)^{\frac{1}{3}}\frac{1}{a(\xi)}.
\end{equation*}
Substituting the latter into $\dot{\mathcal{P}}(\xi) = -\frac{1}{2}S(\xi)$ and using \eqref{exp2} we verify that
\begin{equation*}
  -\left(\frac{3}{2}\right)^{\frac{1}{3}}\mathcal P_n\left(\left(\frac{3}{2}\right)^{\frac{1}{3}}\xi\right) = -\left(\frac{n}{2}\right)^{\frac{1}{3}}
\left(1+\mathcal O\left(n^{-\frac{2}{3}}\right)\right),
\hspace{0.5cm}n\rightarrow\infty
\end{equation*}
hence \eqref{exp1} matches \eqref{exp2} to leading order.
\er
\subsection{Expansions inside the star}
 For $x\in\mathbb{C}:\textnormal{dist}(x,\mathbb{C}\backslash\overline{\Delta})\geq\delta>0$ away from the zeroset $\mathcal Z_n$ \eqref{zeroset} we
have as $z\rightarrow\infty$
\begin{equation*}
  g(z) = \ln z-\frac{x}{2z}+\mathcal O\left(z^{-2}\right),\hspace{0.5cm}M(z) = I+\frac{M_1}{z}+\mathcal O\left(z^{-2}\right)
\end{equation*}
involving
\begin{equation*}
  M_1 = e^{-i\frac{\pi}{4}\sigma_3}(2\pi i)^{-\frac{1}{2}\sigma_3}\big(\mathcal{D}(\infty)\big)^{-\sigma_3}Q_1\big(\mathcal{D}(\infty)\big)^{\sigma_3}(2\pi i)^{\frac{1}{2}\sigma_3}
  e^{i\frac{\pi}{4}\sigma_3}.
\end{equation*}
Since
\begin{equation*}
 \mathcal{D}(\infty) = \exp\left[-\frac{1}{2\pi i}\left(\sum_{j=1}^3\int_{a_{2j-1}}^{a_{2j}}\frac{w^2\ln w}{\sqrt{R(w)}_+}\d w-
  \sum_{j=1}^2\int_{a_{2j}}^{a_{2j+1}}\frac{i\pi\delta_j w^2}{\sqrt{R(w)}_+}\d w\right)\right]\neq 0 
\end{equation*}
we can continue with
\begin{equation*}
  \Gamma^o(z)z^{-n\sigma_3}-I = \frac{e^{\frac{n\ell}{2}\sigma_3}}{z}\left\{M_1-\frac{x\sigma_3}{2}+\frac{i}{2\pi}\int_{\Sigma_{\mathcal E}}{\mathcal E}_-(w)\big(G_{\mathcal E}(w)-I\big)\,\d w+
  \mathcal O\left(z^{-1}\right)\right\}e^{-\frac{n\ell}{2}\sigma_3}
\end{equation*}
and thus obtain the following analogue to \eqref{hnexp1} inside the star (recall the change of orientation in genus two)
\begin{equation}\label{hnexp2}
  \big(h_{n-1}^o(x)\big)^{-1} = ie^{-n\ell}\big(\mathcal{D}(\infty)\big)^{2}\Big\{Q_1^{21}+\mathcal O\left(n^{-1}\right)\Big\},\hspace{0.5cm}n\rightarrow\infty,
\end{equation}
where $Q_1^{21}$ is given in \eqref{coeff}. Here the leading coefficient $Q_1^{21}$ is written in terms of theta functions on a genus two hyperelliptic Riemann surface. 

Using Lemma \ref{cute} and along the same lines as \eqref{g22g1} we rewrite $Q_1^{21}$ as
\begin{equation*}
  Q_1^{21}=\frac{C_0^{-1}e^{2\pi i(\langle {\bf e}_1+\vec{\boldsymbol \alpha},\mathfrak{u}(\infty)\rangle}}{\Theta\big(2\mathfrak{u}(\infty)-\mathfrak{u}(a_6)-\mathcal{K}\big)}
  \frac{\Theta\left(\rho_n\begin{bmatrix}
                      1\\
-1\\
                     \end{bmatrix}+3\mathfrak{u}(\infty)+\mathcal{K}-\frac{1}{2}(\varkappa_2+1)\begin{bmatrix}
1\\
-1\\
\end{bmatrix}\right)
}{\Theta\left(\rho_n\begin{bmatrix}
                      1\\
-1\\
                     \end{bmatrix}+\mathfrak{u}(\infty)+\mathcal{K}-\frac{1}{2}(\varkappa_2+1)\begin{bmatrix}
1\\
-1\\
\end{bmatrix}\right)}
\end{equation*}
and therefore in \eqref{hnexp2}
\begin{equation}
\label{h0n1}
  \big(h_{n-1}^o(x)\big)^{-1} = ie^{-n\ell}\big(C(\infty^+)\big)^{-1}C_0^{-1}e^{2\pi i(\rho_n+\langle{\bf e}_1+\vec{\boldsymbol \alpha},\mathfrak{u}(\infty)\rangle)}
\left\{\frac{T(\rho_n)+\mathcal O\left(n^{-1}\right)}{\vartheta^2(\rho_n)}\right\}
\end{equation}
as $n\rightarrow\infty$ away from the zeroset $\mathcal{Z}_n$ determined in \eqref{zerosetneat}. We introduced 
\begin{equation}\label{Teq}
 T(\rho_n) = \frac{\Theta\left(
 \le(\rho_n - \frac {\varkappa_2+1}2\ri) \begin{bmatrix}
                      1\\
-1\\
                     \end{bmatrix}+3\mathfrak{u}(\infty)+\mathcal{K}
                     \right)}{\Theta\big(2\mathfrak{u}(\infty)-\mathfrak{u}(a_6)-\mathcal{K}\big)}
\end{equation}
The formula \eqref{h0n1} is our fundamental pivot to analyze the location of the zeros of $\wh Q_n (x) =\mathcal Q_n(n^\frac  23 x)$; indeed we remind the reader that 
\be
 \big(h_{n-1}^o(x)\big)^{-1} = \le(\frac {\wh Q_{n-1}(x)}{\wh Q_n(x)}\ri)^2.
\ee
The error term in the numerator of \eqref{h0n1} prevents us from localizing the zeros of $\wh Q_{n-1}$; however we can detect those of $\wh Q_n$ because they appear as poles of $h^o_{n-1}(x)$. In particular the poles of $h^o_{n-1}(x)$ must be of second order, which is automatically guaranteed in our approximation \eqref{h0n1} by the fact that the denominator is a square.

We shall thus verify (Proposition \ref{nocommon} below) that the zeros of the leading approximation $T(\rho_n)$ never coincide with the denominator's. Then, using the argument principle on a small circle around a point of $\mathcal Z_n$ we shall see that indeed the function $(h^o_{n-1}(x))^{-1}$ has a double pole within the enclosed disk.


\begin{proposition}
\label{nocommon} The functions $\vartheta(z),z\in\mathbb{C}$ and $T(z),z\in\mathbb{C}$ have no common roots.
\end{proposition}
\begin{proof} The roots of $\vartheta(z)$ are located at $z^{\ast}\equiv\frac{1}{2}(1+\varkappa_2)\mod(\mathbb{Z}+\varkappa_2\mathbb{Z})$, or equivalently 
(compare Lemma \ref{cute} and Corollary \ref{generalThetadiv}), we have for some $P_0\in X$
\begin{equation*}
  \mathfrak{u}(P_0)\equiv\mathfrak{u}(\infty^+)+\mathcal{K}\ \ \mod(\mathbb{Z}^2+{\boldsymbol \tau}\mathbb{Z}^2).
\end{equation*}
But at the points $z=z^{\ast}$ the numerator in \eqref{Teq} is proportional to
\begin{equation*}
  \Theta\big(\mathfrak{u}(P_0)+\mathcal{K}+2\mathfrak{u}(\infty)\big) = \Theta\big(\mathfrak{u}(P_0)+\mathfrak{u}(\infty^+)-\mathfrak{u}(\infty^-)+\mathcal{K}\big)
\end{equation*}
and vanishes precisely if $P_0=\infty^-$. But then we would have
\begin{equation*}
  \mathfrak{u}(\infty^-)\equiv \mathfrak{u}(\infty^+)+\mathcal{K}\ \ \Leftrightarrow\ \ 2\mathfrak{u}(\infty^-) \equiv \mathfrak{u}(a_3)+\mathfrak{u}(a_5)\ \ 
\mod(\mathbb{Z}^2+{\boldsymbol \tau}\mathbb{Z}^2)
\end{equation*}
and in the last equality both sides are equal to the Abel map of a nonspecial divisor of degree $2$. However the genus of $X$ is $g=2$ and the Abel map is one-to-one on the
set of nonspecial divisors of degree two, hence both sides cannot be the same. Thus $\vartheta(z)$ cannot be zero at the same time as $T(z)$.
\end{proof}
In order to detect poles of $(h^o_{n-1})^{-1}$ in \eqref{h0n1} we shall use the argument principle by tracking the increment of the argument as $x$ makes a  small loop around a point of the zeroset $\mathcal Z_n$ \eqref{zerosetneat}. There are two salient points worth mentioning here;
\begin{enumerate}
\item the approximation \eqref{h0n1} is  a uniform approximation of the holomorphic function $(h^o_{n-1})^{-1}(x)$ by a {\em smooth} function of $x$;
\item the circle used in the detection of the poles must not contain any zero of the leading term approximation.
\end{enumerate}
The first point follows from the fact that the conditions \eqref{gen2con3} that determine the branchpoints of the Riemann surface $X$ are real--analytic constraints. Nonetheless the argument principle can be used because the approximation is uniform.\smallskip

In regard to the second point, the strategy is as follows; we shall prove that $\rho_n(x)$ given by \eqref{simple} is a  locally smooth function from $\C \simeq \R^2$ to $\C \simeq \R^2$. Therefore, if $x$ makes a small loop around a point $x_\star$, then so does $\rho_n (x)$ around $\rho_n(x_\star)$. If the loop is chosen sufficiently small around a point of $\mathcal Z_n$ we can exclude the zeros of $T(\rho)$ because by Prop. \ref{nocommon} the zeros of $T(\rho)$ and $\vartheta(\rho)$ never coincide and thus the argument of \eqref{h0n1} has the same increment as the argument of the denominator $\vartheta^2(\rho)$, thus proving that $(h^o_{n-1}(x))^{-1}$ (which is a priori a meromorphic function) must have a double pole within the loop in the $x$-plane.\smallskip

We thus now recall that \eqref{zerosetneat} holds iff
\begin{equation*}
  \rho_n=\rho_n(x)=\frac{n}{2\pi i}\left[\oint_{\mathcal{B}_1}\d\phi+\varkappa_2\oint_{\mathcal{A}_1}\d\phi\right]
  +\frac{1}{2}\left[\int_{a_1}^0\frac{\eta_2}{\mathbb{A}_{22}}+\frac{\varkappa_2}{2}\right]\equiv \frac{1}{2}(1+\varkappa_2)\ \ \mod(\mathbb{Z}+\varkappa_2\mathbb{Z})
\end{equation*}
in other words iff we choose $x=x_{n,j,k}$ in such a way that
\begin{equation*}
  \frac{n}{2\pi i}\left[\oint_{\mathcal{B}_1}\d\phi+\varkappa_2\oint_{\mathcal{A}_1}\d\phi\right]+\frac{1}{2}\left[\int_{a_1}^0\frac{\eta_2}{\mathbb{A}_{22}}-\frac{\varkappa_2+2}{2}\right]=
j+\varkappa_2k,\ \ j,k\in\mathbb{Z}.
\end{equation*}
We aim at showing that $\rho_n(x)$ makes a loop around $\rho_n(x_{n,j,k})$ as $x$ makes a loop around  $x_{n,j,k}$. For this fix
$x\in\mathbb{C}:x-x_{n,j,k}=\frac{\epsilon}{n},\epsilon\in\mathbb{C}$ with $|\epsilon|>0$ sufficiently small and consider
\begin{equation*}
 \Xi=\Xi(x)=\rho_n(x)-\rho_n(x_{n,j,k}) = \frac{n}{2\pi i}\left[\oint_{\mathcal{B}_1}\d\phi+\varkappa_2\oint_{\mathcal{A}_1}\d\phi\right]+\frac{1}{2}\left[\int_{a_1}^0\frac{\eta_2}{\mathbb{A}_{22}}-\frac{\varkappa_2+2}{2}\right]
-j-\varkappa_2k,
\end{equation*}
i.e. we need to show that $\Xi$ makes a loop around the origin. This will be achieved by evaluating the Jacobian of the mapping $\Xi=\Xi(u,v)$ with $x=u(\Re\epsilon,\Im\epsilon)+iv(\Re\epsilon,\Im\epsilon),u,v\in\mathbb{R}$ at $\epsilon=0$. Put
\begin{equation*}
  A(\epsilon)=\frac{1}{2\pi i}\oint_{\mathcal{A}_1}\d\phi,\hspace{0.5cm}B(\epsilon)=\frac{1}{2\pi i}\oint_{\mathcal{B}_1}\d\phi
\end{equation*}
and notice that $A,B\in\mathbb{R}$. Now any point in the complex plane can be written as $b+a\varkappa_2,a,b\in\mathbb{R}$, hence
\begin{equation*}
  \Xi=n\big(B(\epsilon)+\varkappa_2(\epsilon)A(\epsilon)\big)-\left(j+\frac{1}{4}+b(\epsilon)\right)-\left(k+\frac{1}{4}+a(\epsilon)\right)\varkappa_2(\epsilon).
\end{equation*}
We recall (compare Section \ref{ggen2}) that the differential $\d\phi$ is the unique meromorphic differential on $X$ such that
\begin{eqnarray*}
	\Re\left(\oint_{\gamma}\d\phi\right)=0\ \ \forall\,\gamma\in H_1(X,\mathbb{Z});\ \ \ \ \d\phi(z)&=&\pm\frac{1}{2}\left(\frac{1}{z^4}-\frac{x}{z^2}+\mathcal O(1)\right)\d z,\ \ z\rightarrow 0^{\pm}\\d\phi(z)&=&\pm\frac{1}{2}\left(z^{-1}+\mathcal O\left(z^{-2}\right)\right)\d z,\ \ z\rightarrow\infty^{\pm}.
\end{eqnarray*}
Hence, $\partial_u\d\phi$ and $\partial_v\d\phi$ are the unique meromorphic differentials on $X$ with a double pole at $z=0^{\pm}$, vanishing residues, purely imaginary periods and behavior $\mathcal O\left(z^{-2}\right)$ as $z\rightarrow\infty^{\pm}$. In order to construct them explicitly, we consider (as a function on the universal covering of $X$)
\begin{equation*}
 G(z) = \mathbb{A}_{11}\sqrt{R(z)}\,\frac{\d}{\d z}\ln\vartheta_1\left(\int_0^z\frac{\eta_1}{\mathbb{A}_{11}}\,\Big|\varkappa_1\right).
\end{equation*}
As $z$ varies on $X$, notice that
\begin{align*}
 z\mapsto z_{\gamma}\ \ \textnormal{along}\ \mathcal{A}_1:\ \ & G(z_{\gamma})=G(z);&& z\mapsto z_{\gamma}\ \ \textnormal{along}\ \mathcal{A}_2: && G(z_{\gamma}) = G(z);\\
 z\mapsto z_{\gamma}\ \ \textnormal{along}\ \mathcal{B}_1:\ \ & G(z_{\gamma})= G(z)-2\pi i;&&
  z\mapsto z_{\gamma}\ \ \textnormal{along}\ \mathcal{B}_2: && G(z_{\gamma})=G(z)-2\pi i,
\end{align*}
and thus
\begin{equation*}
	\partial_u\d\phi=\frac{\d G(z)}{\mathbb{A}_{11}}+2\pi i\frac{\Im(\mathbb{A}_{11}^{-1})}{\Im\varkappa_1}\frac{\eta_1}{\mathbb{A}_{11}},\hspace{1cm}
	\partial_v\d\phi=\frac{i\d G(z)}{\mathbb{A}_{11}}+2\pi i\frac{\Re(\mathbb{A}_{11}^{-1})}{\Im\varkappa_1}\frac{\eta_1}{\mathbb{A}_{11}}.
\end{equation*}
We also compute
\begin{eqnarray*}
	\frac{1}{2\pi i}\oint_{\mathcal{A}_1}\partial_u\d\phi=\frac{\Im(\mathbb{A}_{11}^{-1})}{\Im\varkappa_1}&&\frac{1}{2\pi i}\oint_{\mathcal{B}_1}\partial_u\d\phi =-\Re(\mathbb{A}_{11}^{-1})+\Im(\mathbb{A}_{11}^{-1})\frac{\Re\varkappa_1}{\Im\varkappa_1}\\
	\frac{1}{2\pi i}\oint_{\mathcal{A}_1}\partial_v\d\phi=\frac{\Re(\mathbb{A}_{11}^{-1})}{\Im\varkappa_1}&&\frac{1}{2\pi i}\oint_{\mathcal{B}_1}\partial_v\d\phi =\Im(\mathbb{A}_{11}^{-1})+\Re(\mathbb{A}_{11}^{-1})\frac{\Re\varkappa_1}{\Im\varkappa_1}
\end{eqnarray*}
and obtain therefore
\begin{equation*}
	\det\begin{bmatrix}
	\partial_uA &\partial_uB\\
	\partial_vA &\partial_vB\\
	\end{bmatrix}=\frac{|\mathbb{A}_{11}^{-1}|^2}{\Im\varkappa_1}> 0.
\end{equation*}
The Jacobian of the mapping
\begin{equation*}
	\big(\Re\epsilon,\Im\epsilon\big)\mapsto\Big(\Re\big(\Xi(u,v)\big),\Im\big(\Xi(u,v)\big)\Big)
\end{equation*}
equals
\begin{equation*}
	J(\epsilon)=\frac{1}{n^2}\det\begin{bmatrix}
	\partial_u\Re\Xi &\partial_v\Re\Xi\\
	\partial_u\Im\Xi &\partial_v\Im\Xi\\
	\end{bmatrix}.
\end{equation*}
Hence at $\epsilon=0$,
\begin{eqnarray*}
	J(0)&=&\det\left(\begin{bmatrix}
	B_u+A_u\Re\varkappa_2 & B_v+A_v\Re\varkappa_2\\
	A_u\Im\varkappa_2 & A_v\Im\varkappa_2\\
	\end{bmatrix}-\frac{1}{n}\begin{bmatrix}
	b_u+a_u\Re\varkappa_2 & b_v+a_v\Re\varkappa_2\\
	a_u\Im\varkappa_2 & a_v\Im\varkappa_2\\
	\end{bmatrix}
	\right)\\
	&=&\Im\varkappa_2\det\begin{bmatrix}
	B_u& B_v\\
	A_u&A_v\\
	\end{bmatrix}+\mathcal O\left(n^{-1}\right)=-|\mathbb{A}_{11}^{-1}|^2\frac{\Im\varkappa_2}{\Im\varkappa_1}\big(1+\mathcal O\left(n^{-1}\right)\big)
\end{eqnarray*}
which shows that we can find a sufficiently small $r_0>0$ which is $n$ independent such that the small circle
\begin{equation*}
	x=x_{n,j,k}+\frac{r_0}{n}e^{i\alpha},\ \ \ \alpha\in[0,2\pi)
\end{equation*}
is mapped smoothly onto a curve in the $\Xi$-plane, around the origin with a diameter that is bounded with respect to $n$. By choosing $r_0$ sufficiently small we can thus guarantee that no zeros of $T(\rho)$ are included. 
Then the total increment of the argument in the leading approximation \eqref{h0n1} is solely determined by the denominator $\vartheta^2(\rho)$; this proves that indeed the function $(h^o_{n-1})^{-1}(x)$ has a pole in a $1/n$ neighborhood of  the zeroset $\mathcal Z_n$ \eqref{zerosetneat} and completes the proof of Theorem \ref{theo4}.

\begin{appendix}
\section{Airy parametrices}\label{app1}
Our constructions in Subsection \eqref{RHP0} make use of certain piecewise analytic functions which are constructed out of a Wronskian matrix. On the technical level 
(we use here the identical construction of \cite{BB}), introduce
\begin{equation}\label{pl:2}
	A_0(\zeta) = \begin{bmatrix}
	\frac{\d}{\d\zeta}\textnormal{Ai}(\zeta) & e^{i\frac{\pi}{3}}\frac{\d}{\d\zeta}\textnormal{Ai}\Big(e^{-i\frac{2\pi}{3}}\zeta\Big)\\
	\textnormal{Ai}(\zeta) & e^{i\frac{\pi}{3}}\textnormal{Ai}\Big(e^{-i\frac{2\pi}{3}}\zeta\Big)\\
	\end{bmatrix},\hspace{0.5cm}\zeta\in\mathbb{C}
\end{equation}
where $\textnormal{Ai}(\zeta)$ the solution to Airy's equation
\begin{equation*}
	w''=zw
\end{equation*}
uniquely determined by its asymptotics as $\zeta\rightarrow\infty$ and $-\pi<\textnormal{arg}\ \zeta<\pi$
\begin{equation*}
	\textnormal{Ai}(\zeta) = \frac{\zeta^{-1/4}}{2\sqrt{\pi}}e^{-\frac{2}{3}\zeta^{3/2}}\left(1-\frac{5}{48}\zeta^{-3/2}+\frac{385}{4608}\zeta^{-6/2}+O\left(\zeta^{-9/2}\right)\right).
\end{equation*}
Next assemble the model function
\begin{equation}\label{pl:3}
		A^{RH}(\zeta) = \left\{
                                 \begin{array}{ll}
                                   A_0(\zeta), & \hbox{arg $\zeta\in(0,\frac{2\pi}{3})$,} \smallskip\\
                                   A_0(\zeta)\begin{bmatrix}
                          1 & 0 \\
                          -1 & 1 \\
                        \end{bmatrix}, & \hbox{arg $\zeta\in(\frac{2\pi}{3},\pi)$,} \bigskip \\
                                   A_0(\zeta)\begin{bmatrix}
                                   1 & -1\\
                                   0 & 1\\
                                   \end{bmatrix},& \hbox{arg $\zeta\in(-\frac{2\pi}{3},0)$,} \bigskip \\
                                   A_0(\zeta)\begin{bmatrix}
                                   0 & -1 \\
                                   1 & 1\\
                                   \end{bmatrix}, & \hbox{arg $\zeta\in(-\pi,-\frac{2\pi}{3})$,}
                                 \end{array}
                               \right.
\end{equation}
which solves the RHP with jumps for $\textnormal{arg}\,\zeta=-\pi,-\frac{2\pi}{3},0,\frac{2\pi}{3}$ as depicted in Figure \ref{Airy1}.
\begin{figure}[tbh]
\begin{center}
\resizebox{0.3\textwidth}{!}{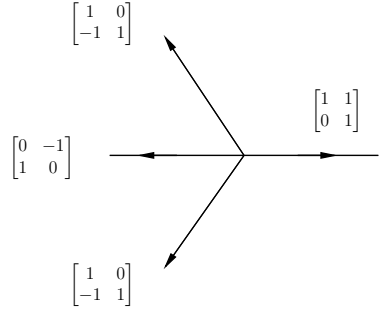}
\caption{A jump behavior which can be modeled explicitly in terms of Airy functions}
\label{Airy1}
\end{center}
\end{figure}
Besides the indicated jump behavior we also have an expansion as $\zeta\rightarrow\infty$ which is valid in a full neighborhood of infinity:
\begin{equation}\label{pl:4}
		A^{RH}(\zeta)=\frac{\zeta^{\sigma_3/4}}{2\sqrt{\pi}}\begin{bmatrix}
		-1 & i\\
		1 & i\\
		\end{bmatrix}\Bigg\{I+\frac{1}{48\zeta^{3/2}}\begin{bmatrix}
	1 & 6i\\
	6i & -1\\
	\end{bmatrix}+\mathcal O\left(\zeta^{-6/2}\right)\Bigg\}e^{-\frac{2}{3}\zeta^{3/2}\sigma_3}.
\end{equation}
Next we construct out of \eqref{pl:2} the function
\begin{equation*}
	\tilde{A}_0(\zeta)=-\bigl[\begin{smallmatrix}
	0 & 1\\
	1 & 0\\
	\end{smallmatrix}\bigr]\sigma_3A_0\big(e^{-i\pi}\zeta\big)\sigma_3,\hspace{0.5cm}\zeta\in\mathbb{C}
\end{equation*}
and then assemble
\begin{equation}\label{tildeARH}
	\tilde{A}^{RH}(\zeta) = \left\{
                                 \begin{array}{ll}
                                   \tilde{A}_0(\zeta)\begin{bmatrix}
                                   0 & 1\\
                                   -1 & 1\\
                                   \end{bmatrix}, & \hbox{arg $\zeta\in(0,\frac{\pi}{3})$,} \smallskip\\
                                   \tilde{A}_0(\zeta)\begin{bmatrix}
                          1 & 1 \\
                          0 & 1 \\
                        \end{bmatrix}, & \hbox{arg $\zeta\in(\frac{\pi}{3},\pi)$,} \smallskip \\
                                   \tilde{A}_0(\zeta),& \hbox{arg $\zeta\in(\pi,\frac{5\pi}{3})$,} \smallskip \\
                                   \tilde{A}_0(\zeta)\begin{bmatrix}
                                   1 & 0 \\
                                   1 & 1\\
                                   \end{bmatrix}, & \hbox{arg $\zeta\in(\frac{5\pi}{3},2\pi)$.}
                                 \end{array}
                               \right.
\end{equation}
This model function solves again a RHP with jumps on the rays $\textnormal{arg}\,\zeta=0,\frac{\pi}{3},\pi,\frac{5\pi}{3}$ (indicated in Figure \ref{Airy2}) and we have the uniform
expansion
\begin{equation}\label{tildeARHasy}
		\tilde{A}_{RH}(\zeta)=\frac{\big(e^{-i\pi}\zeta\big)^{-\sigma_3/4}}{2\sqrt{\pi}}\begin{bmatrix}
	1 & -i\\
	1 & i\\
	\end{bmatrix}\Bigg\{I+\frac{i}{48\zeta^{3/2}}\begin{bmatrix}
	-1 & 6i\\
	6i & 1\\
	\end{bmatrix}+\mathcal O\left(\zeta^{-6/2}\right)\Bigg\}
	 e^{-\frac{2}{3}i\zeta^{3/2}\sigma_3},\ \ \zeta\rightarrow\infty.
\end{equation}
\begin{figure}
\begin{center}
\resizebox{0.25\textwidth}{!}{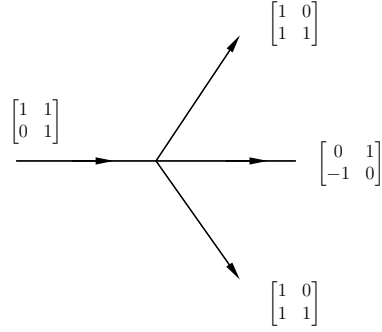}
\caption{Another jump behavior which can be modeled in terms of Airy functions}
\label{Airy2}
\end{center}
\end{figure}

\section{Some basic facts about theta functions and divisors}\label{app2}
The reference for all the following theorems is \cite{FK}, we quote here certain results about general Riemann surfaces of the genus $g\in \N$.\bigskip

The Riemann theta function, associated with a symmetric matrix ${\boldsymbol \tau}$ that has a strictly positive imaginary part,  
is the function of the vector argument $\vec z\in\C^g$ given by
\be\label{Rtheta}
\Theta(\vec z\,|{\boldsymbol \tau})= \sum_{\vec k\in \Z^g} \exp\left[i\pi \langle\vec{k}{\boldsymbol \tau},\vec{k}\,\rangle +2\pi i \langle\vec{k},\vec{z}\,\rangle\right].
\ee
Often the dependence on ${\boldsymbol \tau}$ is  omitted from the notation.
\bp
\label{thetaproperties}
The theta function has the following properties:
\begin{enumerate}
\item $\Theta(\vec z\,|{\boldsymbol \tau})  = \Theta(-\vec z\,|{\boldsymbol \tau})$ (parity);
\item For any $\vec{\lambda},\vec{\mu}\in \Z^g$ we have
\be
\Theta(\vec z + \vec{\mu} + {\boldsymbol \tau}\vec{\lambda}\,|{\boldsymbol \tau}) = \exp\Big[ -2\pi i \langle\vec{\lambda},\vec z\,\rangle - i\pi\langle  \vec{\lambda} {\boldsymbol \tau},\vec{\lambda}\,\rangle 
\Big]
\Theta(\vec z\,|{\boldsymbol \tau}).\label{thetaperiods}
\ee
\end{enumerate}
\ep
In addition to \eqref{Rtheta} we also use the theta function with characteristics $\vec{\boldsymbol \alpha},\vec{\boldsymbol \beta}\in\mathbb{C}^g$
\begin{equation}\label{cRtheta}
 \Theta\begin{bmatrix}
        \vec{\boldsymbol \alpha}\,\\
\vec{\boldsymbol \beta}\,\\
       \end{bmatrix}
(\vec z\,|{\boldsymbol \tau})= \exp \le[ 2\pi i \le(
\frac 1 8 \langle\vec{\boldsymbol \alpha}{\boldsymbol \tau},\vec{\boldsymbol \alpha}\,\rangle + \frac 1 2 \langle\vec{\boldsymbol \alpha},\vec{z}\,\rangle + \frac 1 4 \langle\vec{\boldsymbol \alpha},\vec{\boldsymbol \beta}\,\rangle
\ri)\ri] \Theta \le( \vec{z} + \frac 1 2 \vec{\boldsymbol \beta} + \frac {\boldsymbol \tau} 2 \vec{\boldsymbol \alpha}\,\bigg|{\boldsymbol \tau} \ri)
\end{equation}

\bp
\label{propB2}
The theta function with characteristics $\vec{\boldsymbol \alpha},\vec{\boldsymbol \beta}\in\mathbb{C}^g$ has the properties
\begin{equation*}
 \Theta\begin{bmatrix}
        \vec{\boldsymbol \alpha}\\
\vec{\boldsymbol \beta}\\
       \end{bmatrix}(\vec{z}+\vec{\mu}+{\boldsymbol \tau}\vec{\lambda}\,|{\boldsymbol \tau}) = \exp\left[2\pi i\left(\frac{1}{2}\Big(\langle \vec{\boldsymbol \alpha},\vec{\mu}\,\rangle
  -\langle\vec{\lambda},\vec{\boldsymbol \beta}\,\rangle\Big)-\langle\vec{\lambda},\vec{z}\,\rangle-\frac{1}{2}\langle\vec{\lambda}{\boldsymbol \tau},\vec{\lambda}\,\rangle\right)\right]
\Theta\begin{bmatrix}
\vec{\boldsymbol \alpha}\\
\vec{\boldsymbol \beta}\\
\end{bmatrix}(\vec{z}\,|{\boldsymbol \tau}),\ \ \ \vec{\mu},\vec{\lambda}\in\mathbb{Z}^g.
\end{equation*}
\begin{equation*}
 \Theta\begin{bmatrix}
        \vec{\boldsymbol \alpha} + 2 \vec{\mu}\\
\vec{\boldsymbol \beta} + 2\vec{\lambda} \\
       \end{bmatrix}(\vec{z}\,|{\boldsymbol \tau}) =
        \exp\left[i \pi  \langle \vec{\boldsymbol \alpha} , \vec \lambda \rangle\right]\Theta\begin{bmatrix}
        \vec{\boldsymbol \alpha} \\
\vec{\boldsymbol \beta} \\
       \end{bmatrix}(\vec{z}\,|{\boldsymbol \tau}) 
        ,\ \ \ \vec{\mu},\vec{\lambda}\in\mathbb{Z}^g.
\end{equation*}

\ep
For the case of a hyperelliptic Riemann surface $X$
\begin{equation*}
  X=\Big\{(z,w):\, w^2=\prod_{j=1}^{2g+2}(z-a_j)\Big\}
\end{equation*}
with fixed homology basis $\{\mathcal{A}_j,\mathcal{B}_j\}_{j=1}^g$, let $\{\omega_j\}_{j=1}^g$ denote the collection of holomorphic one forms on $X$ with standard normalization
\begin{equation*}
  \oint_{\mathcal{A}_j}\omega_k=\delta_{jk},\ \ \ j,k=1,\ldots,g
\end{equation*}
and $B$-period matrix ${\boldsymbol \tau}$. We denote with $\mathbb{J}_{{\boldsymbol \tau}} = \mathbb{C}^g/(\mathbb{Z}^g+{\boldsymbol \tau}\mathbb{Z}^g)$ the underlying Jacobian variety. If
\begin{equation*}
  \mathfrak{u}(p)=\int_{a_1}^p\vec{\omega},\hspace{0.5cm}\mathfrak{u}:X\rightarrow\mathbb{J}_{{\boldsymbol \tau}}
\end{equation*}
is the Abel map extended to the whole Riemann surface then
\bth[\cite{FK}, p. 308]
\label{generalTheta}
For ${\bf  f}\in \C^g$ arbitrary,
the (multi-valued) function $\Theta(\mathfrak u(z) - {\bf  f}\,|{\boldsymbol \tau})$ on the Riemann surface either vanishes identically or it 
vanishes at $g$ points ${p}_1,\dots, {p}_g$ (counted with multiplicity).
In the latter case we have 
\be
{\bf  f} =\sum_{j=1}^{g} \mathfrak u(p_j) + \mathcal K.
\ee
where the vector of Riemann constants equals
\begin{equation*}
  \mathcal{K}=\sum_{j=1}^g\mathfrak{u}(a_{2j+1}).
\end{equation*}
\et
An immediate consequence of Theorem \ref{generalTheta} is the following statement.
\bc
\label{generalThetadiv}
The function $\Theta({\bf e}\,|{\boldsymbol \tau})$ vanishes at ${\bf  e}\in \mathbb J_{\boldsymbol \tau}$ iff  there exist $g-1$ points $p_1,\dots, p_{g-1}$ on the Riemann surface such that
\be\label{vece}
{\bf e} =\sum_{j=1}^{g-1} \mathfrak u(p_j) + \mathcal K.
\ee
\ec
On a Riemann surface of genus $g$ a divisor is a collection of points (counted with a multiplicity). We are going to consider here 
only positive divisors, namely, with positive multiplicities.
\bd
\label{defspecial}
A (positive) divisor of degree $k\leq g$ is  called special if the vector space of meromorphic functions with poles at the points of order not exceeding
 the given multiplicities has dimension strictly greater than $1$. (Note that the constant function is always in this space).
\ed
As the definition suggests, generic divisors of degree $\leq g$ do not admit other than the constant function in the above-mentioned vector space.
The other fact that we have used is that a divisor $\mathcal D = p_1 + \dots +p_k$ ($k\leq g$)  on the hyperelliptic Riemann 
surface $X$ is special if and only if at least one pair of points are of the 
form $(z, \pm w)$ (i.e. the points are on the two sheets and with the same $z$ value).

\end{appendix}


\end{document}